%% file: disc4.tex
\documentclass[11pt, peerreview, compsocconf]{IEEEtran}

\usepackage[T1]{fontenc}
\usepackage[english]{babel}
\usepackage{lmodern}
\usepackage{url,xspace}
\usepackage[usenames,dvipsnames]{color}
\usepackage{stmaryrd}
\usepackage{paralist}
\usepackage{wrapfig}

\usepackage[pdftex]{graphicx}

\graphicspath{{graph-cam-ready/pdfs/}}

\usepackage{amsmath,amssymb}\usepackage{amsthm}

\usepackage[labelfont=bf, font=small,skip=0pt,justification=centering,compatibility=false]{caption}
\usepackage{subcaption}
\usepackage{fixltx2e}

\usepackage{array}

\usepackage{tikz}
\usetikzlibrary{calc,positioning}
\usetikzlibrary{shapes}
\usetikzlibrary{matrix}
\usetikzlibrary{decorations}
\usetikzlibrary{decorations.pathmorphing}
\usetikzlibrary{decorations.markings}
\usetikzlibrary{backgrounds}

\usepackage[margin=1in]{geometry}

\usepackage[linesnumbered,vlined,noend,ruled,nofillcomment]{algorithm2e}

\usepackage{ifthen}
\usepackage{arrayjobx}

\newcommand{\inte}[2][0]{\ema{\left\llbracket #1,#2 \right\rrbracket}}
\newcommand{\kth}[1]{\ema{#1^{\mathrm{th}}}}
\newcommand{\kst}[1]{\ema{#1^{\mathrm{st}}}}

\newcommand{\pro}[1]{\ema{(\mathcal{P}_{#1}})}

\newtheorem{theorem}{Theorem}

\newtheorem{lemma}{Lemma}
\newtheorem{corollary}{Corollary}
\newtheorem{property}{Property}
\newtheorem{definition}{Definition}
\theoremstyle{definition}
\newtheorem{remark}{Remark}

\newcommand{\ct}{\ema{P}}

\newcommand{\ema}[1]{\ensuremath{#1}\xspace}
\newcommand{\trl}{\ema{\ct_{rl}}}
\newcommand{\scas}{\ema{\mathit{cc}}}
\newcommand{\fcas}{\ema{\mathit{cc}}}
\newcommand{\mem}{\ema{\mathit{rc}}}

\newcommand{\calrl}{\ema{\mathit{cw}}}

\newcommand{\rint}[4]{\ema{\int_{#1}^{#2} #3 \, \mathit{d#4}}}
\newcommand{\expansion}[1]{\ema{e\left(#1\right)}}

\newcommand{\expa}{\ema{e}}

\newcommand{\expansionp}[1]{\ema{e'\left(#1\right)}}
\newcommand{\trlo}{\ema{\trl^{(0)}}}

\newcommand{\exppl}{(\texttt{+})}
\newcommand{\expmi}{(\texttt{-})}

\newcommand{\cw}{\ema{\mathit{cw}}}
\newcommand{\pw}{\ema{\mathit{pw}}}
\newcommand{\rc}{\ema{\mathit{rc}}}
\newcommand{\cc}{\ema{\mathit{cc}}}
\newcommand{\rlw}{\ema{\mathit{rlw}}}
\newcommand{\rlwp}{\ema{\rlw^{\expmi}}}
\newcommand{\psiz}{\ema{\pw^{\exppl}}}
\newcommand{\rlsiz}{\ema{\rlw^{\exppl}}}

\newcommand{\thru}{\ema{T}}

\newcommand{\ctot}{\ema{P}}
\newcommand{\thr}[1]{\ema{\mathcal{T}_{#1}}}

\newcommand{\shiftf}{\ema{\mathit{d}}}
\newcommand{\shift}[1]{\ema{\mathit{d}\left(#1\right)}}

\newcommand{\rl}{retry loop\xspace}
\newcommand{\rls}{retry loops\xspace}

\newcommand{\RLs}{Retry Loops\xspace}

\newcommand{\re}{retry\xspace}
\newcommand{\res}{retries\xspace}
\newcommand{\REs}{Retries\xspace}

\newcommand{\ps}{parallel section\xspace}
\newcommand{\pss}{parallel sections\xspace}
\newcommand{\ds}{data structure\xspace}
\newcommand{\dss}{data structures\xspace}
\newcommand{\casexp}{{\it Compare-And-Swap}\xspace}
\newcommand{\cas}{{\it CAS}\xspace}

\newcommand{\rf}{{\it Read}\xspace}
\newcommand{\faa}{{\it Fetch-and-Increment}\xspace}

\newcommand{\delmin}{\FuncSty{DeleteMin}\xspace}
\newcommand{\enqop}{\FuncSty{Enqueue}\xspace}
\newcommand{\deqop}{\FuncSty{Dequeue}\xspace}
\newcommand{\popop}{\FuncSty{Pop}\xspace}
\newcommand{\pushop}{\FuncSty{Push}\xspace}
\newcommand{\incop}{\FuncSty{Increment}\xspace}
\newcommand{\decop}{\FuncSty{Decrement}\xspace}

\newcommand{\casop}[1]{\FuncSty{CAS\textsubscript{#1}}\xspace}

\newcommand{\caca}{wasted \re}
\newcommand{\cacas}{wasted \res}

\newcommand{\flo}{\ema{f^{\exppl}}}
\newcommand{\fup}{\ema{f^{\expmi}}}

\newcommand{\tlo}{\ema{\thru^{\exppl}}}
\newcommand{\tup}{\ema{\thru^{\expmi}}}

\newcommand{\ghz}[1]{\ema{#1\,\mathrm{GHz}}}
\newcommand{\megb}[1]{\ema{#1\,\mathrm{MB}}}

\newcommand{\ie}{\textit{i.e.}\xspace}
\newcommand{\etal}{\textit{et al.}\xspace}
\newcommand{\eg}{\textit{e.g.}\xspace}
\newcommand{\etc}{\textit{etc.}\xspace}

\newcommand\rr[1]{#1}

\newcommand\pp[1]{}

\newcommand\tra[1]{}

\newcommand\falseparagraph[1]{{\bf #1:}\xspace}

\newcommand\hardcon{hardware conflict\xspace}

\newcommand\hardcons{hardware conflicts\xspace}
\newcommand\Hardcons{Hardware conflicts\xspace}

\newcommand\logcons{logical conflicts\xspace}
\newcommand\Logcons{Logical conflicts\xspace}

\makeatletter
\newcommand{\removelatexerror}{\let\@latex@error\@gobble}
\makeatother

\newcommand{\abstalgo}{
\removelatexerror
\begin{procedure}[H]
\SetKwData{pet}{execution\_time}
\SetKwData{pdo}{done}
\SetKwData{psucc}{success}
\SetKwData{pcur}{current}
\SetKwData{pnew}{new}
\SetKwData{pacp}{AP}
\SetKwData{pttot}{t}
\SetKwFunction{pinit}{Initialization}
\SetKwFunction{ppw}{Parallel\_Work}
\SetKwFunction{pcw}{Critical\_Work}
\SetKwFunction{pread}{Read}
\SetKwFunction{pcas}{CAS}

\SetAlgoLined
\pinit{}\;\nllabel{alg:li-in}
\While{! \pdo}{
\ppw{}\;\nllabel{alg:li-ps}
\While{! \psucc}{\nllabel{alg:li-bcs}
\pcur $\leftarrow$ \pread{\pacp}\;\nllabel{alg:li-bbcs}
\pnew $\leftarrow$ \pcw{\pcur}\;
\psucc $\leftarrow$ \pcas{\pacp, \pcur, \pnew}\;\nllabel{alg:li-ecs}
}
}
\caption{AbstractAlgorithm()\label{alg:gen-name}}
\end{procedure}
}

\newcommand{\expatikz}{
\def\minh{1.5}
\def\maxh{2.5}
\def\midh{2}
\draw (-1,\minh) rectangle (3.5,\maxh) node[pos=.5] {\rf \& \calrl};
\draw (6.5,\minh) rectangle (11,\maxh) node[midway, align=center] {Previously\\expanded CAS};
\draw [dashed, red](3.5,\minh) rectangle (6.5,\maxh) node[pos=.5] {Expansion};
\def\minh{3}
\def\maxh{3.5}
\def\midh{3.25}
\draw (3,\minh) rectangle (6.5,\maxh) node[pos=.5] {CAS};
\draw[->] (3,\minh) -- ++(0,-.5);
\draw[dotted] (6.5,\minh) -- ++(0,-.5);
\end{tikzpicture}
}

\newcommand{\tikzlemone}[1]{
\begin{center}
\begin{tikzpicture} [scale=#1, font=\small, thdec/.style={ thick, decoration={complete sines, segment length=.06cm,amplitude=\amp},decorate},par/.style={blue}, csu/.style={green}, cfa/.style={red}, ini/.style={black}, ]
\draw[draw=none, use as bounding box]  (0,\yt[3]-.25) rectangle (18.9,\yt[1]+.25);
\def\prevend{0}
\newarray\endxt
\expandarrayelementtrue
\readarray{endxt}{0&0&0&0}
\def\cle{1}
\input{auto-plot/ex-lemma-1.tex}
\draw[draw=none,fill=white]  (18.9,\yt[3]-.25) rectangle (21,\yt[1]+.25);
\end{tikzpicture}
}

\newcommand{\tikzlemthree}[1]{
\begin{center}
\begin{tikzpicture} [scale=#1, font=\small, thdec/.style={ thick, decoration={complete sines, segment length=.06cm,amplitude=\amp},decorate},par/.style={blue}, csu/.style={green}, cfa/.style={red}, ini/.style={black}, ]
\draw[draw=none, use as bounding box]  (0,\yt[4]-.25) rectangle (26.6,\yt[1]+.25);
\def\prevend{0}
\newarray\endxt
\expandarrayelementtrue
\readarray{endxt}{0&0&0&0}
\def\cle{1}
\input{auto-plot/ex-lemma-3.tex}
\draw[draw=none,fill=white]  (26.6,\yt[4]-.25) rectangle (29.5,\yt[1]+.25);
\end{tikzpicture}
}
\newcommand{\tikzlemtwo}[1]{
\begin{center}
\begin{tikzpicture} [scale=#1, font=\small, thdec/.style={ thick, decoration={complete sines, segment length=.06cm,amplitude=\amp},decorate},par/.style={blue}, csu/.style={green}, cfa/.style={red}, ini/.style={black}, ]
\draw[draw=none]  (0,\yt[4]-.25) rectangle (27.1,\yt[1]+.25);
\def\prevend{0}
\newarray\endxt
\expandarrayelementtrue
\readarray{endxt}{0&0&0&0}
\def\cle{1}
\input{auto-plot/ex-lemma-4.tex}
\draw[draw=none,fill=white]  (27.1,\yt[4]-.25) rectangle (29.6,\yt[1]+.25);
\end{tikzpicture}
}

\newcommand{\tikztoy}{
\def\cle{1.8}
\def\ple{3.75}
\def\prevend{0}
\newarray\endxt
\expandarrayelementtrue
\readarray{endxt}{0&0&0&0}
\newarray\wpl

\draw[<->] (9.1,\yt[0]) -- ++ (3*\cle+\ple,0)  node[midway,fill=white] {Cycle};
\node [left] at (0,\yt[1]) {\thr{0}};
\node [left] at (0,\yt[2]) {\thr{1}};
\node [left] at (0,\yt[3]) {\thr{2}};
\node [left] at (0,\yt[4]) {\thr{3}};

\extth{1}{\cle}{csu}

\extth{2}{1.05}{ini} \extth{2}{\cle}{cfa}
\extth{2}{\cle}{csu}

\extth{3}{1.2}{ini}
\extth{3}{\cle}{cfa}
\extth{3}{\cle}{cfa}
\extth{3}{\cle}{csu}

\extth{4}{6.75}{ini}
\extth{4}{\cle}{csu}

\extth{1}{\ple}{par}
\extth{1}{\cle}{cfa}
\extth{1}{\cle}{cfa}
\extth{1}{\cle}{csu}

\extth{2}{\ple}{par}
\extth{2}{\cle}{cfa}
\extth{2}{\cle}{cfa}
\extth{2}{\cle}{csu}

\extth{3}{\ple}{par}
\extth{3}{\cle}{cfa}
\extth{3}{\cle}{cfa}
\extth{3}{\cle}{csu}

\extth{4}{\ple}{par}
\extth{4}{\cle}{cfa}
\extth{4}{\cle}{cfa}
\extth{4}{\cle}{csu}

\extth{1}{\ple}{par}
\extth{1}{\cle}{cfa}
\extth{1}{\cle}{cfa}
\extth{1}{\cle}{csu}

\extth{2}{\ple}{par}
\extth{2}{\cle}{cfa}
\extth{2}{\cle}{cfa}
\extth{2}{\cle}{csu}

\extth{3}{\ple}{par}
\extth{3}{\cle}{cfa}
\extth{3}{\cle}{cfa}
\extth{3}{\cle}{csu}

\extth{4}{\ple}{par}
\extth{4}{\cle}{cfa}
\extth{4}{\cle}{cfa}
\extth{4}{\cle}{csu}

\extth{1}{\ple}{par}
\extth{1}{\cle}{cfa}

\extth{2}{\ple}{par}

\extth{3}{2}{par}

\draw[draw=none,fill=white]  (25.45,\yt[0]) rectangle (28,\yt[1]+.25);

\end{tikzpicture}
\caption{Execution with one \caca, and one inevitable failure\label{fig:ex-f2}}
}

\newcommand\posrem[2]{#2}

\title{Analyzing the Performance of Lock-Free \\Data Structures: A Conflict-based Model}

\rr{
\author{Aras Atalar, Paul Renaud-Goud and Philippas Tsigas\\
 Chalmers University of Technology\\
\url{{aaras|goud|tsigas}@chalmers.se}\\}}

\pp{\author{Aras Atalar \and Paul Renaud-Goud \and Philippas Tsigas}
\institute{Chalmers University of Technology\\
S-412 96 G\"oteborg, Sweden\\
\email{\{aaras|goud|tsigas\}@chalmers.se}}}

\begin{document}

\rr{\setlength{\abovecaptionskip}{10pt}}

\SetFuncSty{textsf}

\maketitle

\begin{abstract}

\pp{This paper considers the modeling and the analysis of the performance
of lock-free concurrent data structures that can be represented as
linear combinations of fixed size \rls.

Our main contribution is a new way of modeling and analyzing a general
class of lock-free algorithms, achieving predictions of throughput
that are close to what we observe in practice.  We emphasize two kinds
of conflicts that shape the performance: 
(i) \hardcons, due to concurrent calls to atomic primitives;
(ii) \logcons, caused by simultaneous operations on the shared \ds.

We propose also a common framework that enables a fair comparison
between lock-free implementations by covering the whole contention
domain, and comes with a method for calculating a good back-off
strategy.

Our experimental results, based on a set of widely used concurrent
data structures and on abstract lock-free designs, show that our
analysis follows closely the actual code behavior.}

\rr{This paper considers the modeling and the analysis of the performance
of lock-free concurrent data structures.  Lock-free designs employ an
optimistic conflict control mechanism, allowing several processes to
access the shared data object at the same time.  They guarantee that
at least one concurrent operation finishes in a finite number of its
own steps regardless of the state of the operations.  Our analysis
considers such lock-free data structures that can be represented as
linear combinations of fixed size \rls.

Our main contribution is a new way of modeling and analyzing a general
class of lock-free algorithms, achieving predictions of throughput
that are close to what we observe in practice. We emphasize
two kinds of conflicts that shape
the performance: (i) \hardcons, due to concurrent calls to atomic primitives;
(ii) \logcons, caused by simultaneous operations on the shared \ds.

We show how to deal with these hardware and logical conflicts separately,
and how to combine them, so as to calculate the throughput of lock-free
algorithms.
We propose also a common framework that enables a fair
comparison between lock-free implementations by covering the whole contention
domain, together with a better understanding of the performance impacting
factors.
This part of our analysis comes with
a method for calculating a good back-off strategy to
finely tune the performance of a lock-free algorithm. Our experimental
results, based on a set of widely used concurrent data structures and
on abstract lock-free designs, show that our analysis follows closely
the actual code behavior.
}

\end{abstract}

\rr{\newpage\tableofcontents\newpage}

\section{Introduction}
\label{sec:intro}

Lock-free programming provides highly concurrent access to data and
has been increasing its footprint in industrial settings.  Providing a
modeling and an analysis framework capable of describing the practical
performance of lock-free algorithms is an essential, missing resource
necessary to the parallel programming and algorithmic research
communities in their effort to build on previous intellectual efforts.
The definition of lock-freedom mainly guarantees that at least one
concurrent operation on the \ds finishes in a finite number of its own
steps, regardless of the state of the operations. On the
individual operation level, lock-freedom cannot guarantee that an
operation will not starve.
\posrem{The analysis frameworks that currently
exist in the literature focus on such worst-case behavior and are far
from capturing the behavior observed in practice.}{}

The goal of this paper is to provide a way to model and analyze the
practically observed performance of lock-free \dss.
In the literature, the common performance measure of a lock-free \ds is
the throughput, \ie the number of successful operations per unit of time.
It is obtained while threads are accessing the \ds according to an
access pattern that interleaves local work between calls to consecutive 
operations on the \ds.
Although this access pattern to the data structure is significant,
there is no consensus in the literature on what access to be used when
comparing two data structures.
So, the amount of local work 
(that we will refer as parallel work for the rest of the paper)
could be
constant (\cite{lf-queue-michael,pq-skiplist-constant}), uniformly 
distributed (\cite{scalable-stack-uniform}, \cite{count-moir}), exponentially
distributed (\cite{Val94}, \cite{pq-survey-exponential}), null 
(\cite{future-ds-nonexist,upp-prio-que}), \etc, and more
questionably, the average amount is rarely scanned, which leads to a
partial covering of the contention domain.

We propose here a common framework enabling a fair comparison between
lock-free \dss, while exhibiting the main phenomena that drive
performance, and particularly the contention, which leads to different
kinds of conflicts. As this
is the first step in this
direction, we want to deeply analyze the core of the problem, without
impacting factors being diluted within a probabilistic
smoothing. Therefore, we choose a constant local work, hence constant access rate to the \dss.
In addition to the prediction of the \ds performance, our model
provides a good back-off strategy, that achieves the peak performance
of a lock-free algorithm.

\hyphenation{Abstract-Algorithm}

Two kinds of conflict appear during the execution of a lock-free
algorithm, both of them leading to additional work. \Hardcons occur
when concurrent operations call atomic primitives on the same data:
these calls collide and conduct to stall time, that we name here {\it
  expansion}. \Logcons take place if concurrent operations overlap:
because of the lock-free nature of the algorithm, several concurrent
operations can run simultaneously, but only one \re can logically
succeed. We show that the additional work produced by the failures is
not necessarily harmful for the system-wise performance.

We then show how throughput
can be computed by connecting these two key factors in an iterative
way. We start by estimating the expansion probabilistically, and
emulate the effect of stall time introduced by the \hardcons as extra
work added to each thread. Then
we estimate the number of failed operations,
that in turn lead to additional extra work,
by computing again the expansion on a system setting where those two
new amounts of work have been incorporated, and reiterate the
process; the convergence is ensured by a fixed-point search.

We consider the class of lock-free algorithms that can be modeled as a
linear composition of fixed size \rls.  This class covers numerous
extensively used lock-free designs such as stacks~\cite{lf-stack}
(\popop, \pushop), queues~\cite{lf-queue-michael} (\enqop, \deqop),
counters~\cite{count-moir} (\incop, \decop) and priority
queues~\cite{upp-prio-que} (\delmin).

To evaluate the accuracy of our model and analysis framework, we
performed experiments both on synthetic tests, that capture
a wide range of possible abstract algorithmic designs,
and on 
several reference implementations of
extensively studied lock-free \dss.
Our evaluation results reveal that our model is able to capture the
behavior of all the synthetic and real designs for all different
numbers of threads and sizes of parallel work (consequently also
contention). \posrem{Our model follows the performance behavior of the \dss
exactly in low contention, when our lower and upper bounds meet in one
line with the observed behavior; and follows closely also the
performance in high contention. }{}We also evaluate the use of our
analysis as a tool for tuning the performance of lock-free code by
selecting the appropriate back-off strategy that will maximize
throughput by comparing our method with against widely known back-off
policies, namely linear and exponential.

The rest of the paper is organized as follows.  We discuss related
work in Section~\ref{sec:rel}, then the problem is formally described
in Section~\ref{sec:ps}. We consider the \logcons in the absence of
\hardcons
in Section~\ref{sec:wt}, while in Section~\ref{sec:exp-glue},
we firstly show how to compute the expansion, then combine 
hardware and \logcons
to obtain the final throughput estimate.  We
describe the experimental results in Section~\ref{sec:xp}.

\section{Related Work}
\label{sec:rel}

Anderson \etal~\cite{anders-rt} evaluated the performance of lock-free objects
in a single processor real-time system by emphasizing the impact of \rl
interference. Tasks can be preempted during the \rl execution, which can lead to
interference, and consequently to an inflation in \rl execution due to \res. They obtained
upper bounds for the number of interferences under various scheduling schemes
for periodic real-time tasks.

Intel~\cite{intel-emp} conducted an empirical study to illustrate performance and
scalability of locks. They showed that the critical section size, the time interval
between releasing and re-acquiring the lock (that is similar to our parallel section size)
and number of threads contending the lock are vital parameters.

Failed \res do not only lead to useless effort but also degrade the performance of
successful ones by contending the shared resources. Alemany \etal~\cite{alemany-os} have
pointed out this fact, that is in accordance with our two key factors, and, without trying
to model it, have mitigated those effects by designing non-blocking algorithms with
operating system support.

Alistarh \etal~\cite{ali-same} have studied the same class of lock-free structures that
we consider in this paper. The analysis is done in terms of scheduler steps, in a 
system where only one thread can be scheduled (and can then 
run) at each step. If compared with execution time, this is 
particularly appropriate to a system with a single processor and 
several threads, or to a system where the instructions of the threads 
cannot be done in parallel (\eg multi-threaded program on a 
multi-core processor with only read and write on the same cache line 
of the shared memory). In our paper, the execution is evaluated 
in terms of processor cycles, strongly related to the execution 
time. In addition, the ``parallel work'' and the ``critical work'' can be 
done in parallel, and we only consider retry-loops with one Read and 
one CAS, which are serialized. 
In addition, they bound the asymptotic expected system latency (with a 
big O, when the number of threads tends to infinity), while in our 
paper we estimate the throughput (close to the inverse of system 
latency) for any number of threads.

\section{Problem Statement}
\label{sec:ps}

\subsection{Running Program and Targeted Platform}
\label{sec:prog-plat}

\rr{\begin{figure}[t!]
\abstalgo
\caption{Thread procedure}\label{alg:gen-nb}
\end{figure}
}

\pp{
\setlength{\textfloatsep}{0\baselineskip plus 0.2\baselineskip minus 0.2\baselineskip}
\begin{figure}
\begin{minipage}[b]{.5\textwidth}
\begin{scriptsize}
\abstalgo
\end{scriptsize}
\caption{Thread procedure}\label{alg:gen-nb}
\end{minipage}\hfill\begin{minipage}[b]{.45\textwidth}
\begin{tikzpicture} [scale=0.5, font=\scriptsize]
\expatikz
\caption{Expansion}\label{fig:expansion}
\end{minipage}
\end{figure}
}

In this paper, we aim at evaluating the throughput of a multi-threaded algorithm that is based on the utilization of a shared lock-free \ds. Such a
program can be abstracted by the Procedure~\ref{alg:gen-name} (see
Figure~\ref{alg:gen-nb}) that represents the skeleton of the function which is
called by each spawned thread. It is decomposed in two main phases: the {\it \ps},
represented on line~\ref{alg:li-ps}, and the {\it \rl}, from line~\ref{alg:li-bcs} to
line~\ref{alg:li-ecs}. A {\it \re} starts at line~\ref{alg:li-bbcs} and ends at
line~\ref{alg:li-ecs}.

As for line~\ref{alg:li-in}, the function \pinit
shall be seen as an abstraction of the delay between the spawns of the threads,
that is expected not to be null, even when a barrier is used. We then consider that the
threads begin at the exact same time, but have different initialization times.

The \ps is the part of the code where the thread does not access the shared
\ds; the work that is performed inside this \ps can possibly depend on the value that has been read from the \ds, \eg in the case of processing an element
that has been dequeued from a FIFO (First-In-First-Out) queue.

In each \re, a thread tries to modify the \ds, and does not exit the \rl until it
has successfully modified the \ds. It does that by firstly reading the access point \DataSty{AP} of the
\ds, then according to the value that has been read, and possibly to other previous
computations that occurred in the past, the thread prepares the new desired value as an
access point of the \ds. Finally, it atomically tries to perform the change through a call
to the \casexp (\cas) primitive. If it succeeds, \ie if the access point has not been changed by
another thread between the first \rf and the \cas, then it goes to the next \ps, otherwise
it repeats the process. The \rl is composed of at least one \re, and we number the \res
starting from $0$, since the first iteration of the \rl is actually not a \re, but
a try.

\pgfdeclaredecoration{complete sines}{initial}
{
    \state{initial}[
        width=+0pt,
        next state=sine,
        persistent precomputation={\pgfmathsetmacro\matchinglength{
            \pgfdecoratedinputsegmentlength / int(\pgfdecoratedinputsegmentlength/\pgfdecorationsegmentlength)}
            \setlength{\pgfdecorationsegmentlength}{\matchinglength pt}
        }] {}
    \state{sine}[width=\pgfdecorationsegmentlength]{
        \pgfpathsine{\pgfpoint{0.25\pgfdecorationsegmentlength}{0.5\pgfdecorationsegmentamplitude}}
        \pgfpathcosine{\pgfpoint{0.25\pgfdecorationsegmentlength}{-0.5\pgfdecorationsegmentamplitude}}
        \pgfpathsine{\pgfpoint{0.25\pgfdecorationsegmentlength}{-0.5\pgfdecorationsegmentamplitude}}
        \pgfpathcosine{\pgfpoint{0.25\pgfdecorationsegmentlength}{0.5\pgfdecorationsegmentamplitude}}
}
    \state{final}{}
}
\pgfdeclarelayer{background}
\pgfdeclarelayer{foreground}
\pgfsetlayers{background,main,foreground}

\def\amp{5}

\newcounter{wit}
\setcounter{wit}{1}

\newcommand{\extth}[3]{\checkendxt(#1)
\edef\prev{\cachedata}
\pgfmathparse{\prev+#2}
\endxt(#1)={\pgfmathresult}
\checkendxt(#1)
\edef\nene{\cachedata}
\draw[-,very thick,#3] (\prev,\yt[#1]) -- (\nene,\yt[#1]);
\begin{pgfonlayer}{foreground}
\draw[orange] (\prev,\yt[#1]-.25) -- ++(0,.5);
\end{pgfonlayer}
}

\rr{\def\yt{{-5.5,0,-1.5,-3,-4.5}}}
\pp{\def\yt{{-3,0,-.8,-1.6,-2.4}}}

\rr{\begin{figure}
\begin{tikzpicture} [scale=0.6, font=\small, thdec/.style={ thick, decoration={complete sines, segment length=.1cm,amplitude=\amp},decorate},
par/.style={blue}, csu/.style={green}, cfa/.style={red}, ini/.style={black}, ]
\tikztoy
\end{figure}}

\setcounter{wit}{1}

\rr{\begin{figure}\begin{center}
\hspace*{1cm}\begin{tikzpicture} [scale=0.6, font=\small, thdec/.style={thick, decoration={complete sines, segment length=.1cm,amplitude=\amp},decorate}, par/.style={blue}, csu/.style={green}, cfa/.style={red}, ini/.style={black}, ]
\def\cle{2.4}
\def\ple{5}
\def\prevend{0}
\newarray\endxt
\expandarrayelementtrue
\readarray{endxt}{0&0&0&0}
\newarray\wpl

\draw[<->] (.6+3*\cle,\yt[0]) -- ++ (2*\cle+\ple,0)  node[midway,fill=white] {Cycle};
\node [left] at (0,\yt[1]) {\thr{0}};
\node [left] at (0,\yt[2]) {\thr{1}};
\node [left] at (0,\yt[3]) {\thr{2}};
\node [left] at (0,\yt[4]) {\thr{3}};

\extth{1}{\cle}{csu}

\extth{2}{.2}{ini}
\extth{2}{\cle}{cfa}
\extth{2}{\cle}{csu}

\extth{3}{.4}{ini}
\extth{3}{\cle}{cfa}
\extth{3}{\cle}{cfa}
\extth{3}{\cle}{csu}

\extth{4}{.6}{ini}
\extth{4}{\cle}{cfa}
\extth{4}{\cle}{cfa}
\extth{4}{\cle}{cfa}
\extth{4}{\cle}{csu}

\extth{1}{\ple}{par}
\extth{1}{\cle}{cfa}
\extth{1}{\cle}{csu}

\extth{2}{\ple}{par}
\extth{2}{\cle}{cfa}
\extth{2}{\cle}{csu}

\extth{3}{\ple}{par}
\extth{3}{\cle}{cfa}
\extth{3}{\cle}{csu}

\extth{4}{\ple}{par}
\extth{4}{\cle}{cfa}
\extth{4}{\cle}{csu}

\extth{1}{\ple}{par}
\extth{1}{\cle}{cfa}
\extth{1}{\cle}{csu}

\extth{2}{\ple}{par}
\extth{2}{\cle}{cfa}

\extth{3}{\ple}{par}

\extth{4}{\ple}{par}

\draw[draw=none,fill=white]  (22,\yt[4]-.25) rectangle (25.5,\yt[1]+.25);

\end{tikzpicture}
\caption{Execution with minimum number of failures\label{fig:ex-f1}}
\end{center}\end{figure}}

\rr{
We analyze the behavior of \ref{alg:gen-name} from a
throughput
perspective, which 
is defined as the number of successful \ds operations per unit of
time. In the context of Procedure~\ref{alg:gen-name}, it is equivalent to the
number of successful \cas{}s.

\medskip}

The throughput of the lock-free algorithm\pp{, \ie the number of successful \ds operations per unit of
time}, that we denote by \thru, is impacted by several
parameters.
\begin{compactitem}
\item {\it Algorithm parameters}: the amount of work inside a call to
  \FuncSty{Parallel\_Work} (resp. \FuncSty{Critical\_Work}) denoted by \pw
  (resp. \cw).

\item {\it Platform parameters}: \rf and \cas latencies (\rc and \cc
  respectively), and the number \ct of processing units (cores). We assume
  homogeneity for the latencies, \ie every thread experiences the same
  latency when accessing an uncontended shared data, which is achieved
  in practice by pinning threads to the same
  socket.

\end{compactitem}

\medskip

\subsection{Examples and Issues}

We first present two straightforward upper bounds on the throughput, and describe the two
kinds of conflict
that keep the actual throughput away from those upper bounds.

\subsubsection{Immediate Upper Bounds}
\label{sec:imm-bou}

Trivially, the minimum amount of work \rlwp in a given \re is $\rlwp = \rc+\cw+\cc$, as we
should pay at least the memory accesses
and the critical work \cw in between.

\falseparagraph{Thread-wise} A given thread can at most perform one
successful \re every $\pw+\rlwp$ units of time. In the best case,
\ctot threads can then lead to a throughput of $\ctot/(\pw+\rlwp)$.

\falseparagraph{System-wise} By definition, two successful \res cannot
overlap, hence we have at most $1$ successful \re every \rlwp units of
time.

\medskip

\rr{
Altogether, the throughput \thru is bounded by
\[ \thru \leq \min \left( \frac{1}{\rc+\cw+\cc} , \frac{\ctot}{\pw+\rc+\cw+\cc}\right), 
\text{\ie} \]
\hspace*{-1cm}\begin{equation}
\label{eq:imm-bou}
 \thru \leq
\left\{ \begin{array}{ll}
\frac{1}{\rc+\cw+\cc} & \rr{\quad} \text{ if } \pw \leq (\ctot-1)(\rc+\cw+\cc)\\
\frac{\ctot}{\pw+\rc+\cw+\cc} & \rr{\quad} \text{ otherwise.}
\end{array} \right.
\end{equation}}
\pp{
Altogether, the throughput \thru is upper bounded by the minimum
of $1/(\rc+\cw+\cc)$ and $\ctot/(\pw+\rc+\cw+\cc)$, 
\ie
\begin{equation}
\label{eq:imm-bou}
 \thru \leq
\left\{ \begin{array}{ll}
\frac{1}{\rc+\cw+\cc} & \rr{\quad} \text{ if } \pw \leq (\ctot-1)(\rc+\cw+\cc)\\
\frac{\ctot}{\pw+\rc+\cw+\cc} & \rr{\quad} \text{ otherwise.}
\end{array} \right.
\end{equation}}

\subsubsection{Conflicts}

\paragraph{\Logcons}

Equation~\ref{eq:imm-bou} expresses the fact that when \pw is small enough, \ie when $\pw
\leq (\ctot-1) \rlwp$, we cannot expect that every thread performs a successful \re every
$\pw+\rlwp$ units of time, since it is more than what the \rl can afford.
As a result, some \logcons, hence unsuccessful \res, will be inevitable, while the others,
if any, are called {\it wasted}.

\pp{
\setlength{\textfloatsep}{0.6\baselineskip plus 0.2\baselineskip minus 0.2\baselineskip}
\begin{figure}[b!]
\begin{tikzpicture} [scale=0.45, font=\small, thdec/.style={ thick, decoration={complete sines, segment length=.1cm,amplitude=\amp},decorate},
par/.style={blue}, csu/.style={green}, cfa/.style={red}, ini/.style={black}, ]
\tikztoy
\end{figure}}

\rr{However, different executions can lead to different numbers of failures, which end up with
different throughput values. Figures~\ref{fig:ex-f2} and~\ref{fig:ex-f1} depict two
executions,}
\pp{Figure~\ref{fig:ex-f2} depicts an execution,} where the black parts are the calls to \pinit, the blue parts are the \pss,
and the \res can be either unsuccessful --- in red --- or successful --- in green.\rr{ We
experiment different initialization times, and observe different synchronizations, hence
different numbers of \cacas.} After the initial transient state, \rr{the execution depicted in
Figure~\ref{fig:ex-f1} comprises only the inevitable unsuccessful \res, while the
execution of Figure~\ref{fig:ex-f2} contains one \caca.}\pp{the execution
contains actually, for each thread, one inevitable unsuccessful \re, and one \caca, because there exists
a set of initialization times that lead to a cyclic execution with a single failure per thread and per period.}

We can see on \rr{those two examples}\pp{this example} that a cyclic execution
is reached after the transient behavior; actually, we show in
Section~\ref{sec:wt} that, in the absence of \hardcons, every execution will become periodic, if the
initialization times are spaced enough. In addition, we prove that the shortest
period is such that, during this period, every thread succeeds exactly
once. This finally leads us to define the additional failures as wasted, since
we can directly link the throughput with this number of \cacas: a higher number
of \cacas implying a lower throughput.

\paragraph{\Hardcons}

\rr{\begin{figure}[h!]\begin{center}\begin{tikzpicture}
\expatikz
\caption{Expansion\label{fig:expansion}}
\end{center}
\end{figure}
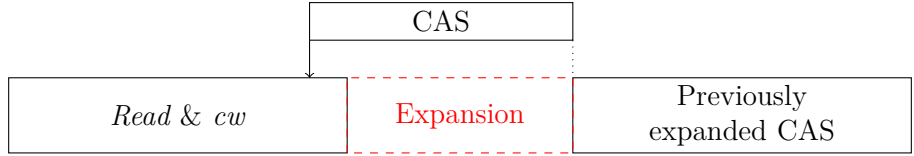
}

The requirement of atomicity compels the ownership of the data in an exclusive manner by the
executing core.\rr{ This fact prohibits concurrent execution of atomic instructions if they
are operating on the same data.} Therefore, overlapping parts of atomic instructions
are serialized by the hardware, leading to stalls in subsequently issued ones. For our target
lock-free algorithm, these stalls that we refer to as expansion become an important slowdown factor
in case threads interfere in the \rl. As illustrated in Figure~\ref{fig:expansion},
the latency for \cas can expand and cause remarkable decreases in throughput since the
\cas of a successful thread is then expanded by others; for this reason, the amount of work inside
a \re is not constant, but is, generally speaking, a function depending on the number of threads
that are inside the \rl.

\subsubsection{Process}

We deal with the two kinds of conflicts separately
and connect them together through the fixed-point iterative convergence.

In Section~\ref{sec:exp}, we compute the expansion in execution time of a \re, noted
\expa, by following a probabilistic approach. The estimation takes as input the expected
number of threads inside the \rl at any time, and returns the expected increase in the
execution time of a \re due to the serialization of atomic primitives.

In Section~\ref{sec:wt}, we are given program without \hardcon 
described by the size of the \ps \psiz and the
size of a \re \rlsiz. We compute upper and lower bounds on the throughput \thru, the
number of wasted \res $w$, and the average number of threads inside the \rl \trl. Without
loss of generality, we can normalize those execution times by the execution time of a \re,
and define the \ps size as $\psiz = q + r$, where $q$ is a non-negative integer and $r$ is
such that $0 \leq r < 1$. This pair (together with the number of threads \ct) constitutes
the actual input of the estimation.

Finally, we combine those two outcomes in Section~\ref{sec:thr}
by emulating expansion through work not prone to \hardcons
 and obtain the full
estimation of the throughput\posrem{ according to the model parameters that have been described in
Section~\ref{sec:prog-plat}}{}.

\section{Execution without \hardcon}
\label{sec:wt}

\newcommand{\suc}[1]{\ema{S_{#1}}}
\newcommand{\att}[2]{\ema{R_{#1}^{#2}}}

\newcommand{\fis}[1]{\suc{#1}}

\newcommand{\gap}[2]{\ema{G_{#1}^{(#2)}}}
\newcommand{\gapt}[2]{\ema{\widetilde{G}_{#1}^{(#2)}}}
\newcommand{\tinit}[1]{\ema{t_{#1}}}

\newcommand{\wfs}{well-formed seed\xspace}
\newcommand{\wfss}{well-formed seeds\xspace}

\newcommand{\wifs}{weakly-formed seed\xspace}
\newcommand{\cyex}[2]{$(#1,#2)$-cyclic execution\xspace}
\newcommand{\cyexs}[2]{$(#1,#2)$-cyclic executions\xspace}

\newcommand{\fa}[1]{\ema{f\left( #1\right)}}
\newcommand{\soe}{\ema{\mathcal{S}}}
\newcommand{\wt}{\ema{\mathit{WT}}}
\newcommand{\wati}{\textit{wasted time}\xspace}

\newcommand{\pl}{\ema{n_0}}
\newcommand{\fl}{\ema{f_0}}

We show in this section that, in the absence of \hardcons, the execution becomes periodic, which eases the calculation
of the throughput.  We start by defining some useful concepts: \cyexs{f}{\ct} are special kind of
periodic executions such that within the shortest period, each thread
performs exactly $f$ unsuccessful \res and $1$ successful \re. The {\it \wfs} is a set of events
that allows us to detect an {\it \cyex{f}{\ct}} early, and the {\it gaps} are a measure of the quality
of the synchronization between threads. The idea is to iteratively add threads into the
game and show that the periodicity is maintained.
Theorem~\ref{th.wfs.gap} establishes a fundamental relation between gaps and \wfss, while
Theorem~\ref{th.cyc.exec} proves the periodicity, relying on the disjoint cases
\pp{depicted on Figures~\ref{fig:ex-lem.1},~\ref{fig:ex-lem.3} and~\ref{fig:ex-lem.2}.
We recall that the complete version of the proofs can be found in~\cite{our-long}, together
with additional Lemmas.}
\rr{of Lemma~\ref{lem.cyc.exec.1},~\ref{lem.cyc.exec.3}, and~\ref{lem.cyc.exec.2}.}
Finally, we exhibit upper and lower bounds on throughput and number of failures,
along with the average number of threads inside the \rl.

\subsection{Setting}
\label{sec:wt-set}

\rr{

\subsubsection{Initial Restrictions}
\begin{remark}
Concerning correctness, we assume that the reference point of the \rf and the \cas occurs
when the thread enters and exits any \re, respectively.
\end{remark}

\begin{remark}
We do not consider simultaneous events, so all inequalities that refer to time
comparison are strict, and can be viewed as follows: time instants are real
numbers, and can be equal, but every event is associated with a thread; also, in
order to obtain a strict order relation, we break ties according to the thread
numbers (for instance with the relation $<$).
\end{remark}
}

\pp{ In preamble, note that the events are strictly ordered (according to their instant of occurrence, with the
  thread id as a tie-breaker). As for correctness, \ie to decide for the success or
  the failure of a \re, we need instants of occurrence for \rf and \cas; we consider that
  the entrance (resp. exit) time of a \re is the instant of occurrence of the \rf
  (resp. \cas).}

\subsubsection{Notations and Definitions}

We recall that \ct threads are executing the pseudo-code described in
Procedure~\ref{alg:gen-name}, one \re is of unit-size, and the \ps is of size $\psiz = q
+r$, where $q$ is a non-negative integer and $r$ is such that  $0 \leq r < 1$.
Considering a thread \thr{n} which succeeds at time \suc{n}; this thread completes a
whole \re in $1$ unit of time, then executes the \ps of size \psiz, and attempts to perform
again the operation every
unit of time, until one of the attempt is successful.

\begin{definition}
An execution with \ct threads is called \cyex{C}{\ct} if
and only if (i) the execution is periodic, \ie at every time, every thread is in
the same state as one period before, (ii) the shortest period contains exactly one
successful attempt per thread, (iii) the shortest period is $1+q+r+C$.
\end{definition}

\begin{definition}
Let $\soe = \left( \thr{i}, \fis{i} \right)_{i \in \inte{\ct-1}}$,
where \thr{i} are threads and \fis{i} ordered times, \ie such that
\rr{$\fis{0} < \fis{1} < \dots < \fis{\ct-1}$.}
\pp{$\fis{0} <  \dots < \fis{\ct-1}$.} \soe is a {\it seed} if and
only if for all $i \in \inte{\ct -1}$, \thr{i} does not succeed
between \fis{0} and \fis{i}, and starts a \re at \fis{i}.

We define \fa{\soe} as the smallest non-negative integer such that
$\fis{0} + 1 + q + r + \fa{\soe} > \fis{\ct-1} + 1$, \ie $ \fa{\soe} =
\max\left( 0 , \left\lceil \fis{\ct-1} - \fis{0} - q
-r\right\rceil\right)$. When \soe is clear from the context, we denote
\fa{\soe} by $f$.
\end{definition}

\begin{definition}

\soe is a {\it \wfs} if and only if for each $i \in \inte{\ct-1}$, the execution of thread
\thr{i} contains the following sequence: firstly a success beginning at \fis{i}, the
parallel section,  $f$ unsuccessful \res, and finally a successful \re.
\end{definition}

Those definitions are coupled through the two natural following properties:
\pp{
}\begin{property}
\label{pr.wfs.soe1}
Given a \cyex{C}{\ct}, any seed \soe including \ct consecutive successes is a \wfs,
with $\fa{\soe} = C$.
\end{property}
\rr{
\begin{proof}
Choosing any set of \ct consecutive successes, we are ensured, by the
definition of a \cyex{f}{\ct}, that for each thread, after the first success,
the next success will be obtained after $f$ failures. The order will be
preserved, and this shows that a seed including our set of successes is 
actually a \wfs.
\end{proof}
}\pp{
}\begin{property}
\label{pr.wfs.soe2}
If there exists a \wfs in an execution, then after each thread succeeded once, the
execution coincides with an \cyex{f}{\ct}.
\end{property}

\rr{
\begin{proof}
By the definition of a \wfs, we know that the threads will first
succeed in order, fails $f$ times, and succeed again in the same
order. Considering the second set of successes in a new \wfs, we
observe that the threads will succeed a third time in the same order,
after failing $f$ times. By induction, the execution coincides with an
\cyex{f}{\ct}.
\end{proof}
}

\newcommand{\lag}{lagging time\xspace}
\newcommand{\lagt}{\ema{\ell}}

Together with the seed concept, we define the notion of \textit{gap}\rr{ that we will use
extensively in the next subsection}. The general idea of those gaps is that within an
\cyex{f}{\ct}, the period is higher than $\ct \times 1$, which is the total execution time
of all the successful \res within the period. The difference between the period (that
lasts $1+q+r+f$) and \ct, reduced by $r$ (so that we obtain an integer), is referred as \textit{\lag} in the following. If
the threads are numbered according to their order of success (modulo \ct), as the time elapsed between
the successes of two given consecutive threads is constant (during the next period, this time
will remain the same), this \lag can be seen in a circular manner\rr{ (see
Figure~\ref{fig:gaps})}: the threads are
represented on a circle whose length is the \lag increased by $r$,
and the length between two consecutive threads is the time between
the end of the successful \re of the first thread and the begin of the successful \re of the
second one. More formally, for all $(n,k) \in \inte{\ct-1}^2$, we define the gap
\gap{n}{k} between \thr{n} and its \kth{k} predecessor based on the gap with the first
predecessor:
\rr{\begin{figure}
\begin{center}
\begin{tikzpicture}[arrows = {<->}]\def\cim{2}
\def\cii{1.4}
\def\cio{2.5}
\draw[thick] (0,0) circle (\cim);
\draw[red]
 (90:\cii) arc (91:449:\cii);
\node[red] at (-60:\cii-.8) {$\displaystyle \sum_{n=0}^{\ct-1} \gap{n}{1}$};
\draw[fill=white,thick] (90:\cim) circle (.25) node {\thr{0}};
\draw[fill=white,thick] (45:\cim) circle (.25) node {\thr{1}};
\draw[fill=white,thick] (-45:\cim) circle (.25) node {\thr{2}};
\draw[fill=white,thick] (180:\cim) ellipse (.4 and .3) node {\thr{\ctot-1}};
\draw[blue] (89.5:\cio) arc (89.5:45.5:\cio);
\node[blue] at (66:1.2*\cio) {\gap{1}{1}};
\draw[blue] (44.5:\cio) arc (44.5:-44.5:\cio);
\node[blue] at (0:1.2*\cio) {\gap{2}{1}};
\draw[blue] (90.5:\cio) arc (90.5:314.5:\cio);
\node[blue] at (200:1.2*\cio) {\gap{0}{2}};
\end{tikzpicture}
\end{center}
\caption{Gaps\label{fig:gaps}}
\end{figure}}\pp{
}\[
\left\{ \begin{array}{l}
\forall n \in \inte[1]{\ct-1} \quad ; \quad \gap{n}{1} = \suc{n} - \suc{n-1} - 1\\
\gap{0}{1} = \suc{0} + q + r + f - \suc{\ct-1}
\end{array} \right.,
\]
which leads to the definition of higher order gaps:
\rr{\[ \forall n \in \inte{\ct-1} \rr{\quad} ; \rr{\quad} \forall k > 0 \rr{\quad} ; \rr{\quad} \gap{n}{k} = \sum_{j=n-k+1}^n \gap{j \bmod \ct}{1}.\]}
\pp{$\forall n \in \inte{\ct-1} \rr{\quad} ; \rr{\quad} \forall k > 0 \rr{\quad} ; \rr{\quad} \gap{n}{k} = \sum_{j=n-k+1}^n \gap{j \bmod \ct}{1}$.}

For consistency, for all $n \in \inte{\ct-1}$, $\gap{n}{0} = 0$.

Equally, the gaps can be obtained\rr{ directly} from the successes: for all $k \in \inte[1]{\ct-1}$,
\begin{equation}
\gap{n}{k} = \left\{ \begin{array}{ll}
\suc{n} - \suc{n-k} - k & \mathrm{~if~} n>k\\
\suc{n} - \suc{\ct + n-k} +1+q+r+f-k & \mathrm{ otherwise}\\
\end{array} \right.
\label{eq.gats}
\end{equation}

Note that, in an \cyex{f}{\ct}, the \lag is the sum of all first order gaps, reduced by $r$.

\rr{

Now we extend the concept of \wfs to \wifs.

\begin{definition}
Let $\soe = \left( \thr{i}, \fis{i} \right)_{i \in \inte{\ct-1}}$ be a seed.

\soe is a \wifs for \ct threads if and only if: $\left( \thr{i}, \fis{i} \right)_{i \in \inte{\ct-2}}$
is a \wfs for $\ct-1$ threads, and the first thread succeeding after
\thr{\ct-2} is \thr{\ct-1}.
\end{definition}

\begin{property}
Let $\soe = \left( \thr{i}, \fis{i} \right)_{i \in \inte{\ct-1}}$ be a \wifs.

Denoting $f = \fa{\left( \thr{i}, \fis{i} \right)_{i \in \inte{\ct-2}}}$, for each $n \in
\inte{\ct-1}$, $\gap{n}{f} < 1$.
\end{property}
\rr{\begin{proof}
We have $\suc{\ct-2} + 1 < \suc{\ct-1} < \att{0}{f}$,
and if we note indeed \gapt{n}{k} the gaps within $\left( \thr{i}, \fis{i} \right)_{i \in \inte{\ct-2}}$,
the previous \wfs with $\ct-1$ threads, we know that for
all $n \in \inte[1]{\ct-2}$, $\gapt{n}{1} = \gap{n}{1}$, and $\gap{\ct-1}{1} +
\gap{0}{1} = \gapt{0}{1}$, which leads to $\gap{n}{k} \leq \gapt{n}{k}$, for all
$n \in \inte{\ct-1}$ and $k$; hence the weaker property.
\end{proof}}

\rr{\begin{figure}
\tikzlemone{.6}
\caption{Lemma~\ref{lem.cyc.exec.1} configuration \label{fig:ex-lem.1}}
\end{center}
\end{figure}}

\begin{lemma}
\label{lem.cyc.exec.1}
Let \soe be a \wifs, and $f=\fa{\left( \thr{i}, \fis{i} \right)_{i \in \inte{\ct-2}}}$. If, for all
$n \in \inte{\ct-1}$, $\gap{n}{f+1} < 1$, then there exists later in the execution a
\wfs $\soe'$ for \ct threads such that $\fa{\soe'} = f +1$.
\end{lemma}
\begin{proof}
The proof is straightforward; \soe is actually a \wfs such that $\fa{\soe} = f +
1$. \rr{Since $\att{0}{f} - \suc{\ct-1} < \gap{0}{1} <1 $, the first success of
\thr{0} after the success of \thr{\ct-1} is its \kth{f+1} \re.}
\end{proof}

}

\subsection{Cyclic Executions}
\label{sec:wt-cyc}

\pp{We only give the two main theorems used to show
the existence of cyclic executions. The details can be found in the
companion research report~\cite{our-long}.}

\begin{theorem}
\label{th.wfs.gap}
Given a seed $\soe = \left( \thr{i}, \fis{i} \right)_{i \in \inte{\ct-1}}$,
\soe is a \wfs if and only if for all $n \in \inte{\ct-1}$,
$0 \leq \gap{n}{f}<1$.
\end{theorem}
\rr{\begin{proof}

Let $\soe = \left( \thr{i}, \fis{i} \right)_{i \in \inte{\ct-1}}$ be a seed.

\noindent$(\Leftarrow)$
We assume that for all $n \in \inte{\ct-1}$, $0<\gap{n}{f}<1$, and we first show
that the first successes occur in the following order: \thr{0} at \suc{0},
\thr{1} at \suc{1}, $\dots$, \thr{\ct-1} at \suc{\ct-1}, \thr{0} again at
$\att{0}{f}$. The first threads that are successful executes their \ps after their success, then
enters their second \rl: from this moment, they can make the first attempt of
the threads, that has not been successful yet, fail. Therefore, we will look at which
\re of which already successful threads could have an impact on which other threads.

We can notice that for all $n \in \inte{\ct -1}$, if the first success of
\thr{n} occurs at \suc{n}, then its next attempts will potentially occur at
$\att{n}{k} = \suc{n} + 1 + q + r + k$, where $k\geq 0$. More specifically, thanks
to Equation~\ref{eq.gats}, for all $n \leq f$, $\att{n}{k} = \suc{\ct+n-f} +
\gap{n}{f} + k$. Also, for all $k \leq f -n$,
\rr{\begin{align}
\att{n}{k}  - \suc{\ct+n-f+k} &= -\left(\suc{\ct+n-f+k} - \suc{\ct+n-f} -k \right) +\gap{n}{f}\nonumber\\
&= \gap{n}{f} - \gap{\ct+n-f+k}{k}\nonumber\\
\att{n}{k}  - \suc{\ct+n-f+k} &= \gap{n}{f-k}, \label{eq.att.suc}
\end{align}}
\pp{\begin{equation}
\att{n}{k}  - \suc{\ct+n-f+k} = -\left(\suc{\ct+n-f+k} - \suc{\ct+n-f} -k \right) +\gap{n}{f}
= \gap{n}{f} - \gap{\ct+n-f+k}{k} = \gap{n}{f-k}, \label{eq.att.suc}
\end{equation}}
and this implies that if $k>0$,
\begin{equation}
\suc{\ct+n-f+k} - \att{n}{k-1} = 1- \gap{n}{f-k}. \label{eq.att.suc2}
\end{equation}

We know, by hypothesis, that $0<\gap{n}{f-k}<1$, equivalently $0<1-\gap{n}{f-k}<1$. Therefore
Equation~\ref{eq.att.suc} states that if a thread \thr{n'} starts a successful attempt at
\suc{\ct+n-f+k}, then this thread will make the \kth{k} \re of \thr{n} fail, since
\thr{n} enters a \re while \thr{n'} is in a successful \re. And
Equation~\ref{eq.att.suc2} shows that, given a thread \thr{n'} starting a new \re at
\suc{\ct+n-f+k}, the only \re of \thr{n} that can make \thr{n'} fail on its attempt is
the \kth{(k-1)} one. There is indeed only one \re of \thr{n} that can enter a \re before
the entrance of \thr{n'}, and exit the \re after it.

\thr{0} is the first thread to succeed at \suc{0}, because no other thread is in
the \rl at this time. Its next attempt will occur at \att{0}{0}, and all
thread attempts that start before \suc{\ct-f} (included) cannot fail because of
\thr{0}, since it runs then the \ps. Also, since all gaps are positive, the threads \thr{1} to \thr{\ct-f}
will succeed in this order, respectively starting at times \suc{1} to
\suc{\ct-f}.

Then, using induction, we can show that \thr{\ct-f+1}, $\dots$, \thr{\ct-1} succeed in
this order, respectively starting at times \suc{\ct-f+1}, $\dots$, \suc{\ct-1}. For $j \in
\inte{f-1}$, let \pro{j}
be the following property: for
all $n \in \inte{\ct-f+j}$, \thr{n} starts a successful \re at \suc{n}.  We assume that
for a given $j$, \pro{j} is true, and we show that it implies that \thr{\ct-f+j+1} will
succeed at \suc{\ct-f+j+1}. The successful attempt of \thr{\ct-f+j} at
\suc{\ct-f+j} leads, for all $j' \in \inte{j}$, to the failure of the \kth{j'} \re of
\thr{j-j'} (explanation of Equation~\ref{eq.att.suc}). But for each \thr{j'}, this
attempt was precisely the one that could have made
\thr{\ct-f+j+1} fail on its attempt at \suc{\ct-f+j+1} (explanation of Equation~\ref{eq.att.suc}).
Given that all threads \thr{n},
where $n>\ct-f+j+1$, do not start any \rl before \suc{\ct-f+j+1}, \thr{\ct-f+j+1} will
succeed at \suc{\ct-f+j+1}.  By induction, \pro{j} is true for all $j \in \inte{f-1}$.

Finally, when \thr{\ct-1} succeeds, it makes the
\kth{(f-1-n)} \re of \thr{n} fail, for all $n \in \inte{f-1}$;
also the next potentially successful
attempt for \thr{n} is at \att{n}{f-n}. (Naturally, for all $n \in
\inte[f]{\ct-1}$, the next potentially successful attempt for \thr{n} is at
\att{n}{0}.)

We can observe that for all $n < \ct$,  $j \in \inte{\ct-1-n}$, and all $k \geq j$,
\rr{\begin{align}
\att{n+j}{k-j} - \att{n}{k} &= \suc{n+j} + k - j - \left( \suc{n}  +k\right)\nonumber\\
\att{n+j}{k-j} - \att{n}{k} &= \gap{n+j}{j}, \label{eq.att.att}
\end{align}}\pp{$ \att{n+j}{k-j} - \att{n}{k} = \suc{n+j} + k - j - \left( \suc{n}  +k\right)
=  \gap{n+j}{j}$,}
hence for all $n \in \inte[1]{f}$, $\att{n}{f-n} - \att{0}{f} = \gap{n}{n} > 0$.
\rr{\[ \att{n}{f-n} - \att{0}{f} = \gap{n}{n} > 0. \]}

As we have as well, for all $n \in\inte[f+1]{\ct-1}$,
$\att{n}{0} > \att{f}{0}$, we obtain that among all the threads, the
earliest possibly successful attempt is \att{0}{f}. Following
\thr{\ct-1}, \thr{0} is consequently the next successful thread in its
\kth{f} \re.

To conclude this part, we can renumber the threads (\thr{n+1}
becoming now \thr{n} if $n>0$, and \thr{0} becoming \thr{\ct-1}), and follow the
same line of reasoning. The only difference is the fact that \thr{\ct-1}
(according to the new numbering) enters the \rl $f$ units of time before
\suc{\ct-1}, but it does not interfere with the other threads, since we know
that those attempts will fail.

There remains the case where there exists $n \in \inte{\ct-1}$ such that $\gap{n}{f}=0$.
This implies that $f=0$, thus we have a \wfs.

\rr{\bigskip}
\pp{\smallskip}

\noindent$(\Rightarrow)$ \pp{In a nutshell, we prove the implication by
contraposition, \ie if a gap does not belong to $[0,1[$, then one of the properties
of the \wfs is violated.\qedhere}
\rr{We prove now the implication by contraposition; we assume that there exists $n
\in \inte{\ct-1}$ such that $\gap{n}{f}>1$ or $\gap{n}{f}<0$, and show that \soe
is not a \wfs.

We assume first that an \kth{f} order gap is negative. As it is a sum of \kst{1} order
gaps, then there exists $n'$ such that \gap{n'}{1} is negative; let $n''$ be the highest
one.

If $n''>0$, then either the threads $\thr{0}, \dots, \thr{n''-1}$ succeeded in order at
their \kth{0} \re, and then \thr{n''-1} makes \thr{n''} fail at its \kth{0} \re (we have
a seed, hence by definition, $\suc{n''-1}<\suc{n''}$, and $\gap{n''}{1}<0$, thus $\suc{n''-1}<\suc{n''}<\suc{n''-1}+1$ ),
or they did not succeed in order
at their first try. In both cases, \soe is not a \wfs.

If $n''=0$, let us assume that \soe is a \wfs. Let also a new seed be $\soe' =
\left( \thr{i}, \suc{i}' \right)_{i \in \inte{\ct-1}}$, where for all $n \in \inte{\ct-2}$,
$\suc{n+1}'=\suc{n}$, and $\suc{0}' = \suc{\ct-1} - (q+1+f+r) $. Like \soe, $\soe'$ is a
\wfs; however, \gap{1}{1} is negative, and we fall back into the previous case, which
shows that $\soe'$ is not a \wfs. This is absurd, hence \soe is not a \wfs.

\medskip

We assume now that every gap is positive and choose \pl defined by: $\pl = \min \{ n \; ;
\; \exists k \in \inte{\ct-1} / \gap{n+k}{k} >1 \}$, and $\fl = \min \{ k \; ; \; \gap{\pl
  + k}{k} >1 \}$: among the gaps that exceed $1$, we pick those that concern the earliest
thread, and among them the one with the lowest order.

Let us assume that threads \thr{0}, $\dots$, \thr{\ct-1} succeed at their \kth{0} \re in
this order, then \thr{0}, $\dots$, \thr{\pl} complete their second successful \rl at their \kth{f}
\re, in this order. If this is not the case, then \soe is not a \wfs, and the proof is
completed. According to Equation~\ref{eq.att.att}, we have, on the one hand,
$\att{\pl+1}{\fl-1} - \att{\pl}{\fl} = \gap{\pl+1}{1}$, which implies $\att{\pl+1}{\fl} -1
- \att{\pl}{\fl} = \gap{\pl+1}{1}$, thus $\att{\pl+1}{f} - (\att{\pl}{f}+1) =
\gap{\pl+1}{1}$; and on the other hand, $\att{\pl+\fl}{0} - \att{\pl}{\fl} =
\gap{\pl+\fl}{\fl}$ implying $\att{\pl+\fl}{f-\fl} - \left(\att{\pl}{f}+1\right) =
\gap{\pl+\fl}{\fl} -1$. As we know that $\gap{\pl+\fl}{\fl} - \gap{\pl+1}{1} =
\gap{\pl+\fl}{\fl-1} < 1$ by definition of \fl (and \pl), we can derive that
$\att{\pl+1}{f} - (\att{\pl}{f}+1) > \att{\pl+\fl}{f-\fl} - (\att{\pl}{f}+1)$. We have
assumed that \thr{\pl} succeeds at its \kth{f} \re, which will end at
$\att{\pl}{f}+1$. The previous inequality states then that \thr{\pl+1} cannot be
successful at its \kth{f} \re, since either a thread succeeds before \thr{\pl+\fl} and
makes both \thr{\pl+\fl} and \thr{\pl+1} fail, or \thr{\pl+\fl} succeeds and makes
\thr{\pl+1} fail. We have shown that \soe is not a \wfs.}
\end{proof}
}

\rr{

\begin{lemma}
\label{lem.event.succ}
Assuming $r \neq 0$, if a new thread is added to an \cyex{f}{\ct}, it will
eventually succeed.
\end{lemma}
\begin{proof}

Let \att{\ct}{0} be the time of the \kth{0} \re of the new thread, that we
number \thr{\ct}. If this \re is successful, we are done; let us assume now
that this \re is a failure, and let us shift the thread numbers (for the
threads \thr{0}, $\dots$, \thr{\ct-1}) so that \thr{0} makes \thr{\ct} fail on
its first attempt. We distinguish two cases, depending on whether $\gap{0}{\ct}
> \att{\ct}{0} - \suc{0}$ or not.

We assume that $\gap{0}{\ct} > \att{\ct}{0} - \suc{0}$. We know that $n \mapsto
\gap{n}{n}$ is increasing on \inte{\ct-1} and that $\gap{0}{0} = 0$, hence let
$\pl = \min \{ n \in \inte{\ct-1} \; ; \; \gap{n}{n} < \att{\ct}{0} - \suc{0} \}$.
For all $k \in \inte{\pl}$, we have
$\att{\ct}{k} - \suc{k} = k + \att{\ct}{0} - (\gap{k}{k} + \suc{0} + k) =
\att{\ct}{0} - \suc{0} - \gap{k}{k}$ hence $\att{\ct}{k} - \suc{k}>0$ and
$\att{\ct}{k} - \suc{k} < \att{\ct}{0} - \suc{0} < 1$.  This shows that \thr{0},
$\dots$, \thr{\pl}, because of their successes at \suc{0}, $\dots$, \suc{\pl},
successively make \kth{0}, $\dots$, \kth{\pl} \res (respectively) of \thr{\ct} fail.
The next attempt for \thr{\ct} is at \att{\ct}{\pl+1}, which fulfills the
following inequality: $\att{\ct}{\pl+1} - (\suc{\pl}+1) < \suc{\pl+1} - (\suc{\pl}+1)$
since
\begin{align*}
\att{\ct}{\pl+1} - \suc{\pl+1} &=  (\pl+1 + \att{\ct}{0}) -(\gap{\pl+1}{\pl+1} + \suc{0} + \pl+1)\\
\att{\ct}{\pl+1} - \suc{\pl+1} &> 0.
\end{align*}
\thr{\pl+1} should have been the successful thread, but \thr{\ct} starts a
\re before \suc{\pl+1}, and is therefore succeeding.

We consider now the reverse case by assuming that $\gap{0}{\ct} < \att{\ct}{0} -
\suc{0}$. With the previous line of reasoning, we can show that \thr{0}, $\dots$,
\thr{\ct-1}, because of their successes at \suc{0}, $\dots$, \suc{\ct-1}, successively
make \kth{0}, $\dots$, \kth{(\ct-1)} \res (respectively) of \thr{\ct} fail.  Then we are
back in the same situation when \thr{0} made \thr{\ct} fail for the first time (\thr{0}
makes \thr{\ct} fail), except that the success of \thr{0} starts at $\suc{0}' = \suc{0} +
\gap{0}{\ct}$. As $\gap{0}{\ct} = q+r+f-\ct >0$ and $q$, f and \ct are integers, we have
that $\gap{0}{\ct} \geq r$. By the way, if we had $\gap{0}{\ct} > r$, we would have
$\gap{0}{\ct} \geq 1+r > \att{\ct}{0} - \suc{0}$, which is absurd. \suc{0} makes indeed
\att{\ct}{0} fail, therefore \gap{0}{\ct} should be less than $1$. Consequently, we are
ensured that $\gap{0}{\ct} = r$. We define
\[ k_0 = \left\lfloor \frac{\att{\ct}{0} - \suc{0}}{r} \right\rfloor;  \]
also, for every $k \in \inte[1]{k_0}$, $r < \att{\ct}{0} - (\suc{0} +
k\times r)$ and $r > \att{\ct}{0} - (\suc{0} + (k_0+1)\times r)$: the cycle of
successes of \thr{0}, $\dots$, \thr{\ct-1} is executed $k_0$ times. Then the
situation is similar to the first case, and \thr{\ct} will succeed.

\end{proof}
}

\newcommand{\perfun}{\ema{\sigma}}
\newcommand{\per}[1]{\ema{\perfun \left( #1 \right)}}
\newcommand{\iperfun}{\ema{\perfun^{-1}}}
\newcommand{\iper}[1]{\ema{\iperfun \left( #1 \right)}}
\newcommand{\enfun}{\ema{m_2}}
\newcommand{\en}[1]{\ema{\enfun \left( #1 \right)}}

\newcommand{\enofun}{\ema{m_1}}
\newcommand{\eno}[1]{\ema{\enofun \left( #1 \right)}}

\newcommand{\evafun}{\ema{\mathit{rank}}}
\newcommand{\eva}[1]{\ema{ \evafun \left( #1 \right)}}
\newcommand{\card}[1]{\ema{\# #1}}
\newcommand\restr[2]{{    \left.\kern-\nulldelimiterspace   #1   \vphantom{\big|}   \right|_{#2} }}

\rr{

\rr{\begin{figure}
\tikzlemthree{.6}
\caption{Lemma~\ref{lem.cyc.exec.3} configuration \label{fig:ex-lem.3}}
\end{center}
\end{figure}}

\begin{lemma}
\label{lem.cyc.exec.3}
Let \soe be a \wifs, and $f=\fa{\left( \thr{i}, \fis{i} \right)_{i \in \inte{\ct-2}}}$. If
$\gap{f}{f+1} > 1$,  and if the second success of
\thr{\ct-1} does not occur before the second success of \thr{f-1}, then we can find in
the execution a \wfs $\soe'$ for \ct threads such that $\fa{\soe'} = f$.
\end{lemma}
\begin{proof}

\newcommand{\seth}[1]{\ema{\mathcal{S}_{#1}}}

Let us first remark that, by the definition of a \wifs, all threads will succeed once, in
order.  Then two ordered groups of threads will compete for each of the next successes,
until \thr{f-1} succeeds for the second time.

Let $e$ be the smallest integer of \inte[f]{\ct-1} such that the second success of \thr{e}
occurs after the second success of \thr{f-1}. Let then \seth{1} and \seth{2} be
the two groups of threads that are in competition, defined by
\begin{align*}
\seth{1} = \{ \thr{n} \; ; \; n \in \inte{f-1} \}\\
\seth{2} = \{ \thr{n} \; ; \; n \in \inte[f]{e-1} \}\\
\end{align*}

For all $n \in \inte{e-1}$, we note
\[ \eva{n}  = \left\{ \begin{array}{ll}
\gap{n}{n+1} & \text{ if } \thr{n} \in \seth{1}\\
\gap{n}{n+1} -1 \quad & \text{ if } \thr{n} \in \seth{2}
\end{array} \right. . \]
We define \perfun, a permutation of \inte{e-1} that describes the reordering of
the threads during the round of the second successes, such that, for all $(i,j)
\in \inte{e-1}^2$, $\per{i} < \per{j}$ if and only if $\eva{i} < \eva{j}$.

We also define a function that will help in expressing the \iper{k}'s:
\[
\begin{array}{cccc}
\enfun:& \inte{e-1} & \longrightarrow & \inte[f]{e-1}\\
 & k & \longmapsto & \max {\{ \ell \in \inte[f]{e-1} \; ; \; \thr{\ell} \in \seth{2} \; ; \; \per{\ell} \leq k  \} }
\end{array}.
\]

We note that $\restr{\evafun}{\inte{f-1}}$ is increasing, as well as
$\restr{\evafun}{\inte[f]{e-1}}$. This shows that
$\card \{ \thr{\ell} \in \seth{2} \; ; \; \per{\ell} \leq k  \} = \en{k} -(f -1)$.
Consequently, if $\thr{\iper{k}} \in \seth{2}$, then
\begin{align*}
\en{k} &= \card \{ \thr{\ell} \in \seth{2} \; ; \; \per{\ell} \leq k  \} +f -1\\
&= \card \{ \thr{\ell} \in \seth{2} \; ; \; \ell \leq \iper{k}  \} +f -1\\
&= \iper{k}-f+1+f-1\\
\en{k} &= \iper{k}.
\end{align*}

Conversely, if $\thr{\iper{k}} \in \seth{1}$, among $\{ \thr{\per{n}} \; ; \; n
\in \inte{k} \}$, there are exactly $\en{k}-f+1$ threads in \seth{2}, hence
\[ \iper{k} = k+1-(\en{k}-f+1)-1 = f + k -\en{k} -1.\]

In both cases, among $\{ \thr{\per{n}} \; ; \; n \in \inte{k} \}$, there are
exactly $\en{k}-f+1$ threads in \seth{2}, and $\eno{k} = k - (\en{k}-f)$ threads
in \seth{1}.

We prove by induction that after this first round, the next successes will be,
respectively, achieved by \thr{\iper{0}}, \thr{\iper{1}}, $\dots$,
\thr{\iper{e-1}}. In the following, by ``\kth{k} success'', we mean \kth{k}
success after the first success of \thr{\ct-1}, starting from $0$, and the
\att{i}{j}'s denote the attempts of the second round.

Let \pro{K} be the following property: for all $k \leq K$, the \kth{k} success
is achieved by \thr{\iper{k}} at \att{\iper{k}}{f+k-\iper{k}}. We assume \pro{K}
true, and we show that the \kth{(K+1)} success is achieved by \thr{\iper{K+1}} at
\att{\iper{K+1}}{f+K+1-\iper{K+1}}.

We first show that if $\thr{\iper{K}} \in \seth{1}$, then
\begin{equation}
 \att{\en{K}+1}{\eno{K}-1} > \att{\iper{K}}{f+K-\iper{K}} > \att{\en{K}}{\eno{K}}.
\label{eq.inv.fm1}
\end{equation}
On the one hand,
\begin{align*}
\att{\iper{K}}{f+K-\iper{K}} &= K - \iper{K} + \att{\iper{K}}{f}\\
&= K - \iper{K} + \att{0}{f} + \iper{K} + \gap{\iper{K}}{\iper{K}}\\
&= K + \suc{\ct-1} +1 + \gap{0}{1} + \gap{\iper{K}}{\iper{K}}\\
\att{\iper{K}}{f+K-\iper{K}}
&= K + \suc{\ct-1} +1 + \gap{\iper{K}}{\iper{K}+1}.\\
\end{align*}
On the other hand,
\begin{align*}
\att{\en{K}}{f+K-\en{K}} &= (\en{K}-f) + \att{f}{K-(\en{K}-f)}  + \gap{\en{K}}{\en{K}-f}\\
&= (\en{K}-f)  + K - (\en{K}-f) + \att{f}{0} + \gap{\en{K}}{\en{K}-f}\\
&= (\en{K}-f)  + K - (\en{K}-f) + \suc{\ct-1} + 1 + (\gap{f}{f+1}-1) + \gap{\en{K}}{\en{K}-f}\\
\att{\en{K}}{f+K-\en{K}}
&= K + \suc{\ct-1} + 1 + \gap{\en{K}}{\en{K}+1} -1.\\
\end{align*}
Therefore,
\begin{align*}
\att{\iper{K}}{f+K-\iper{K}} - \att{\en{K}}{\eno{K}}
&= \att{\iper{K}}{f+K-\iper{K}} - \att{\en{K}}{f+K-\en{K}}\\
&= \gap{\iper{K}}{\iper{K}+1} - \left( \gap{\en{K}}{\en{K}+1} -1  \right)\\
\att{\iper{K}}{f+K-\iper{K}} - \att{\en{K}}{\eno{K}}&= \eva{\iper{K}} - \eva{\en{K}}.
\end{align*}

In a similar way, we can obtain that if $\thr{\iper{K}} \in \seth{2}$, then
\begin{equation}
\att{\eno{K}}{\en{K}} > \att{\iper{K}}{f+K-\iper{K}} > \att{\eno{K}-1}{\en{K}+1}.
\label{eq.inv.em1}
\end{equation}

In addition, we recall that if $\thr{\iper{K}} \in \seth{2}$, $\iper{K} = \en{K}$,
thus the second inequality of Equation~\ref{eq.inv.fm1} becomes an equality, and
if $\thr{\iper{K}} \in \seth{1}$, $\iper{K} = f + K -\en{K} -1$, hence the second
inequality of Equation~\ref{eq.inv.em1} becomes an equality.

Now let us look at which attempt of other threads \thr{\iper{K}} made fail. From
now on, and until explicitly said otherwise, we assume that $\thr{\iper{K}} \in
\seth{1}$. According to Equation~\ref{eq.inv.fm1}, we have
\begin{align*}
\att{\en{K}+1}{\eno{K}-1} &>& \att{\iper{K}}{f+K-\iper{K}} &>& \att{\en{K}}{\eno{K}}\\
\att{\en{K}+j}{\eno{K}-j} - \att{\en{K}+1}{\eno{K}-1} &<& \att{\en{K}+j}{\eno{K}-j} - \att{\iper{K}}{f+K-\iper{K}}
&<& \att{\en{K}+j}{\eno{K}-j} - \att{\en{K}}{\eno{K}}\\
\gap{\en{K}+j}{j-1} &<& \att{\en{K}+j}{\eno{K}-j} - \att{\iper{K}}{f+K-\iper{K}} &<& \gap{\en{K}+j}{j}
\end{align*}
This holds for every $j \in \inte[1]{\eno{K}}$, implying $j \leq f$, since there
could not be more than $f$ threads in \seth{1}.
Therefore, as by assumptions gaps of at most \kth{f} order are between $0$ and $1$,
\[ 0 < \att{\en{K}+j}{\eno{K}-j} - \att{\iper{K}}{f+K-\iper{K}} < 1; \]
showing that the success of \thr{\iper{K}} makes thread \thr{\en{K}+j} fail on
its attempt at \att{\en{K}+j}{\eno{K}-j}, for all $j \in \inte[1]{\eno{K}}$.

Since $\thr{\iper{K}} \in \seth{1}$, $\iper{K} = \eno{K}-1$. Also, for all $j
\in \inte{f-1-\eno{K}}$,
\begin{align*}
\att{\eno{K} + j}{\en{K} - j} - \att{\iper{K}}{f+K-\iper{K}}
&= \att{\eno{K} + j}{\en{K} - j} - \att{\eno{K} - 1}{\en{K} + 1} \\
&= \left( \att{\eno{K} -1}{\en{K} - j} + (j+1) + \gap{\eno{K}+j}{j+1} \right)
 - \left( \att{\eno{K} - 1}{\en{K} -j} + (j+1)  \right) \\
\att{\eno{K} + j}{\en{K} - j} - \att{\iper{K}}{f+K-\iper{K}}
&= \gap{\eno{K}+j}{j+1}
\end{align*}
As a result, \thr{\iper{K}} makes \thr{\eno{K} + j} fail on its attempt at
\att{\eno{K}+j}{\en{K} - j}, for all $j\in \inte{f-1-\eno{K}}$, and the next
attempt will occur at \att{\eno{K}+j}{\en{K} - j+1}.

Altogether, the next attempt after the end of the success of \thr{\iper{K}} for
\thr{\eno{K}+j} is \att{\eno{K}+j}{\en{K}-j+1}, for $j\in \inte{f-1-\eno{K}}$,
and for \thr{\en{K}+j} is \att{\en{K}+j}{\eno{K}-j+1}, for all $j \in
\inte[1]{\eno{K}}$.

Additionally, a thread will begin a new \rl, the \kth{0} \re being
at $\att{\en{K}+\eno{K}+1}{0} = \att{f+K+1}{0}$. We note that $f+K+1$ could be
higher than $\ct-1$, referring to a thread whose number is more than $\ct-1$.
Actually, if $n>\ct-1$, \att{n}{j} refers to the \kth{j} \re of
\thr{\eva{n-\ct+1}}, after its first two successes.

The two heads, \ie the two smallest indices, of $\seth{1} \cap \iper{\inte[K+1]{e-1}}$ and
$\seth{2} \cap \iper{\inte[K+1]{e-1}}$ will then compete for being successful. Indeed,
within \seth{1}, for $j\in \inte{f-1-\eno{K}}$,
\[ \att{\eno{K}+j}{\en{K}-j+1} - \att{\eno{K}}{\en{K}+1} = \gap{\eno{K}+j}{j} > 0, \]
thus if someone succeeds in \seth{1}, it will be \thr{\eno{K}}.
In the same way, for all $j \in \inte[1]{\eno{K}+1}$,
\[ \att{\en{K}+j}{\eno{K}-j+1}  - \att{\en{K}+1}{\eno{K}} = \gap{\en{K}+j}{j-1} > 0, \]
meaning that if someone succeeds in \seth{2}, it will be \thr{\en{K}+1}.

Let us compare now those two candidates:
\begin{align*}
\att{\eno{K}}{\en{K}+1} - \att{\en{K}+1}{\eno{K}}
&= \en{K}+1-f+\suc{\ct-1} + \eno{K} + \gap{\eno{K}}{\eno{K}+1}\\
& \quad \quad - \left( \eno{K} + \att{f}{0} + \en{K}+1 -f  + \gap{\en{K}+1}{\en{K}+1-f} \right)\\
&= \suc{\ct-1} -1 + \gap{\eno{K}}{\eno{K}+1}\\
& \quad \quad - \left(  \suc{\ct-1} + \gap{f}{f+1} -1  + \gap{\en{K}+1}{\en{K}+1-f} \right)\\
&= \gap{\eno{K}}{\eno{K}+1} - \left( \gap{\en{K}+1}{\en{K}+2} -1 \right)\\
\att{\eno{K}}{\en{K}+1} - \att{\en{K}+1}{\eno{K}}&= \eva{\eno{K}} - \eva{\en{K}+1}.\\
\end{align*}

By definition, \iper{K+1} is either \eno{K} or $\en{K}+1$ and corresponds to the
next successful thread.
We can follow the same line of reasoning in the case where $\thr{\iper{K}} \in \seth{2}$
and prove in this way that \pro{K+1} is true.

\pro{0} is true, and the property spreads until \pro{e-1}, where all threads of
\seth{1} and \seth{2} have been successful, in the order ruled by \iperfun, \ie
\thr{\iper{0}}, $\dots$, \thr{\iper{e-1}}. And before those successes the threads
\thr{e-1}=\thr{\iper{e-1}}, $\dots$, \thr{\ct-1} have been successful as
well. The seed composed of those successes is a \wfs. Given a thread, the gap
between this thread and the next one in the new order could indeed not be higher
than the gap in the previous order with its next thread. Also the \kth{f} order
gaps remain smaller than $1$. And as \thr{e-1} succeeds the second time after
$f$ failures, it means that the new seed $\soe''$ is such that $\fa{\soe''} =
f$.\qedhere

\end{proof}

\begin{figure}
\tikzlemtwo{.6}
\caption{Lemma~\ref{lem.cyc.exec.2} configuration \label{fig:ex-lem.2}}
\end{center}
\end{figure}

\begin{lemma}
\label{lem.cyc.exec.2}
Let \soe be a \wifs, and $f=\fa{\left( \thr{i}, \fis{i} \right)_{i \in \inte{\ct-2}}}$. If
$\gap{f}{f+1} > 1$ and if the second success of \thr{\ct-1} occurs before the
second success of \thr{f-1}, then we can find in the execution a \wfs $\soe'$ for \ct
threads such that $\fa{\soe'} = f$.
\end{lemma}
\begin{proof}
\pp{In~\cite{our-long}, we show that in this case, the gap between two successes of
\thr{\ct-1} is $r$, hence all gaps at any order less than \ct will eventually become less
than $r$, which builds the \wfs.
}
\rr{Until the second success of \thr{\ct-1}, the execution follows the same pattern as in
  Lemma~\ref{lem.cyc.exec.3}. Actually, the case invoked in the current lemma could have
  been handled in the previous lemma, but it would have implied tricky notations, when we
  referred to \thr{\eva{n-\ct+1}}. Let us deal with this case independently then, and come
  back to the instant where \thr{\ct-1} succeeds for the second time.

We had $0 < \att{f-1}{0} - \suc{\ct-1} = \gap{f-1}{f} < 1$. For the thread \thr{\per{j}}
to succeed at its \kth{k} \re after the first success of \thr{\ct-1} and before \thr{f-1},
it should necessary fill the following condition: $j+1 < \att{\per{j}}{k} - \suc{\ct-1} <
j+1 + \gap{f-1}{f}$. This holds also for the second success of \thr{\ct-1}, which implies
that $P' < \suc{\ct-1} + 1 + q + r + h - \suc{\ct-1} < P' + \gap{f-1}{f}$, where $h$ is
the number of failures of \thr{\ct-1} before its second success and $P'$ is the number of
successes between the two successes of \thr{\ct-1}. As $\gap{f-1}{f}<1$, and $q$, $P'$ and
$h$ are non-negative integers, we have $r< \gap{f-1}{f}$ and $h = P' - 1 -q$.

To conclude, as any gap at any order is less than the gap between the two
successes of \thr{\ct-1}, which is $r < 1$, we found a \wfs for $P'$ threads.

Finally any other thread will eventually succeed (see
Lemma~\ref{lem.event.succ}). We can renumber the threads such that
\thr{P'} is the first thread that is not in the \wfs to succeed, and
the threads of the \wfs succeeded previously as \thr{0}, $\dots$,
\thr{P'-1}. As explained before, for all $(k,n) \in \inte{P'-1}^2$,
$\gap{n}{k} < \gap{n}{n} = r$. With the new thread, the first order
gaps are changed by decomposing \gap{0}{1} into \gap{P'}{1} and the
new \gap{0}{1}. All gaps can only be decreased, hence we have a new
\wfs for $P'+1$ threads. We repeat the process until all threads have
been encountered, and obtain in the end $\soe'$, a \wfs with \ct
threads such that $\fa{\soe'} = \ct -1 -q$, which is an optimal cyclic
execution.

Still, as \thr{f} succeeds between two successes of \thr{\ct-1} that are separated by $r$,
we had, in the initial configuration: $\gap{\ct-1}{\ct-1-f} < r$. As, in addition, we have
both $\gap{f-1}{f} <1 $ and $\gap{f}{1} <1 $, we conclude that the lagging time was
initially less than $2+r$. By hypothesis, we know that $\gap{f}{f+1} >1 $, which implies
that, before the entry of the new thread, the lagging time was $1+r$.
In the final execution with one more thread, the lagging time is $r$ and we have one more
success in the cycle, thus $\fa{\soe'} = f$.}
\end{proof}

}

\begin{theorem}
\label{th.cyc.exec}
Assuming $r \neq 0$, if a new thread is added to an \cyex{f}{\ct-1}, then all the
threads will eventually form either an \cyex{f}{\ct}, or an \cyex{f+1}{\ct}.
\end{theorem}
\rr{ \begin{proof}
According to Lemma~\ref{lem.event.succ}, the new thread will eventually
succeed.  In addition, we recall that Properties~\ref{pr.wfs.soe1} and~\ref{pr.wfs.soe2} ensure that
before the first success of the new thread, any set of $\ct -1$ consecutive
successes is a \wfs with $\ct-1$ threads. We then consider a seed (we number the
threads accordingly, and number the new thread as \thr{\ct-1}) such that the
success of the new thread occurs between the success of \thr{\ct-2} and \thr{0};
we obtain in this way a \wifs $\soe = \left( \thr{n}, \fis{n} \right)_{n \in \inte{\ct-1} \& }$. We differentiate between two cases.

Firstly, if for all $n \in \inte{\ct-1}$, $\gap{n}{f+1} < 1$, according to
Lemma~\ref{lem.cyc.exec.1}, we can find later in the execution a \wfs $\soe'$
for \ct threads such that $\fa{\soe'} = f +1$, hence we reach eventually an
\cyex{f+1}{\ct}.

Let us assume now that this condition is not fulfilled. There exists $n_0 \in
\inte{\ct-1}$ such that $\gap{n_0}{f+1} > 1$. We shift the thread numbers, such
that $n_0$ is now $f$, and we have then $\gap{f}{f+1} > 1$. Then two cases are
feasible. If the second success of \thr{\ct-1} occurs before the second success
of \thr{f-1}, then Lemma~\ref{lem.cyc.exec.3} shows that we will reach an
\cyex{f}{\ct}. Otherwise, from Lemma~\ref{lem.cyc.exec.3}, we conclude that an
\cyex{f}{\ct} will still occur.
} \pp{
\begin{proof}
We decompose the Theorem into three Lemmas which we describe here
  graphically:
\begin{compactitem}
\item If all gaps of \kth{(f+1)} order are less than 1, then every existing thread will
  fail once more, and the new steady-state is reached immediately. See
  Figure~\ref{fig:ex-lem.1}.
\item Otherwise
\begin{compactitem}
\item Either: everyone succeeds once, whereupon a new \cyex{f}{\ct} is formed. See
  Figure~\ref{fig:ex-lem.3}.
\item Or: before everyone succeeds again, a new \cyex{f}{\ct'}, where $\ct' \leq \ct$,
  is formed, which finally leads to an \cyex{f}{\ct}. See Figure~\ref{fig:ex-lem.2}.\qedhere
\end{compactitem}
\end{compactitem}
}
\end{proof}

\pp{
\setlength{\intextsep}{0.2\baselineskip plus 0.2\baselineskip minus 2\baselineskip}
\setlength{\belowcaptionskip}{15pt}
\begin{figure}
\begin{subfigure}[b]{\textwidth}
\captionsetup{skip=-10pt}
\caption{New thread does not lead to a reordering}\label{fig:ex-lem.1}
\tikzlemone{.44}
\end{center}
\end{subfigure}
\begin{subfigure}[b]{\textwidth}
\captionsetup{skip=-10pt}
\caption{Reordering and immediate new seed}\label{fig:ex-lem.3}
\tikzlemthree{.44}
\end{center}
\end{subfigure}
\begin{subfigure}[b]{\textwidth}
\captionsetup{skip=-10pt}
\caption{Reordering and transient state}\label{fig:ex-lem.2}
\tikzlemtwo{.42}
\end{center}
\end{subfigure}
\setlength{\abovecaptionskip}{0pt}
\setlength{\belowcaptionskip}{0pt}
\caption{Illustration of Theorem~\ref{th.cyc.exec}}
\end{figure}
}

\subsection{Throughput Bounds}
\label{sec:wt-thr}

\pp{
The periodicity offers an easy way to compute the expected number of threads inside the \rl,
and to bound the number of failures and the throughput.

\begin{lemma}
\label{lem:av-thr}
In an \cyex{f}{\ct}, the throughput is $\thru = \frac{\ct}{q+r+1+f}$, and the average number of threads in the \rl
$\trl = \ct \times \frac{f+1}{q+r+f+1}$.
\end{lemma}

\begin{lemma}
\label{lem:wt-bounds}
The number of failures is tighly bounded by $\fup\leq f\leq\flo$, and throughput
by $\tup\leq\thru\leq\tlo$, where
\begin{align*}
\fup = \left\{ \begin{array}{ll}
\ct-q-1 \quad& \text{if } q \leq \ct-1\\
0 \quad& \text{otherwise}
\end{array}\right.,  \qquad
 \tup =
\left\{ \begin{array}{ll}
\frac{\ct}{\ct+r} \quad& \text{if } q \leq \ct-1\\
\frac{\ct}{q+r+1} \quad& \text{otherwise.}
\end{array}\right.\\
\flo = \left\lfloor \frac{1}{2} \left( (\ct-1-q-r) +
\sqrt{(\ct-1-q-r)^2 + 4\ct} \right) \right\rfloor,  \;
\tlo = \frac{\ct}{q+r+1+\flo}.
\end{align*}
\end{lemma}

\newcommand{\wupu}{\ema{\tilde{w}}}
\newcommand{\wupf}[1]{\ema{\wupu(#1)}}
\newcommand{\wup}{\ema{\wupf{\ct}}}
\newcommand{\wupp}{\ema{\wupu'(\ct)}}

\newcommand{\hfu}{\ema{a}}
\newcommand{\hf}{\ema{\hfu(\ct)}}
\newcommand{\hfp}{\ema{\hfu'(\ct)}}

\newcommand{\hhfu}{\ema{h}}
\newcommand{\hhff}[1]{\ema{\hhfu(#1)}}
\newcommand{\hhf}{\ema{\hhff{\ct}}}
\newcommand{\hhfp}{\ema{\hhfu'(\ct)}}

\newcommand{\flofun}[1]{\ema{\left\lfloor #1 \right\rfloor}}
\newcommand{\ceifun}[1]{\ema{\left\lceil #1 \right\rceil}}

\smallskip
}

\rr{
Firstly we calculate the expression of throughput and the expected number of threads inside
the \rl (that is needed when we gather expansion and \cacas). Then we
exhibit upper and lower bounds on both throughput and the number of failures, and show that
those bounds are reached. Finally, we give the worst case on the number of \cacas.

\begin{lemma}
\mbox{In an \cyex{f}{\ct}, the throughput is}
\begin{equation}
\thru = \frac{\ct}{q+r+1+f}.\label{eq:cye-thr}
\end{equation}
\end{lemma}
\begin{proof}
By definition, the execution is periodic, and the period lasts $q+r+1+f$ units of time. As
\ct successes occur during this period, we end up with the claimed expression.
\end{proof}

\begin{lemma}
\label{lem:av-thr}
In an \cyex{f}{\ct}, the average number of threads \trl in the \rl
is given by
\[ \trl = \ct \times \frac{f+1}{q+r+f+1}.  \]
\end{lemma}
\begin{proof}
Within a period, each thread spends $f+1$ units of time in the \rl, among the $q+r+f+1$
units of time of the period, hence the Lemma.
\end{proof}

\bigskip

\begin{lemma}
\label{lem:wt-upp}
\mbox{The number of failures is not less than \fup, where}
\begin{equation}
\fup = \left\{ \begin{array}{ll}
\ct-q-1 \quad& \text{if } q \leq \ct-1\\
0 \quad& \text{otherwise}
\end{array}\right., \text{ and accordingly, }\qquad
 \thru \leq
\left\{ \begin{array}{ll}
\frac{\ct}{\ct+r} \quad& \text{if } q \leq \ct-1\\
\frac{\ct}{q+r+1} \quad& \text{otherwise.}
\end{array}\right. \label{eq:wt-upp}
 \end{equation}
\end{lemma}
\begin{proof}
According to Equation~\ref{eq:cye-thr}, the throughput is maximized when the number of
failures is minimized. In addition, we have two lower bounds on the number of failures:
(i) $f \geq 0$, and (ii) \ct successes should fit within a period, hence $q+1+f \geq
\ct$. Therefore, if $\ct-1-q < 0$, $\thru \leq \ct/(q+r+1+0)$, otherwise,
\[ \thru \leq \frac{\ct}{q+r+1+\ct-1-q} = \frac{\ct}{\ct+r}.\]
\end{proof}

\begin{remark}
We notice that if $q > \ct-1$, the upper bound in
Equation~\ref{eq:wt-upp} is actually the same as the immediate upper
bound described in Section~\ref{sec:imm-bou}. However, if $q \leq
\ct-1$, Equation~\ref{eq:wt-upp} refines the immediate upper bound.
\end{remark}

\begin{lemma}
\label{lem:wt-low}
The number of failures is bounded by
\[ f \leq \flo = \left\lfloor \frac{1}{2} \left( (\ct-1-q-r) +
\sqrt{(\ct-1-q-r)^2 + 4\ct} \right) \right\rfloor, \text{ and accordingly,}\]
the throughput is bounded by
\[\thru \geq \frac{\ct}{q+r+1+\flo}. \]
\end{lemma}
\begin{proof}
We show that a necessary condition so that an \cyex{f}{\ct}, whose \lag is \lagt, exists, is
$f \times (\lagt+r) < \ct$. \rr{According to Property~\ref{pr.wfs.soe1}, any set of \ct
consecutive successes is a \wfs with $\ct$ threads. Let \soe be any of them. As we have
$f$ failures before success, Theorem~\ref{th.wfs.gap} ensures that for all $n \in
\inte{\ct-1}$, $\gap{n}{f} <1$. We recall that for all $n \in \inte{\ct-1}$, we also have
$\gap{n}{\ct} = \lagt+r$.

On the one hand, we have
\begin{align*}
\sum_{n=0}^{\ct-1} \gap{n}{f} &= \sum_{n=0}^{\ct-1} \sum_{j=n-f+1}^{n} \gap{j\bmod \ct}{1}   \\
&= f \times \sum_{n=0}^{\ct-1}  \gap{n}{1} \\
\sum_{n=0}^{\ct-1} \gap{n}{f} &= f \times (\lagt+r).
\end{align*}
On the other hand, $\sum_{n=0}^{\ct-1} \gap{n}{f} < \sum_{n=0}^{\ct-1} 1 = \ct$.

Altogether, the necessary condition states that $f \times (\lagt+r) < \ct$, which can be
rewritten as $f \times (q+1+f-\ct+r) < \ct$. The proof is complete since minimizing the
throughput is equivalent to maximizing the number of failures.
\end{proof}

\begin{lemma}
\label{lem:wt-reach}
For each of the bounds defined in Lemmas~\ref{lem:wt-upp}
and~\ref{lem:wt-low}, there exists an \cyex{f}{\ct} that reaches the
bound.
\end{lemma}
\begin{proof}
According to Lemmas~\ref{lem:wt-upp} and~\ref{lem:wt-low}, if an \cyex{f}{\ct} exists, then
the number of failures is such that $\fup \leq f \leq \flo$.
\\We show now that this double
necessary condition is also sufficient. We consider $f$ such that $\fup \leq f \leq
\flo$, and build a \wfs $\soe = \left( \thr{i}, \fis{i} \right)_{i \in \inte{\ct-1}}$.

For all $n \in \inte{\ct-1}$, we define \fis{i} as}
\[ \fis{n} = n \times \left( \frac{q+1+f-\ct+r}{\ct} + 1 \right). \]

We first show that $\fa{\soe} = f$.
By definition, $\fa{\soe} = \max \left(0, \left\lceil \fis{\ct-1}-\fis{0}-q-r \right\rceil \right)$;
we have then
\begin{align*}
\fa{\soe} &=
\max \left(0, \left\lceil (\ct -1) \times \left( \frac{q+1+f-\ct+r}{\ct} + 1 \right) -q-r \right\rceil \right)\\
&= \max \left(0, \left\lceil (\ct -1 - q -r) + (q+1+f-\ct+r) - \frac{q+1+f-\ct+r}{\ct} \right\rceil \right)\\
\fa{\soe}&= \max \left(0, \left\lceil f - \frac{q+1+f-\ct+r}{\ct} \right\rceil \right).\\
\end{align*}
Firstly, we know that $q+1+f-\ct \geq 0$, thus if $f=0$, then the second term of the
maximum is not positive, and $\fa{\soe}=0=f$. Secondly, if $f>0$, then according to
Lemma~\ref{lem:wt-upp}, $(q+1+f-\ct+r)/\ct<1/f\leq 1$. As we also have
$(q+1+f-\ct+r)/\ct\geq 0$, we conclude that $\fa{\soe} = \left\lceil f -
\frac{q+1+f-\ct+r}{\ct} \right\rceil = f$.

Additionally, for all $n \in \inte{\ct-1}$,
\begin{align*}
\gap{n}{f}  &=
\left\{ \begin{array}{ll}
\fis{n} - \fis{n-f} - f \quad& \text{if } n>f \\
\fis{n} - \fis{\ct+n-f} +1+q+r \quad& \text{otherwise} \\
\end{array}\right.\\
&=
\left\{ \begin{array}{l}
n \times \left( \frac{q+1+f-\ct+r}{\ct} + 1 \right) - (n-f) \times \left( \frac{q+1+f-\ct+r}{\ct} + 1 \right) - f\\
n \times \left( \frac{q+1+f-\ct+r}{\ct} + 1 \right) - (\ct+n-f)\times \left( \frac{q+1+f-\ct+r}{\ct} + 1 \right) +1+q+r \\
\end{array}\right.\\
&=
\left\{ \begin{array}{l}
f \times \frac{q+1+f-\ct+r}{\ct}\\
- (\ct-f) - (q+1+f-\ct+r) + f \times \frac{q+1+f-\ct+r}{\ct} +1+q+r\\
\end{array}\right.\\
\gap{n}{f} &=
f \times \frac{w+r}{\ct}\\
\end{align*}
As $w \leq 0$ and $f\leq 0$, $\gap{n}{f} > 0$. Since $f \leq \flo$, $\gap{n}{f} < 1$.
Theorem~\ref{th.wfs.gap} implies that \soe is a \wfs that leads to an \cyex{f}{\ct}.

We have shown that for all $f$ such that $\fup \leq f \leq \flo$ there exists an
\cyex{f}{\ct}; in particular there exist an \cyex{\flo}{\ct} and an \cyex{\fup}{\ct}.
\end{proof}

\newcommand{\wupu}{\ema{\tilde{w}}}
\newcommand{\wupf}[1]{\ema{\wupu(#1)}}
\newcommand{\wup}{\ema{\wupf{\ct}}}
\newcommand{\wupp}{\ema{\wupu'(\ct)}}

\newcommand{\hfu}{\ema{a}}
\newcommand{\hf}{\ema{\hfu(\ct)}}
\newcommand{\hfp}{\ema{\hfu'(\ct)}}

\newcommand{\hhfu}{\ema{h}}
\newcommand{\hhff}[1]{\ema{\hhfu(#1)}}
\newcommand{\hhf}{\ema{\hhff{\ct}}}
\newcommand{\hhfp}{\ema{\hhfu'(\ct)}}

\newcommand{\flofun}[1]{\ema{\left\lfloor #1 \right\rfloor}}
\newcommand{\ceifun}[1]{\ema{\left\lceil #1 \right\rceil}}

\smallskip

\begin{corollary}
\label{th:compar}
The highest possible number of wasted repetitions is $\left\lceil\sqrt{\ct}-1\right\rceil$
and is achieved when $\ct=q+1$.
\end{corollary}
\begin{proof}

The highest possible number of wasted repetitions \wup with \ct threads is given by
\[ \wup = \flo-\fup = \flofun{\frac{1}{2} \left( - \hf + \sqrt{\hf^2 + 4\ct} \right) -\fup} .\]

Let \hfu and \hhfu be the functions respectively defined as $\hf = q+1-\ct+r$, which
implies $\hfp=-1$, and $\hhf = (- \hf + \sqrt{\hf^2 + 4\ct})/2 - \fup$, so that
$\wup = \flofun{\hhf}$.

Let us first assume that $\hf>0$. In this case, $q\leq \ct-1$, hence $\fup = 0$.
We have
\begin{align*}
2\hhfp &= 1 + \frac{-2\hf + 4}{2 \sqrt{\hf^2 + 4\ct}}\\
2\hhfp&= 2 \times \frac{2 - \hf + \sqrt{\hf^2 + 4\ct}}{2 \sqrt{\hf^2 + 4\ct}}
\end{align*}
Therefore, \hhfp is negative if and only if $\sqrt{\hf^2 + 4\ct} < \hf -2$. It cannot be
true if $\hf<2$. If $\hf\geq 2$, then the previous inequality is equivalent to $\hf^2 +
4\ct < \hf^2 - 4\hf +4$, which can be rewritten in $q+1+r<1$, which is absurd. We have
shown that \hhfu is increasing in $]0,q+1]$.

Let us now assume that $\hf\leq0$. In this case, $q > \ct-1$, hence $\fup = \ct -q -1$,
and $\hhfp = \left( \hf + \sqrt{\hf^2 + 4\ct} \right) /2 - r$. Assuming \hhfp to be
positive leads to the same absurd inequality $q+1+r<1$, which proves that \hhfu is
decreasing on $[q+2,+\infty[$.

Also, the maximum number of wasted repetitions is achieved as $\ct=q+1$ or $\ct=q+2$.
Since
\[ \hhff{q+1} = \frac{1}{2} \left( -r + \sqrt{r^2 + 4\ct} \right) >
\frac{1}{2} \left(-(r+1) + \sqrt{r^2 + 4\ct}\right) = \hhff{q+2}, \]
the maximum number of wasted repetitions is \wupf{q+1}.
In addition,
\begin{equation*}
\bgroup
\def\arraystretch{1.7}\begin{array}{rcccl}
\dfrac{1}{2} \left( -r + \sqrt{4\ct} \right) &<& \hhff{q+1} &<&
\dfrac{1}{2} \left( -r + \sqrt{r^2} + \sqrt{4\ct} \right)\\
\sqrt{\ct}-\dfrac{r}{2} &<& \hhff{q+1} &<&
 \sqrt{\ct}\\
\sqrt{\ct}-1 &\leq& \hhff{q+1} &<&
 \sqrt{\ct}\\
\end{array}
\egroup
\end{equation*}
We conclude that the maximum number of wasted repetitions is \ceifun{\sqrt{\ct}-1}.
\end{proof}
}

\section{Expansion and Complete Throughput Estimation}
\label{sec:exp-glue}

\subsection{Expansion}
\label{sec:exp}

Interference of threads does not only lead to \logcons
but also to \hardcons
which impact the performance significantly.
We model the behavior of the cache coherency
protocols which determine the interaction of overlapping \rf{}s and \cas{}s.
By taking MESIF~\cite{mesif} as basis, we come up
with the following assumptions.
When executing an atomic \cas, the core gets the cache line in exclusive state
and does not forward it to any other requesting core until the instruction
is retired. Therefore, requests stall for the release of the cache line
which implies serialization. On the other hand, ongoing \rf{}s can overlap with
other operations. As a result,
a \cas introduces expansion only
to overlapping \rf and \cas operations that start after it, as illustrated in Figure~\ref{fig:expansion}.
\rr{As a remark, we ignore memory bandwidth issues which are negligible for our
study.}

Furthermore, we assume that \rf{}s that are executed just after a \cas do not experience
expansion (as the thread already owns of the data), which takes effect at the beginning of a \re following a failing attempt.
Thus, read expansions need only to be considered before the \kth{0} \re. In this sense, read
expansion can be moved to parallel section and calculated in the
same way as \cas expansion is calculated.

To estimate expansion, we consider the delay that a thread can introduce,
provided that there is already a given number of threads in the \rl.
The starting point of each \cas is a random
variable which is distributed uniformly within an expanded \re. The cost function \shiftf provides
the amount of delay that the additional thread introduces, depending on
the point where the starting point of its \cas hits. By using this cost function
we can formulate the expansion increase that each new thread introduces and derive the differential equation
below to calculate the expansion of a \cas.

\begin{lemma}
\label{lem.1}
The expansion of a \cas operation is the solution of the following system of equations:
\[ \left\{
\begin{array}{lcl}
\expansionp{\trl} &=& \fcas \times \dfrac{\frac{\fcas}{2} + \expansion{\trl}}{ \mem +  \calrl + \scas + \expansion{\trl}}\\
\expansion{\trlo} &=& 0
\end{array} \right., 
\qquad\begin{array}{l}
\text{where \trlo is the point where}\\
\text{expansion begins.}
\end{array}
\]
\end{lemma}
\begin{proof}
\pp{
To prove the theorem, we compute $\expansion{\trl + h}$, where $h\leq1$, by assuming that there are
already \trl threads in the \rl, and that a new thread attempts
to \cas during the \re, within a probability $h$:
$ \expansion{\trl + h}= \expansion{\trl} + h\times
\rint{0}{\rlsiz}{\frac {\shift{t}}{\rlsiz}}{t}$. 
The complete proof appears in the companion research report~\cite{our-long}.
}

\rr{We compute $\expansion{\trl + h}$, where $h\leq1$, by assuming that there are
already \trl threads in the \rl, and that a new thread attempts
to \cas during the \re, within a probability $h$.

\begin{align*}
\expansion{\trl + h}=
 \expansion{\trl} + h\times&
\rint{0}{\rlsiz}{\frac {\shift{t}}{\rlsiz}}{t} \\
  = \expansion{\trl}
     + h \times &\Big( \rint{0}{\mem+\calrl-\fcas}{\frac{\shift{t}}{\rlsiz}}{t}\\
   & \quad  +
      \rint{\mem+\calrl-\fcas}{\mem+\calrl}{\frac{\shift{t}}{\rlsiz}}{t}\\
    & \quad  + \rint{\mem+\calrl}{\mem+\calrl+\expansion{\trl}}{\frac{\shift{t}}{\rlsiz}}{t}\\
 & \quad + \rint{\mem+\calrl+\expansion{\trl}}{\rlsiz}{\frac{\shift{t}}{\rlsiz}}{t}\Big)\\
  = \expansion{\trl} + h \times & \Big(
    \rint{\mem+\calrl-\fcas}{\mem+\calrl}{\frac{t}{\rlsiz}}{t}\\
    & \quad  + \rint{\mem+\calrl}{\mem+\calrl+\expansion{\trl}}{\frac{\fcas}{\rlsiz}}{t} \Big)\\
\expansion{\trl + h} = \expansion{\trl} + h \times & \frac{ \frac{\fcas^2}{2} + \expansion{\trl}\times\fcas}{\rlsiz}
\end{align*}

This leads to
$\displaystyle \quad\frac{\expansion{\trl + h}- \expansion{\trl}}{ h} = \frac{ \frac{\fcas^2}{2} + \expansion{\trl}\times\fcas}{\rlsiz}$.
When making $h$ tend to $0$, we finally obtain
\[ \expansionp{\trl} = \fcas \times \frac{\frac{\fcas}{2} + \expansion{\trl}}{ \mem +  \calrl + \scas + \expansion{\trl}}. \pp{\qedhere}\]
}
\end{proof}

\subsection{Throughput Estimate}
\label{sec:thr}

\newcommand{\antm}[1]{\ema{h_{#1}(\trl)}}

\newcommand{\antuf}{\ema{h^{\exppl}}}
\newcommand{\antlf}{\ema{h^{\expmi}}}

\newcommand{\antuo}[1]{\ema{h^{\exppl}(#1)}}
\newcommand{\antlo}[1]{\ema{h^{\expmi}(#1)}}

\newcommand{\antu}{\ema{\antuf(\trl)}}
\newcommand{\antl}{\ema{\antlf(\trl)}}

\newcommand{\trlu}{\ema{\trl^{\exppl}}}
\newcommand{\trll}{\ema{\trl^{\expmi}}}

There remains to combine hardware and \logcons
in order to obtain the final
upper and lower bounds on throughput.
We are given as an input an expected number of
threads \trl inside the \rl. We firstly compute the expansion accordingly, by solving
numerically the differential equation of Lemma~\ref{lem.1}.
As explained in the previous subsection, we
have $\psiz = \pw + \expa$, and $\rlsiz = \rc + \cw + \expa + \cc$.
We can then compute $q$ and $r$, that are the inputs (together with the total number of
threads \ct) of the method described in Section~\ref{sec:wt}. Assuming that the
initialization times of the threads are spaced enough, the execution will superimpose
an \cyex{f}{\ct}. Thanks to Lemma~\ref{lem:av-thr}, we can compute the
average number of threads inside the \rl, that we note by \antm{f}.
A posteriori, the solution is consistent if this average number of threads inside the \rl
\antm{f} is equal to the expected number of threads \trl that has been given as an input.

Several \cyexs{f}{\ct} belong to the domain of the possible outcomes, but we are
interested in upper and lower bounds on the number of failures $f$. We can compute them
through \rr{Lemmas~\ref{lem:wt-upp} and~\ref{lem:wt-low}}\pp{Lemma~\ref{lem:wt-bounds}}, along with their corresponding
throughput and average number of threads inside the \rl. We note by \antu and \antl the
average number of threads for the lowest number of failures and highest one,
respectively. Our aim is finally to find \trll and \trlu, such that $\antuo{\trlu} =
\trlu$ and $\antlo{\trll} = \trll$. If several solutions exist, then we want to keep the
smallest, since the \rl stops to expand when a stable state is reached.

Note that we also need to provide the point where the expansion begins. It begins when we
start to have failures, while reducing the \ps. Thus this point is $(2 \ct -1 ) \rlwp$
(resp. $(\ct -1) \rlwp$) for the lower (resp. upper) bound on the throughput.

\begin{theorem}
Let $(x_n)$ be the sequence defined recursively by $x_0=0$ and $x_{n+1} = \antuo{x_n}$. If
$\pw\geq\rc+\cw+\cc$, then
\rr{\[ \trlu = \lim_{n \rightarrow +\infty} x_n. \]}
\pp{$ \trlu = \lim_{n \rightarrow +\infty} x_n$. }
\end{theorem}
\begin{proof}
\pp{In~\cite{our-long}, we prove that \antuf is non-decreasing when $\pw\geq\rc+\cw+\cc$, and obtain
the above theorem by applying the Theorem of Knaster-Tarski.}
\rr{First of all, the average number of threads belongs to $]0,\ct[$, thus for all $x \in
  [0,\ct]$, $0<\antuo{x}<\ct$. In particular, we have $\antuo{0}>0$, and
  $\antuo{\ct}<\ct$, which proves that there exist one fixed point for \antuf.

In addition, we show that \antuf is a non-decreasing function. According to Lemma~\ref{lem:av-thr},
\[ \antu = \ct \times \frac{1+\fup}{q+r+\fup+1}, \]
where all variables except \ct depend actually on \trl. We have
\[ q=\flofun{\frac{\pw + \expa}{\rlwp + \expa}} \text{  and  }
r = \frac{\pw + \expa}{\rlwp + \expa} - q,\]
hence, if $\pw\geq\rlwp$, $q$ and $r$ are non-increasing as \expa is non-decreasing,
which is non-decreasing with \trl. Since \fup is non-decreasing as a function of $q$, we
have shown that if $\pw\geq\rlwp$, \antuf is a non-decreasing function.

Finally, the proof is completed by the theorem of Knaster-Tarski.}
\end{proof}
The same line of reasoning holds for \antlf\rr{ as well}.
\rr{As a remark, w}\pp{We} point out that when $\pw<\rlwp$, we scan the interval of solution, and
have no guarantees about the fact that the solution is the smallest one; still it
corresponds to very extreme cases.

\rr{
\subsection{Several \RLs}
\label{sec:sev-rl}

We consider here a lock-free algorithm that, instead of being a loop
over one \ps and one \rl, is composed of a loop over a sequence of
alternating \pss and \rls.  We show that this algorithm is equivalent
to an algorithm with only one \ps and one \rl, by proving the
intuition that the longest \rl is the only one that fails and hence
expands.

\subsubsection{Problem Formulation}

\newcommand{\maxs}{\ema{S}}

\begin{figure}[b!]
\removelatexerror
\begin{procedure}[H]
\SetKwData{pet}{execution\_time}
\SetKwData{pdo}{done}
\SetKwData{psucc}{success}
\SetKwData{pcur}{current}
\SetKwData{pnew}{new}
\SetKwData{pacp}{AP}
\SetKwData{pttot}{t}
\SetKwFunction{pinit}{Initialization}
\SetKwFunction{ppw}{Parallel\_Work}
\SetKwFunction{pcw}{Critical\_Work}
\SetKwFunction{pread}{Read}
\SetKwFunction{pcas}{CAS}
\SetKwData{pmaxs}{S}

\SetAlgoLined
\pinit{}\;
\While{! \pdo}{
\For{i $\leftarrow$ 1 \KwTo \pmaxs}{
\ppw{i}\;\nllabel{alg:lis-ps}
\While{! \psucc}{\nllabel{alg:lis-bcs}
\pcur $\leftarrow$ \pread{$\pacp[i]$}\;
\pnew $\leftarrow$ \pcw{i,\pcur}\;
\psucc $\leftarrow$ \pcas{\pacp, \pcur, \pnew}\;
}
}\nllabel{alg:lis-ecs}
}
\caption{Combined()\label{alg:gen-name-s}}
\end{procedure}
\caption{Thread procedure with several \rls\label{alg:gen-nb-s}}
\end{figure}

In this subsection, we consider an execution such that each spawned thread runs
Procedure~\ref{alg:gen-name-s} in Figure~\ref{alg:gen-nb-s}. Each thread executes a linear
combination of \maxs independent \rls, \ie operating on separate variables, interleaved
with \pss. We note now as $\rlsiz_i$ and $\psiz_i$ the size of a \re of the \kth{i} \rl and the size of
the \kth{i} \ps, respectively, for each $i \in \inte[1]{\maxs}$. As previously, $q_i$ and
$r_i$ are defined such that $\psiz_i = (q_i + r_i) \times \rlsiz_i$, where $q_i$ is a
non-negative integer and $r_i$ is smaller than $1$.

The Procedure~\ref{alg:gen-name-s} executes the \rls and \pss in a cyclic fashion, so we
can normalize the writing of this procedure by assuming that a \re of the \kst{1} \rl is the
longest one. More precisely, we consider the initial algorithm, and we define $i_0$ as
\[ i_0 = \min \operatorname{argmax}_{i \in \inte[1]{\maxs}}  \rlsiz_i. \]
We then renumber the \rls such that the new ordering is $i_0, \dots, \maxs, 1, \dots,
i_0-1$, and we add in \pinit the first \pss and \rls on access points from $1$ to $i_0$
--- according to the initial ordering.

One success at the system level is defined as one success of the last \cas, and the
throughput is defined accordingly. We note that in steady-state, all \rls have the same
throughput, so the throughput can be computed from the throughput of the \kst{1} \rl
instead.

\subsubsection{Wasted \REs}

\newcommand{\seqt}[1]{\ema{t_{#1}}}
\newcommand{\seqs}[1]{\ema{s_{#1}}}

\begin{lemma}
\label{lem:no-fail}
Unsuccessful \rls can only occur in the \kst{1} \rl.
\end{lemma}
\begin{proof}

We note $(\seqt{n})_{n \in [1,+\infty[}$ the sequence of the thread numbers that succeeds
    in the \kst{1} \rl, and $(\seqs{n})_{n \in [1,+\infty[}$ the sequence of the
        corresponding time where they exit the \rl. We notice that by construction, for
        all $n \in [1,+\infty[$, $\seqs{n}<\seqs{n+1}$. Let, for $i \in \inte[2]{\maxs}$
            and $n \in [1,+\infty[$, \pro{i,n} be the following property: for all $i' \in
                \inte[2]{i}$, and for all $n' \in \inte[1]{n}$, the thread \thr{\seqt{n'}}
                succeeds in the \kth{i} \rl at its first attempt.

We assume that for a given $(i,n)$, \pro{i+1,n} and \pro{i,n+1} is true, and show that
\pro{i+1,n+1} is true. As the threads \thr{\seqt{n}} and \thr{\seqt{n+1}} do not have any
failure in the first $i$ \rls, their entrance time in the \kth{i+1} \rl is given by
\[ \seqs{n} + \sum_{i'=1}^{i} ( \rlsiz_{i'} + \psiz_{i'} ) + \psiz_{i+1} = X_1 \text{  and  }
\seqs{n+1} + \sum_{i'=1}^{i} ( \rlsiz_{i'} + \psiz_{i'} ) + \psiz_{i+1} = X_2, \]
respectively. Thread \thr{\seqt{n}} does not fail in the \kth{i+1} \rl, hence exits at
\[X_1 + \rlsiz_{i+1} < X_1 + \rlsiz_1 = \seqs{n} + X_2 - \seqs{n+1} < X_2. \]
As the previous threads $\thr{n-1}, \dots, \thr{1}$ exits the \kth{i} \rl before \thr{n},
and next threads \thr{n'}, where $n'>n+1$, enters this \rl after \thr{n+1}, this implies
that the thread \thr{\seqt{n+1}} succeeds in the \kth{i+1} \rl at its first attempt, and
\pro{i+1,n+1} is true.

Regarding the first thread that succeeds in the first \rl, we know that he successes in
any \rl since there is no other thread to compete with. Therefore, for all $i \in
\inte[2]{\maxs}$, \pro{i,1} is true. Then we show by induction that all \pro{2,n} is true,
then all \pro{3,n}, \etc, until all \pro{\maxs,n}, which concludes the proof.
\end{proof}

\begin{theorem}
\label{th:mult-eq}
The multi-\rl Procedure~\ref{alg:gen-name-s} is equivalent to the
Procedure~\ref{alg:gen-name}, where
\[ \psiz = \psiz_1 + \sum_{i=2}^\maxs \left( \psiz_i + \rlsiz_i \right) \quad \text{and}\quad
\rlsiz = \rlsiz_1.\]
\end{theorem}
\begin{proof}
According to Lemma~\ref{lem:no-fail} there is no failure in other \rl than the first one;
therefore, all \rls have a constant duration, and can thus be considered as \pss.
\end{proof}

\subsubsection{Expansion}
\label{sec:sev-exp}

The expansion in the \rl starts as threads fail inside this \rl. When threads are launched,
there is no expansion, and Lemma~\ref{lem:no-fail} implies that if threads fail, it
should be inside the first \rl, because it is the longest one. As a result, there will be
some stall time in the memory accesses of this first \rl, \ie expansion, and it will get
even longer. Failures will thus still occur in the first \rl: there is a positive feedback
on the expansion of the first \rl that keeps this first \rl as the longest one among all
\rls. Therefore, in accordance to Theorem~\ref{th:mult-eq}, we can compute the expansion
by considering the equivalent single-\rl procedure described in the theorem.

}

\pp{
}\section{Experimental Evaluation}
\label{sec:xp}

We validate our model and analysis framework through successive steps, from
synthetic tests, capturing a wide range of possible
abstract algorithmic designs,
to several reference implementations of
extensively studied lock-free \ds
designs that include cases with non-constant \ps and \rl.
\pp{The complete results can be found in~\cite{our-long} and the numerical simulation code in~\cite{our-disc15-code}.}

\pp{
}\subsection{Setting}

We have conducted experiments on an Intel ccNUMA workstation system. The system
is composed of two sockets, that is equipped with Intel Xeon E5-2687W v2 
CPUs\pp{.}\rr{ with frequency band \ghz{1.2\text{-}3.4.} The physical cores
have private L1, L2 caches and they share an L3 cache, which is \megb{25}.} 
In a socket, the ring interconnect provides L3 cache accesses and core-to-core 
communication. \rr{Due to the bi-directionality of the ring
interconnect, uncontended latencies for intra-socket communication between cores
do not show significant variability.}\pp{Threads are pinned to a single socket to minimize non-uniformity in \rf and \cas latencies.
  Due to the bi-directionality of the ring that interconnects L3
 caches, uncontended latencies for intra-socket communication between cores do not
 show significant variability.}
\pp{The methodology in~\cite{david-emp-atom} is used to
  measure the \cas and \rf latencies, while the work inside the \ps is
  implemented by a for-loop of {\it Pause} instructions.}

\rr{Our model assumes uniformity in the \cas and \rf latencies on the shared
cache line. Thus, threads
are pinned to a single socket to minimize non-uniformity in \rf and \cas latencies.
In the experiments, we vary the number of threads between 4 and 8
since the maximum number of threads that can be used in the experiments are
bounded by the number of physical cores that reside in one socket.}

In all figures, y-axis provides the throughput,\rr{ which is the
  number of successful operations completed per millisecond. Parallel
  work}\pp{ while the parallel work} is represented in x-axis in
cycles. The graphs contain the
high and low estimates (see Section~\ref{sec:wt}), corresponding to the lower and upper bound on the \cacas,
respectively,
and an additional curve that shows the average of them.

\rr{
As mentioned before, the latencies of \cas and \rf are parameters of our model. We used
the methodology described in~\cite{david-emp-atom} to measure
latencies of these operations in a benchmark program by using two threads that
are pinned to the same socket. The aim is to bring the cache line into the
state used in our model. Our assumption is that the \rf is
conducted on an invalid line. For \cas, the state of the cache line could be
exclusive, forward, shared or invalid. Regardless of the state of the cache line,
\cas requests it for ownership, that compels invalidation
in other cores, which in turn incurs a two-way communication and a memory fence afterwards
to assure atomicity. Thus, the latency of \cas does not show negligible variability
with respect to the state of the cache line, as also revealed in our latency benchmarks.

As for the computation cost, the work inside the \ps is implemented by
a dummy for-loop of {\it Pause} instructions.

}

\newcommand\fuckspaa{}

\pp{
}\subsection{Synthetic Tests}
\label{sec:synt}

\rr{
\subsubsection{Single \rl}
\label{sec:synt-rls}
}
\label{sec:synt-rl}

\rr{\begin{figure}[h!]
\begin{center}
\includegraphics[width=\textwidth]{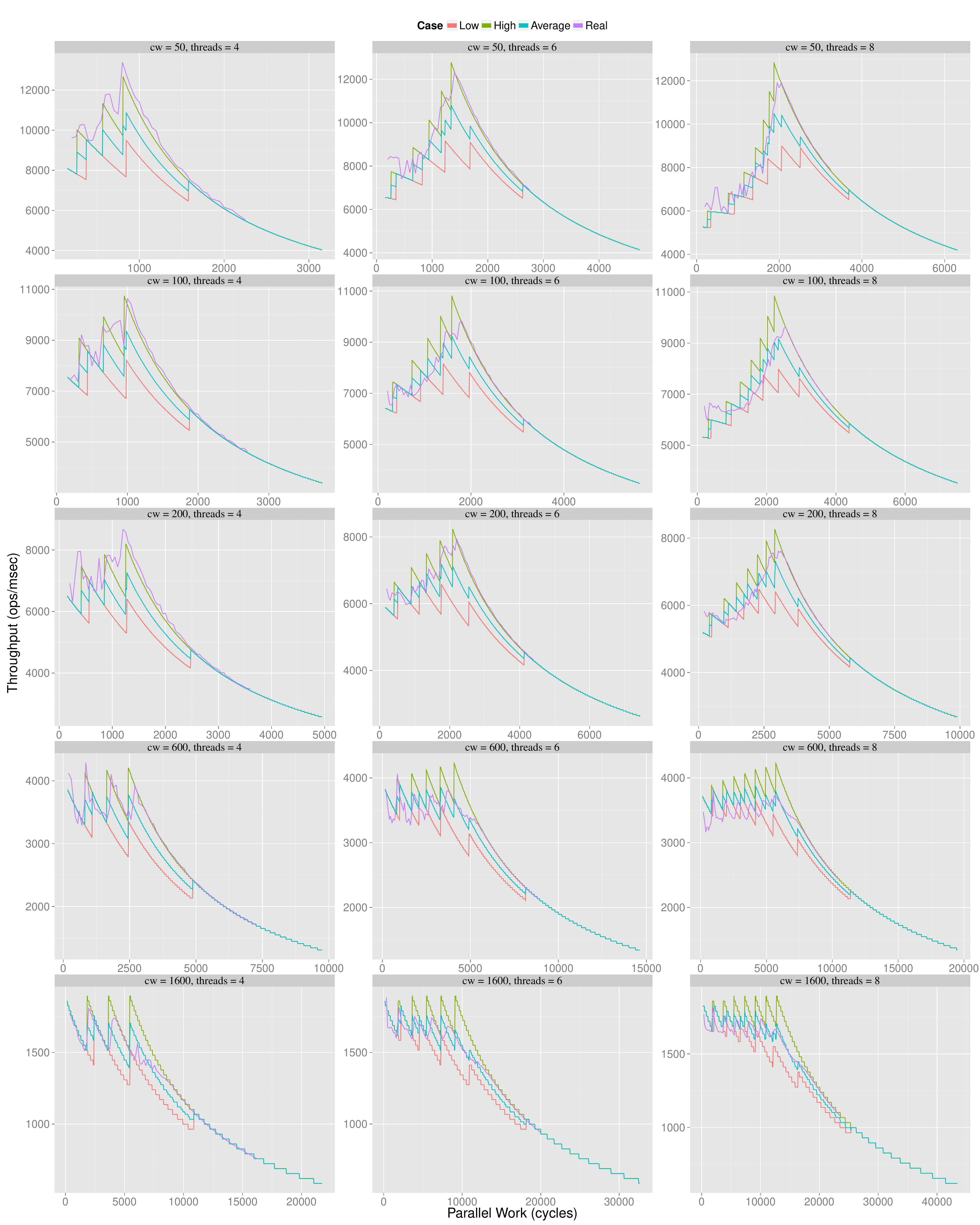}
\end{center}
\caption{Synthetic program\label{fig:synt_smart}}
\end{figure}}

\pp{
\setlength{\belowcaptionskip}{0pt}
\setlength{\textfloatsep}{0.4\baselineskip plus 0.2\baselineskip minus 0.2\baselineskip}
\begin{figure}[b!]
\begin{center}
\begin{minipage}[t]{.68\textwidth}
\includegraphics[width=\textwidth]{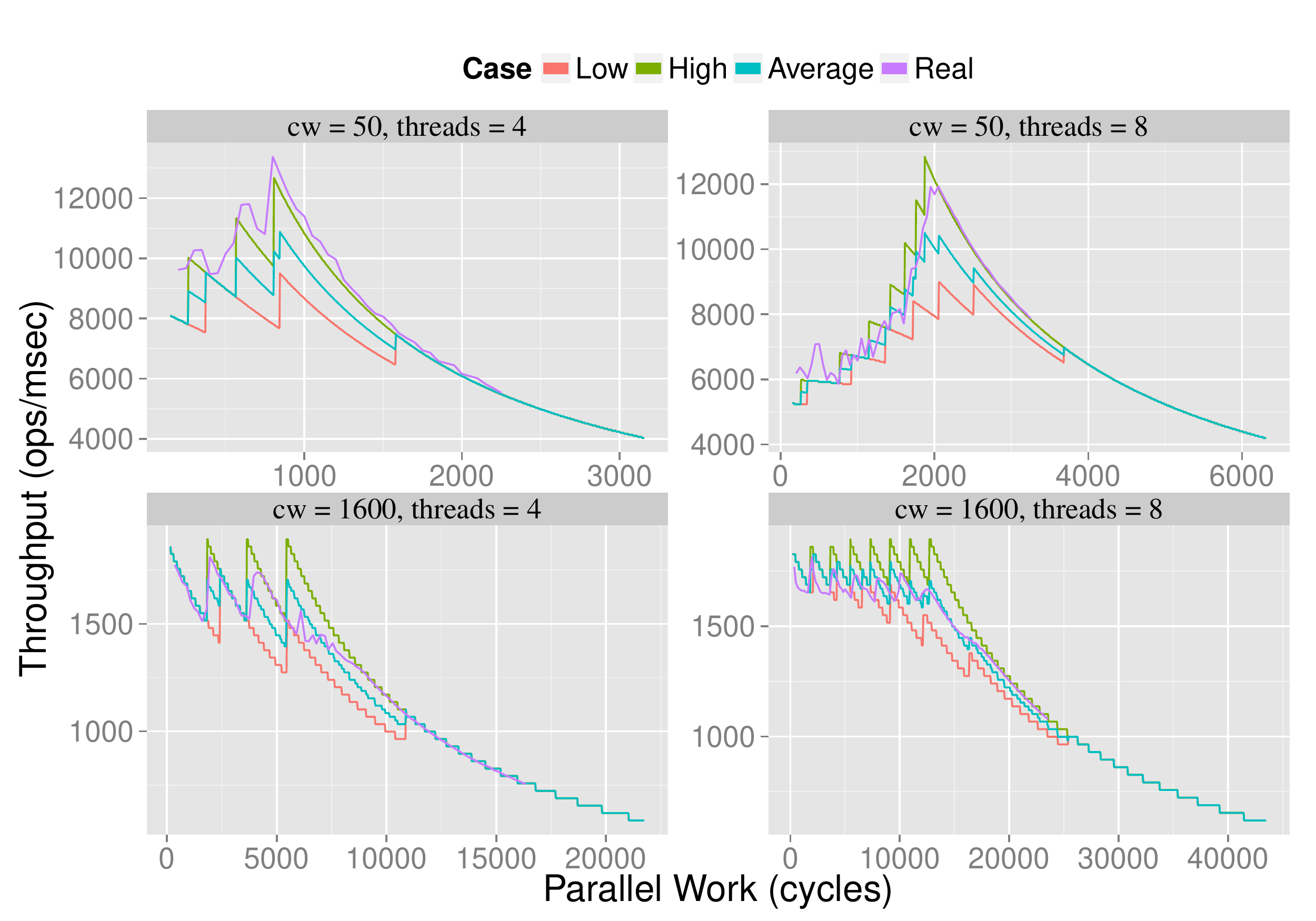}
\caption{Synthetic program\label{fig:synt_smart}}
\end{minipage}\hfill\begin{minipage}[t]{.3\textwidth}
\fuckspaa\includegraphics[width=\textwidth]{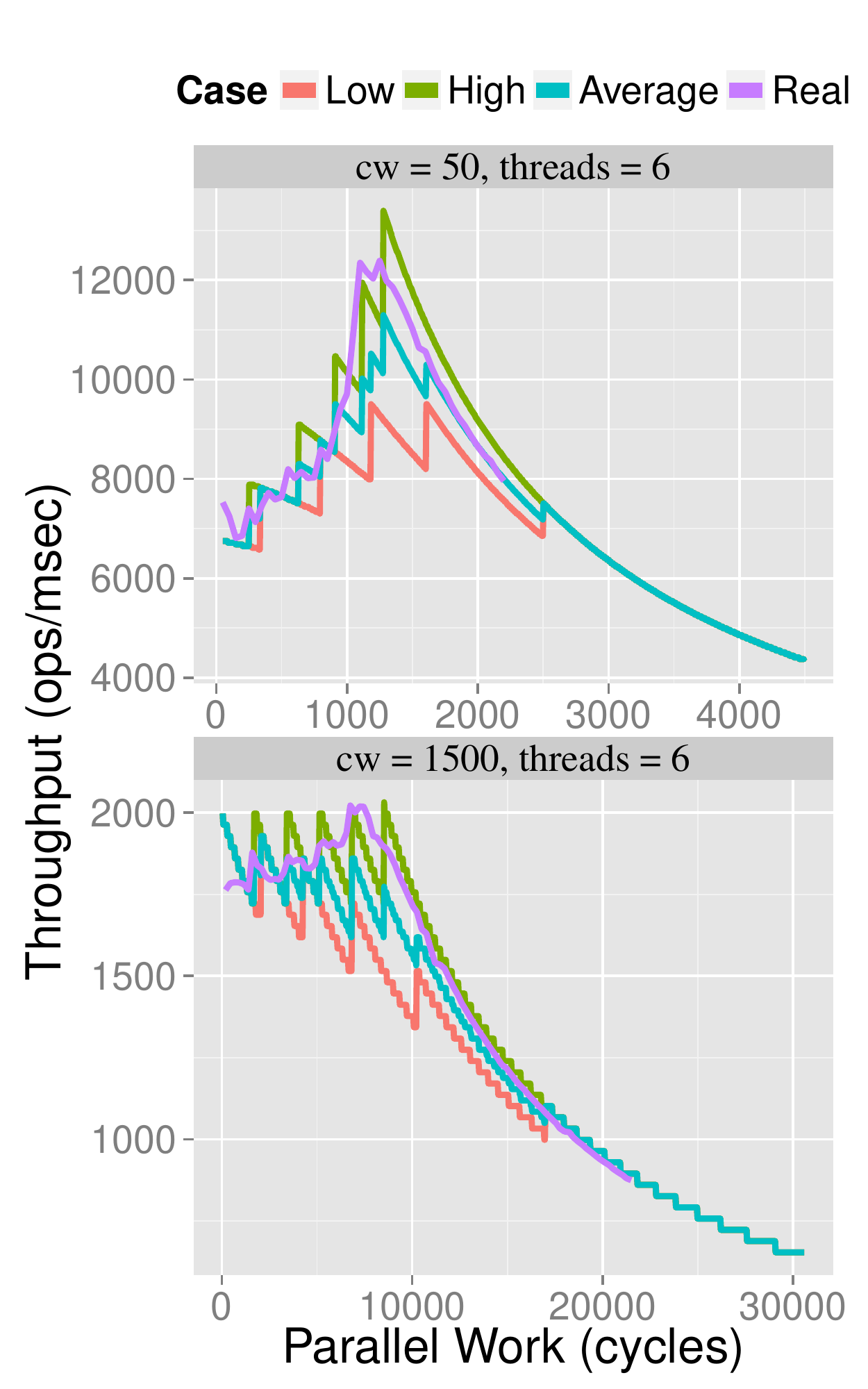}
\caption{\popop on stack\label{fig:treiber}}
\end{minipage}\end{center}
\end{figure}}

For the evaluation of our model, we first create synthetic tests that emulate different
design patterns of lock-free data structures (value of \cw) and different application
contexts (value of \pw).
\rr{As described in the previous subsection, in the
  Procedure~\ref{alg:gen-name}, the amount of work in both the \ps and
  the \rl are implemented as dummy loops, whose costs are adjusted
  through the number of iterations in the loop.}

Generally speaking, in Figure~\ref{fig:synt_smart}, we observe two main behaviors: when \pw is high,
the \ds is not contended, and threads can operate without
failure. When \pw is low, the \ds is contended, and depending on the
size of \cw (that drives the expansion)  a steep decrease
in throughput or just a roughly constant bound on the performance is observed.

The position of the experimental curve between the high and low
estimates, depends on \calrl. It can be observed that the experimental
curve mostly tends upwards as \calrl gets smaller, possibly because
the serialization of the \cas{}s helps the synchronization of the
threads.

\posrem{For the cases with considerable expansion, it is expected to have
unfairness among threads.
This fact loosens the validity of our deterministic model that assumes
uniformity and presumably leads to underestimation of throughput.}{}

Another interesting fact is the waves appearing on the experimental
curve, especially when the number of threads is low or the critical work
big. This behavior is originating because of the variation of $r$
with the change of parallel work, a fact that is captured by our analysis.

\rr{
\subsubsection{Several \rls}
\label{sec:synt-rls}

\begin{figure}[h!]
\begin{center}
\includegraphics[width=.95\textwidth]{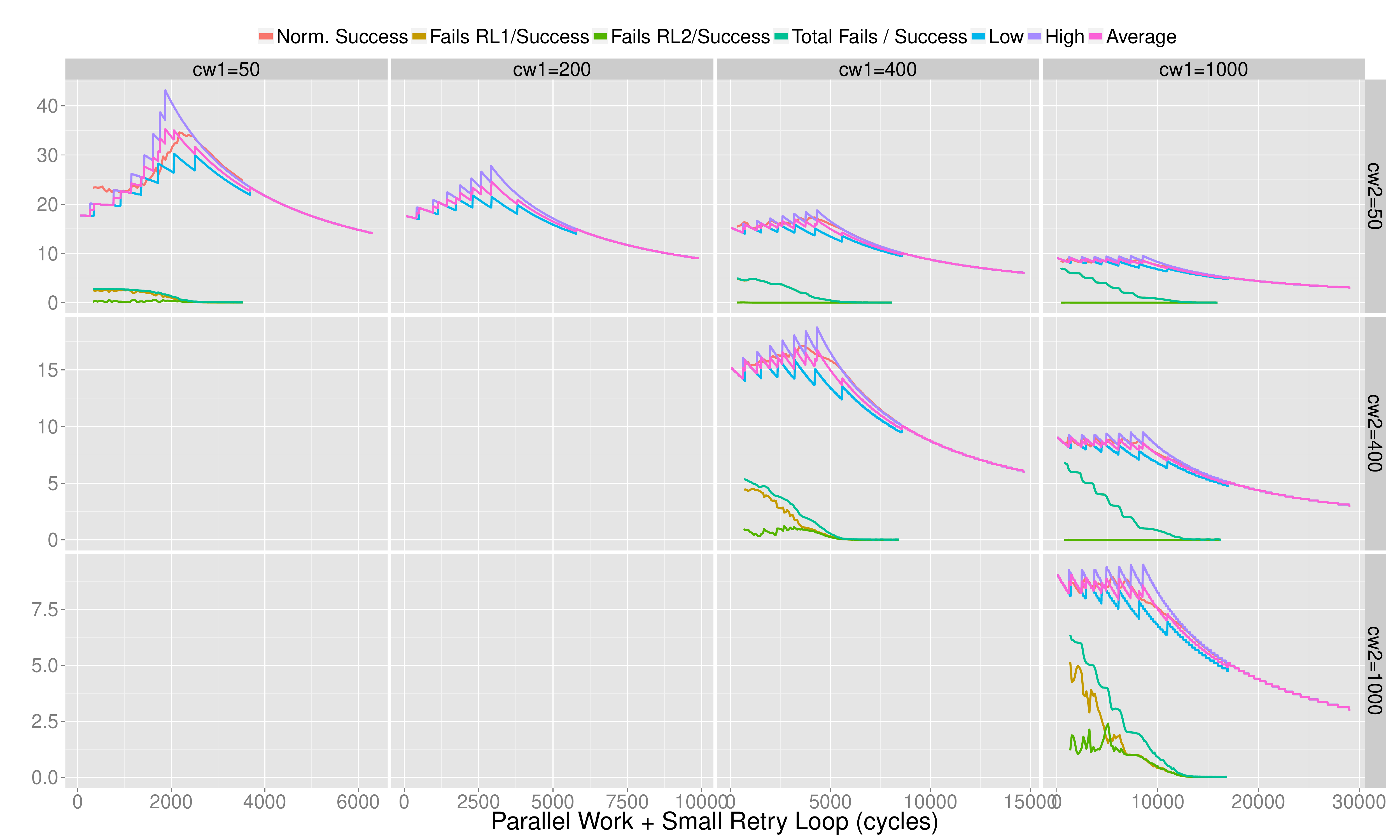}
\end{center}
\caption{Multiple \rls with 8 threads}
\label{fig:multiple_retry}
\end{figure}

We have created experiments by combining several \rls, each operating on an independent variable
which is aligned to a cache line. In Figure~\ref{fig:multiple_retry}, results are compared
with the model for single \rl case where the single \rl is equal to the longest \rl, while
the other \rls are part of the \ps.
The distribution of fails in the \rls are illustrated and all throughput curves are normalized
with a factor of 175 (to be easily seen in the same graph). Fails per success values are
not normalized and a success is obtained after completing all \rls.
}

\subsection{Treiber's Stack}
\label{sec:trei}

\rr{\begin{figure}[h!]
\begin{center}
\includegraphics[width=.95\textwidth]{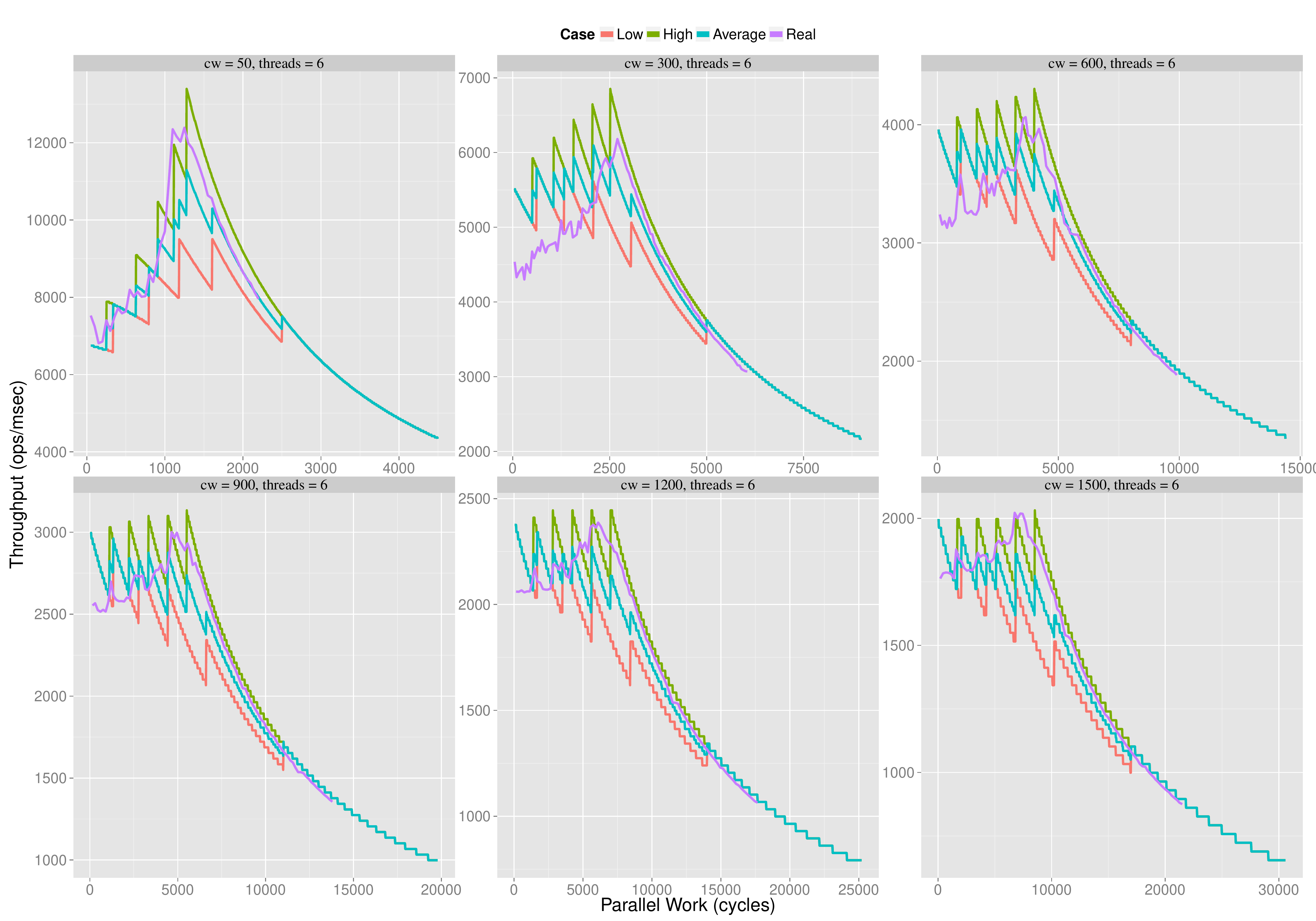}
\end{center}
\caption{\popop on Treiber's stack}
\label{fig:treiber}
\end{figure}}

\pp{ The lock-free stack by Treiber~\cite{lf-stack} is typically the
  first example that is used to validate a freshly-built model on
  lock-free \dss. A \popop contains a \rl that first reads the top
  pointer and gets the next pointer of the element to obtain the
  address of the second element in the stack, before attempting to
  \cas with the address of the second element. The access to the next
  pointer of the first element occurs between the \rf and the
  \cas. Thus, it represents the work in \calrl. By varying the
  number of elements that are popped at the same time, and the
  cache misses implied by the reads,   we vary the  work in \calrl and obtain the results depicted in
  Figure~\ref{fig:treiber}.}

\rr{The lock-free stack by Treiber~\cite{lf-stack} is one of the most
  studied efficient \dss.
\popop and \pushop both contain a \rl, such that each \re starts with
a \rf and ends with \cas on the shared top pointer. In order to
validate our model, we start by using \popop{}s. From a stack
which is initiated with 50 million elements, threads continuously pop
elements for a given amount of time. We count the total number of pop
operations per millisecond. Each \popop first reads the top
pointer and gets the next pointer of the element to obtain the address
of the second element in the stack, before attempting to \cas with the
address of the second element. The access to the next pointer of the
first element occurs in between the \rf and the \cas. Thus, it
represents the work in \calrl. This memory access can possibly
introduce a costly cache miss depending on the locality of the popped
element.

To validate our model with different \calrl values, we make use of
this costly cache miss possibility. We allocate a contiguous chunk of
memory and align each element to a cache line. Then, we initialize the
stack by pushing elements from contiguous memory either with a single
or large stride to disable the prefetcher.  When we measure the
latency of \calrl in \popop for single and large stride cases, we
obtain the values that are approximately 50 and 300 cycles,
respectively.  As a remark, 300 cycles is the cost of an L3 miss in
our system when it is serviced from the local main memory module. To
create more test cases with larger \calrl, we extended the stack
implementation to pop multiple elements with a single operation.
Thus, each access to the next element could introduce an additional L3
cache miss while popping multiple elements. By doing so, we created
cases in which each thread pops 2, 3, \etc elements, and \calrl goes to
600, 900, \etc cycles, respectively. In Figure~\ref{fig:treiber},
comparison of the experimental results from Treiber's stack and our
model is provided.

As a remark, we did not implemented memory reclamation for our
experiments but one can implement a stack that allows pop and push of
multiple elements with small modifications using hazard
pointers~\cite{hazard}.  Pushing can be implemented in the same way as
single element case.  A \popop requires some modifications for memory
reclamation. It can be implemented by making use of hazard pointers
just by adding the address of the next element to the hazard list
before jumping to it. Also, the validity of top pointer should be
checked after adding the pointer to the hazard list to make sure that
other threads are aware of the newly added hazard pointer. By
repeating this process, a thread can jump through multiple elements
and pop all of them with a \cas at the end.

\begin{algorithm}[t!]
\SetAlgoLined
Pop (multiple)

\While{true} {

    t = Read(top)\;
    \For { multiple } {
      \If {t = NULL} {return EMPTY\;}
      hp* = t\;
      \If {top != t} {break\;}
      hp++\;
      next = t.next\;
    }
    \If {CAS(\&top, t, next)} {break\;}
}
RetireNodes (t, multiple)\;
\caption{Multiple Pop\label{fig:stack}}
\end{algorithm}
}

\rr{
\subsection{Shared Counter}

\def\ratio{.47}

\begin{figure}[h!t]
\begin{center}
\begin{minipage}{\ratio\textwidth}\begin{center}
\includegraphics[width=\textwidth]{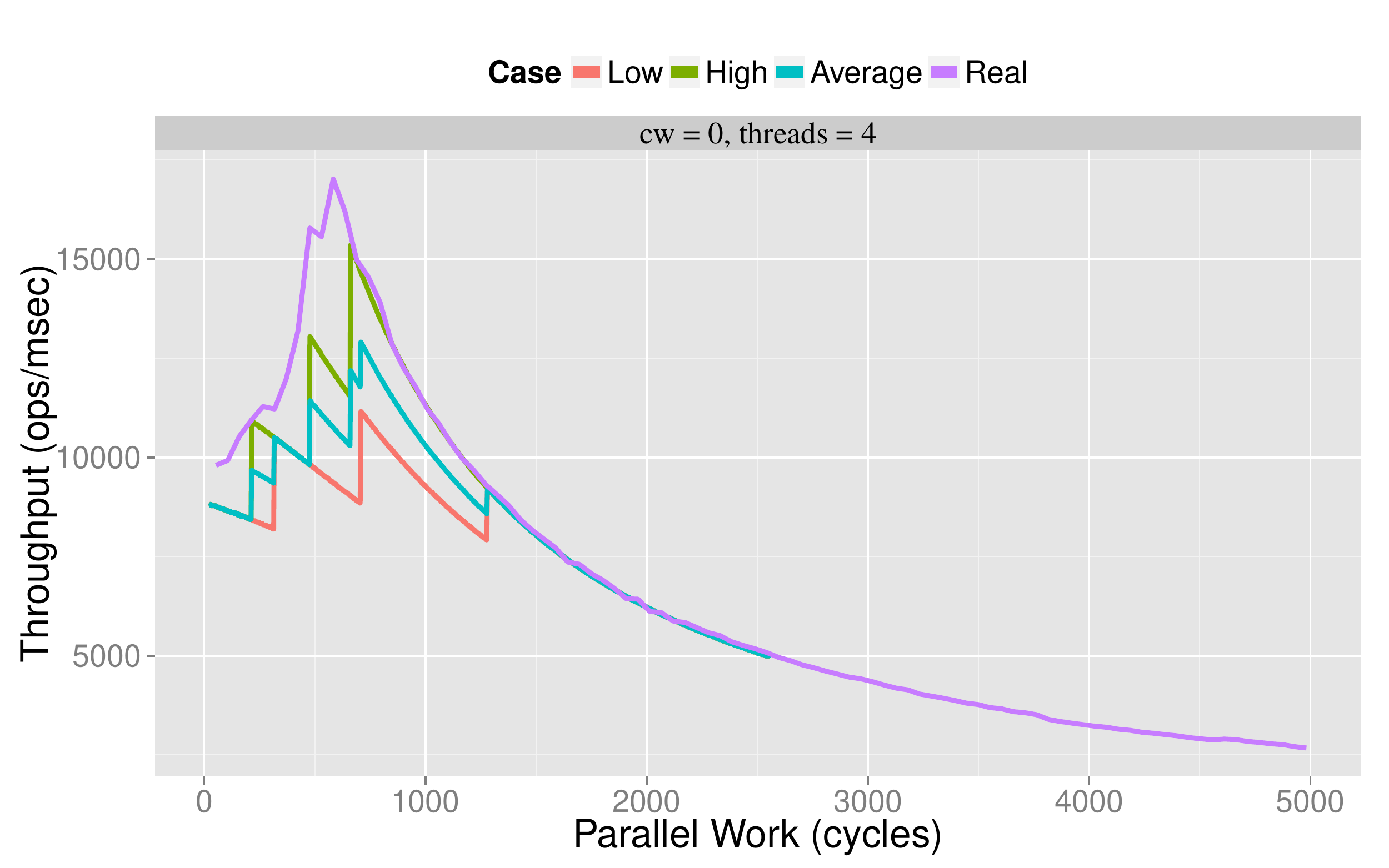}

(a) $4$ threads
\end{center}\end{minipage}\hfill\begin{minipage}{\ratio\textwidth}\begin{center}
\includegraphics[width=\textwidth]{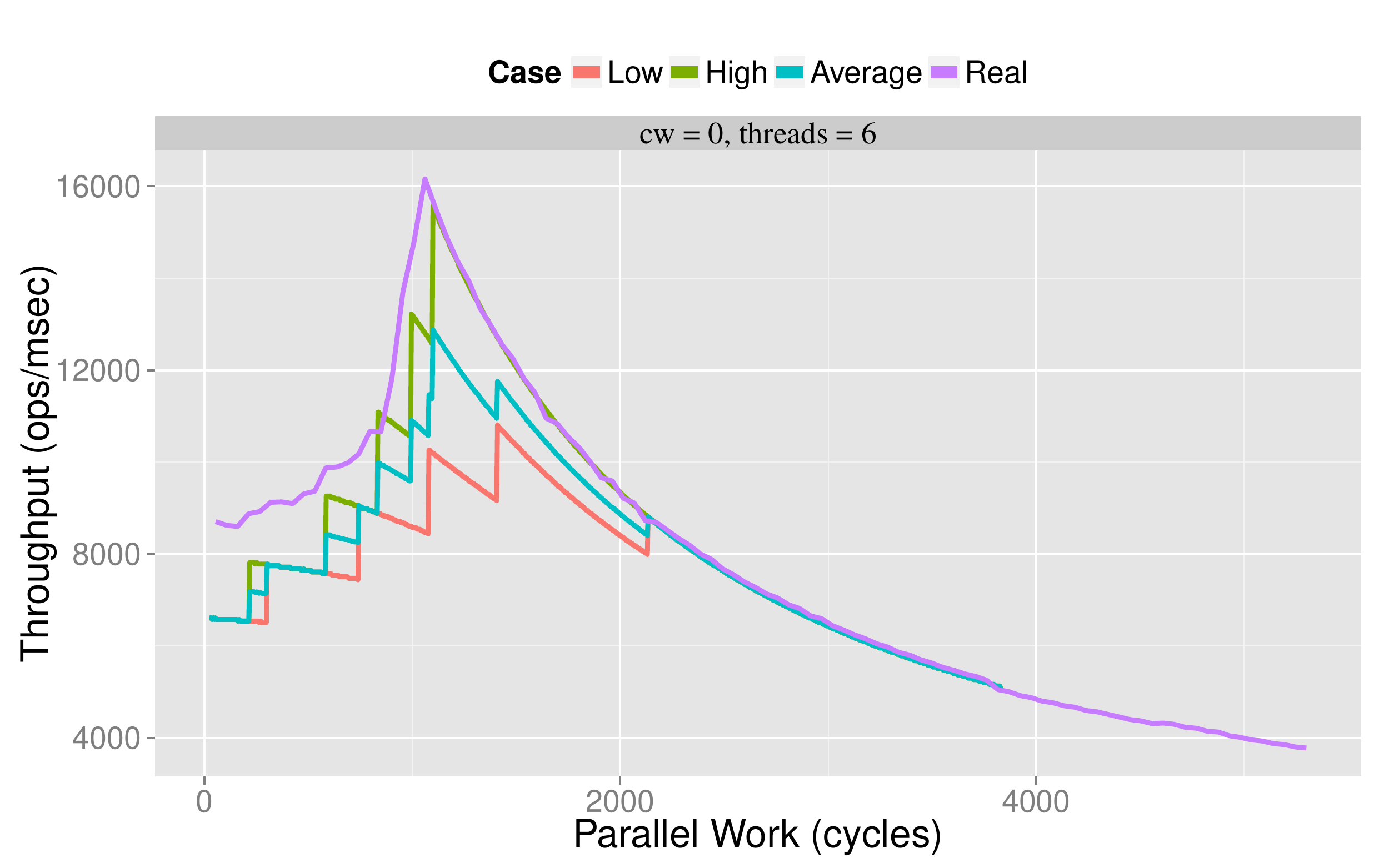}

(b) $6$ threads
\end{center}\end{minipage}

\includegraphics[width=\ratio\textwidth]{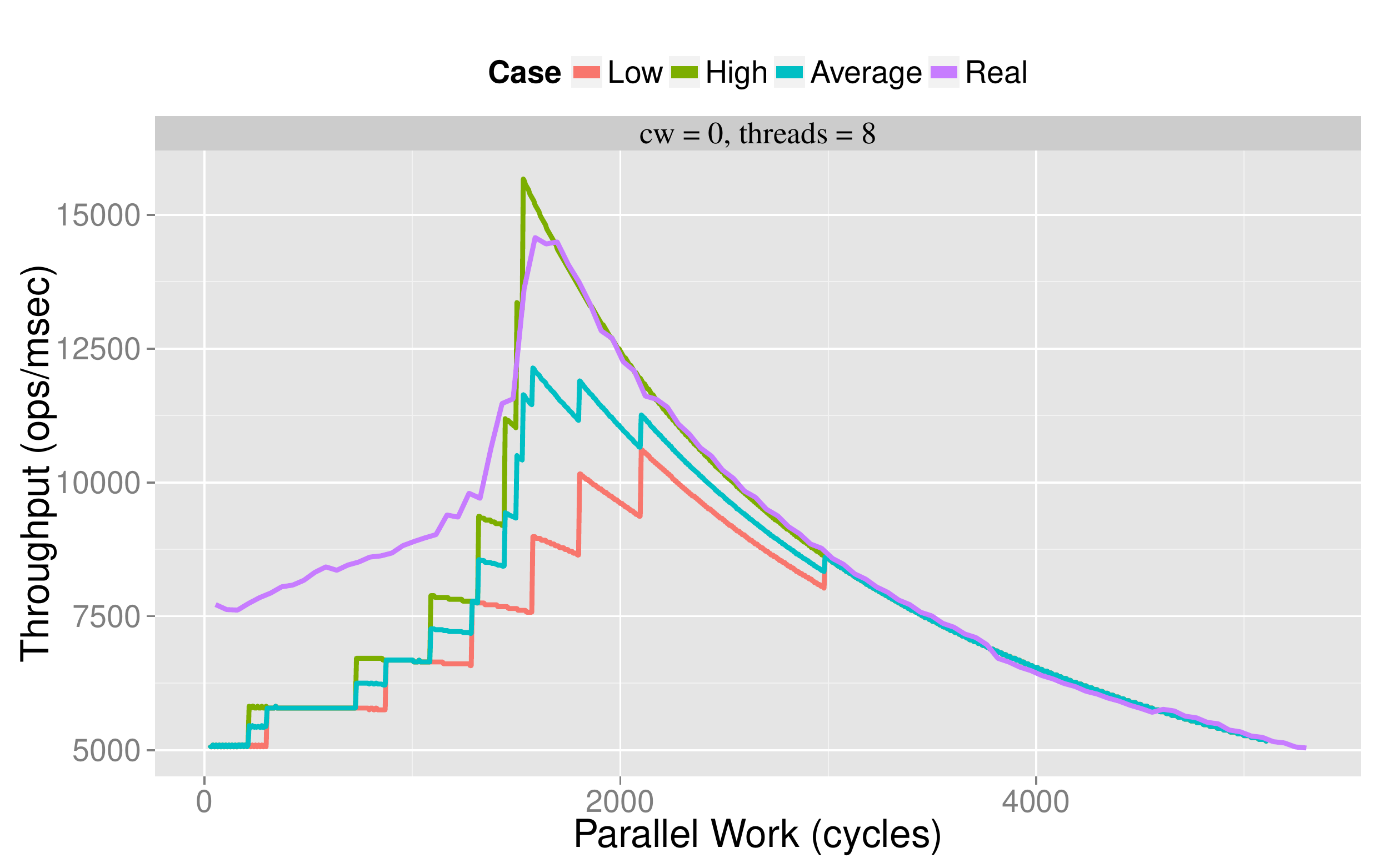}

(c) $8$ threads
\end{center}
\caption{\incop on a shared counter\label{fig:sc}}
\end{figure}

In~\cite{count-moir}, the authors have implemented a ``scalable
statistics counters'' relying on the following idea: when contention
is low, the implementation is a regular concurrent counter with a
\cas; when the counter starts to be contended, it switches to a
statistical implementation, where the counter is actually incremented
less frequently, but by a higher value. One key point of this
algorithm is the switch point, which is decided thanks to the number
of failed increments; our model can be used by providing the peak
point of performance of the regular counter implementation as the
switch point.
We then have implemented a shared counter which is basically a \faa
using a \cas, and compared it with our analysis. The result is
illustrated in Figure~\ref{fig:sc}, and shows that the \ps size corresponding to the peak point is correctly
estimated using our analysis.

\subsection{\delmin in Priority List}

\begin{figure}[h!t]
\begin{center}
\begin{minipage}{\ratio\textwidth}\begin{center}
\includegraphics[width=\textwidth]{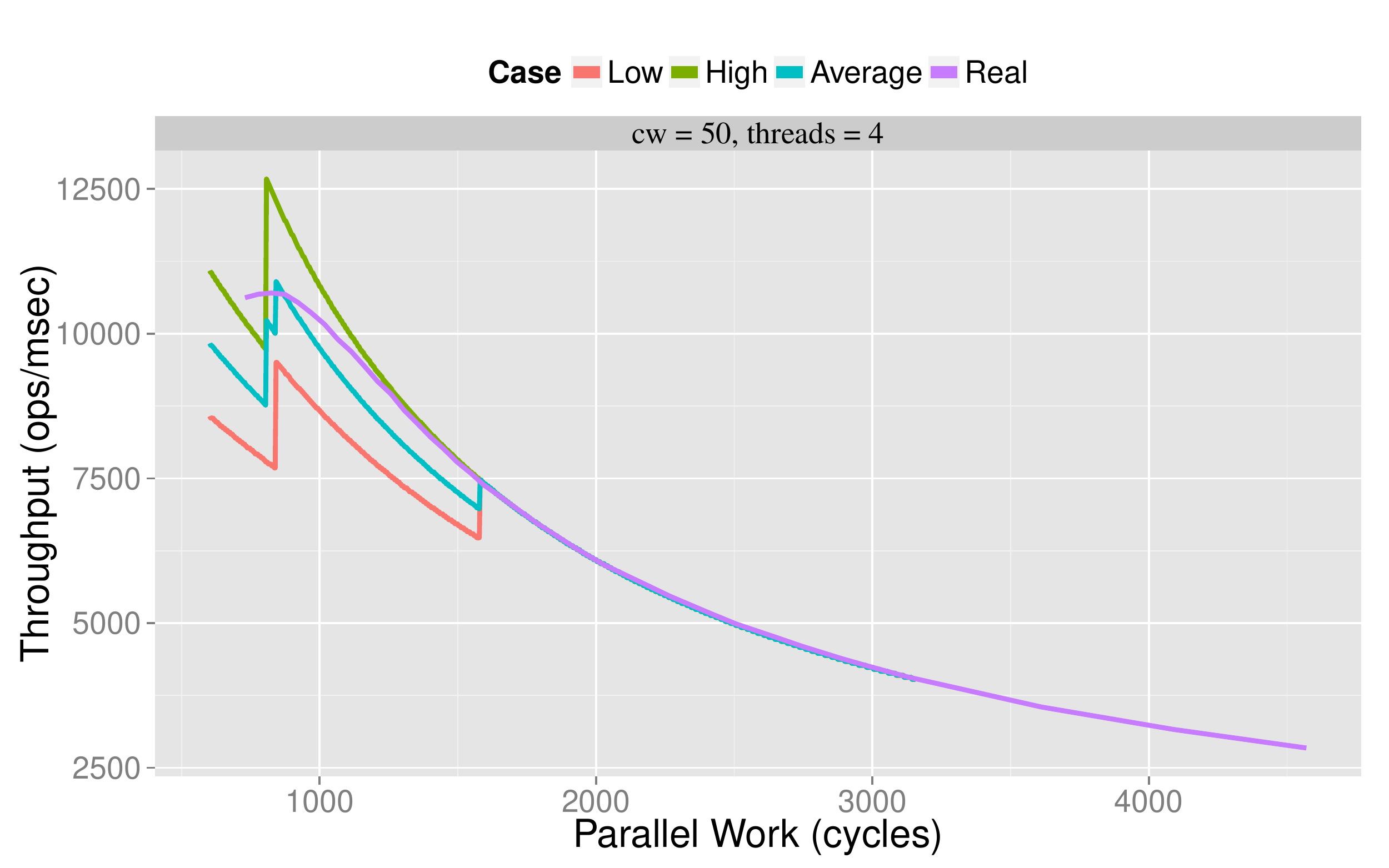}

(a) $4$ threads
\end{center}\end{minipage}\hfill\begin{minipage}{\ratio\textwidth}\begin{center}
\includegraphics[width=\textwidth]{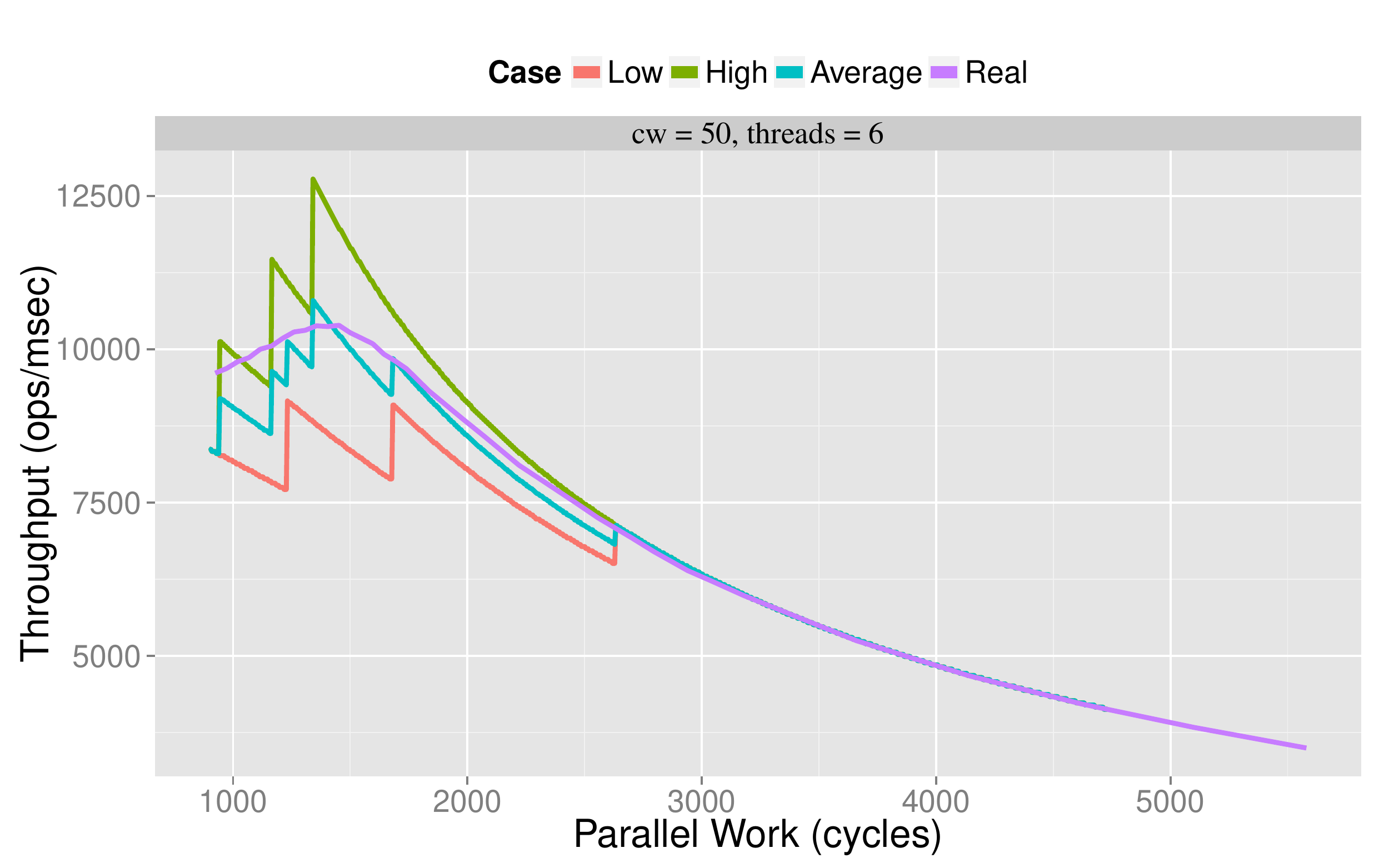}

(b) $6$ threads
\end{center}\end{minipage}

\includegraphics[width=\ratio\textwidth]{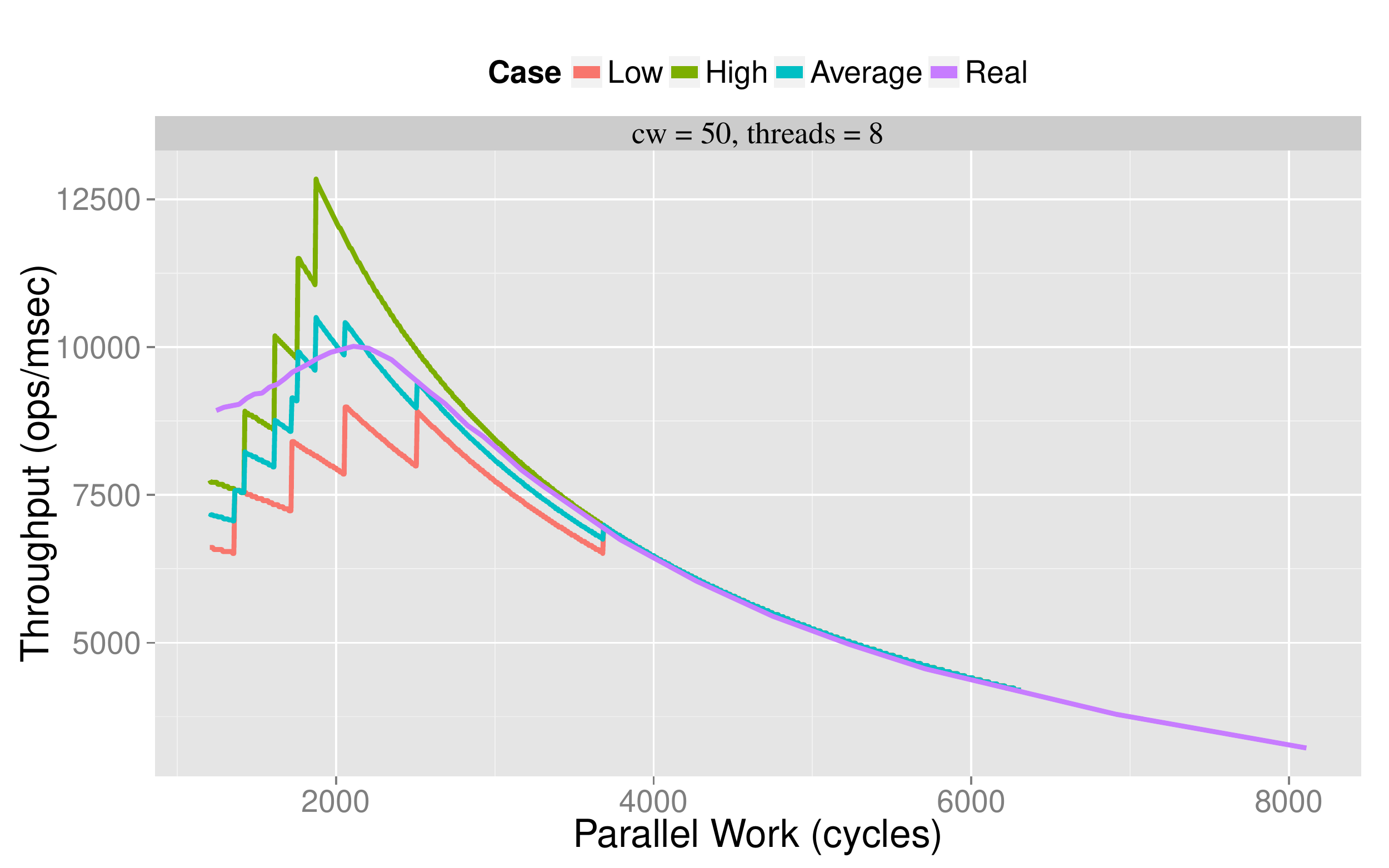}

(c) $8$ threads
\end{center}
\caption{\delmin on a priority list \label{fig:pq}}
\end{figure}

We have applied our model to \delmin of the skiplist based priority queue
designed in~\cite{upp-prio-que}. \delmin traverses the list from the beginning
of the lowest level, finds the first node that is not logically deleted, and
tries to delete it by marking. If the operation does not succeed, it continues
with the next node. Physical removal is done in batches when reaching a
threshold on the number of deleted prefixes, and is followed by a restructuring
of the list by updating the higher level pointers, which is conducted by the
thread that is successful in redirecting the head to the node deleted by itself.

We consider the last link traversal before the logical deletion as critical
work, as it continues with the next node in case of failure. The rest of
the traversal is attributed to the \ps as the threads can proceed concurrently without
interference. We measured the average cost of a traversal under low contention
  for each number of threads, since traversal becomes expensive with more
  threads. In addition, average cost of restructuring is also included in the
\ps since it is executed infrequently by a single thread.

We initialize the priority queue with a large set of elements. As illustrated in
Figure~\ref{fig:pq}, the smallest \pw value is not zero as the average cost of traversal
and restructuring is intrinsically included. The peak point is in the estimated
place but the curve does not go down sharply under high contention. This
presumably occurs as the traversal might require more than one steps (link
access) after a failed attempt, which creates a back-off effect.

}

\rr{
\subsection{\enqop-\deqop on a Queue}

\pp{\begin{wrapfigure}{l}{.5\textwidth}
\begin{center}
\fuckspaa\includegraphics[width=.4\textwidth]{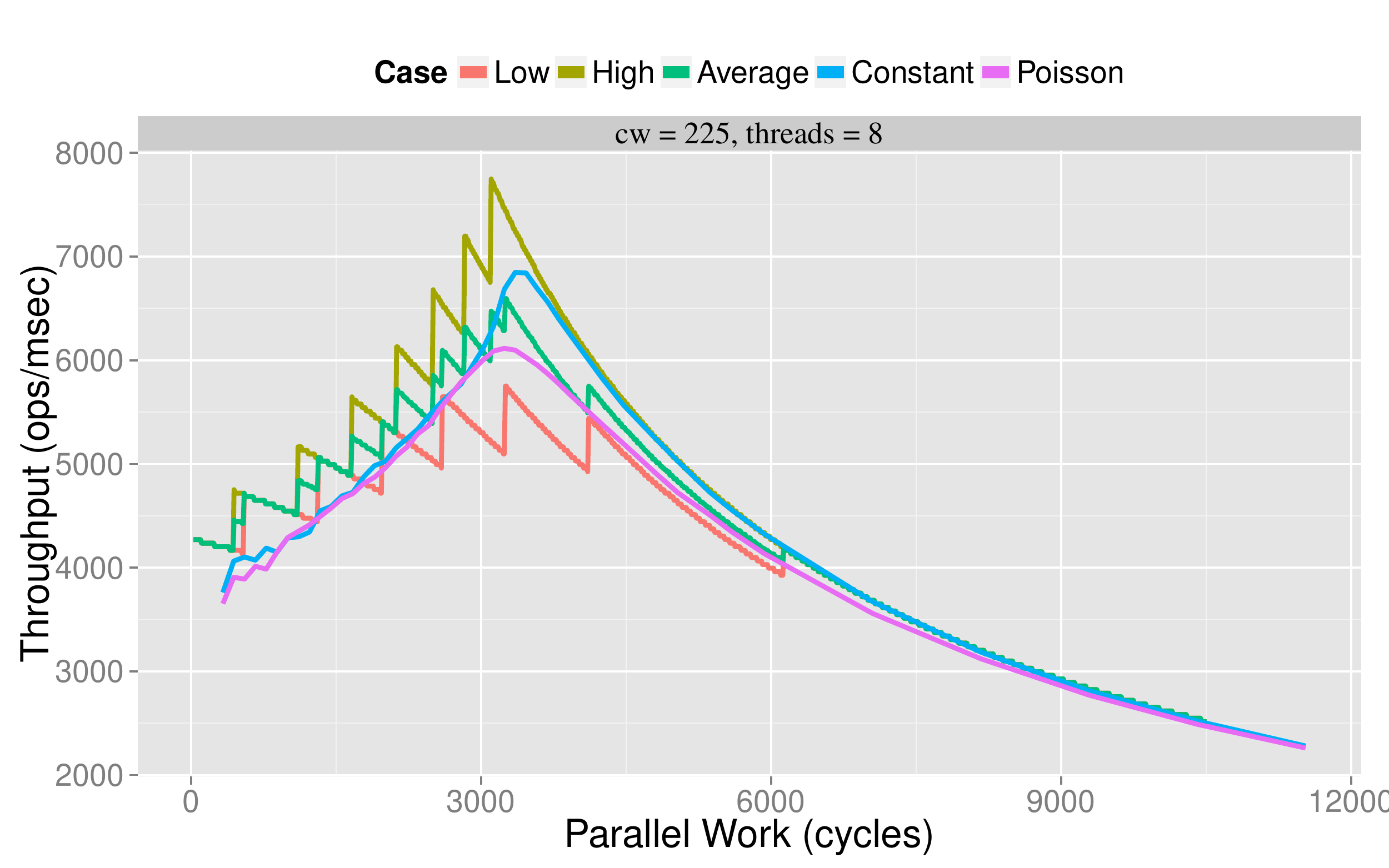}\end{center}
\caption{\enqop-\deqop on Michael and Scott queue\label{fig:ms}}
\end{wrapfigure}}

\rr{\begin{figure}[h!]
\begin{center}
\begin{minipage}{.47\textwidth}\begin{center}
\includegraphics[width=\textwidth]{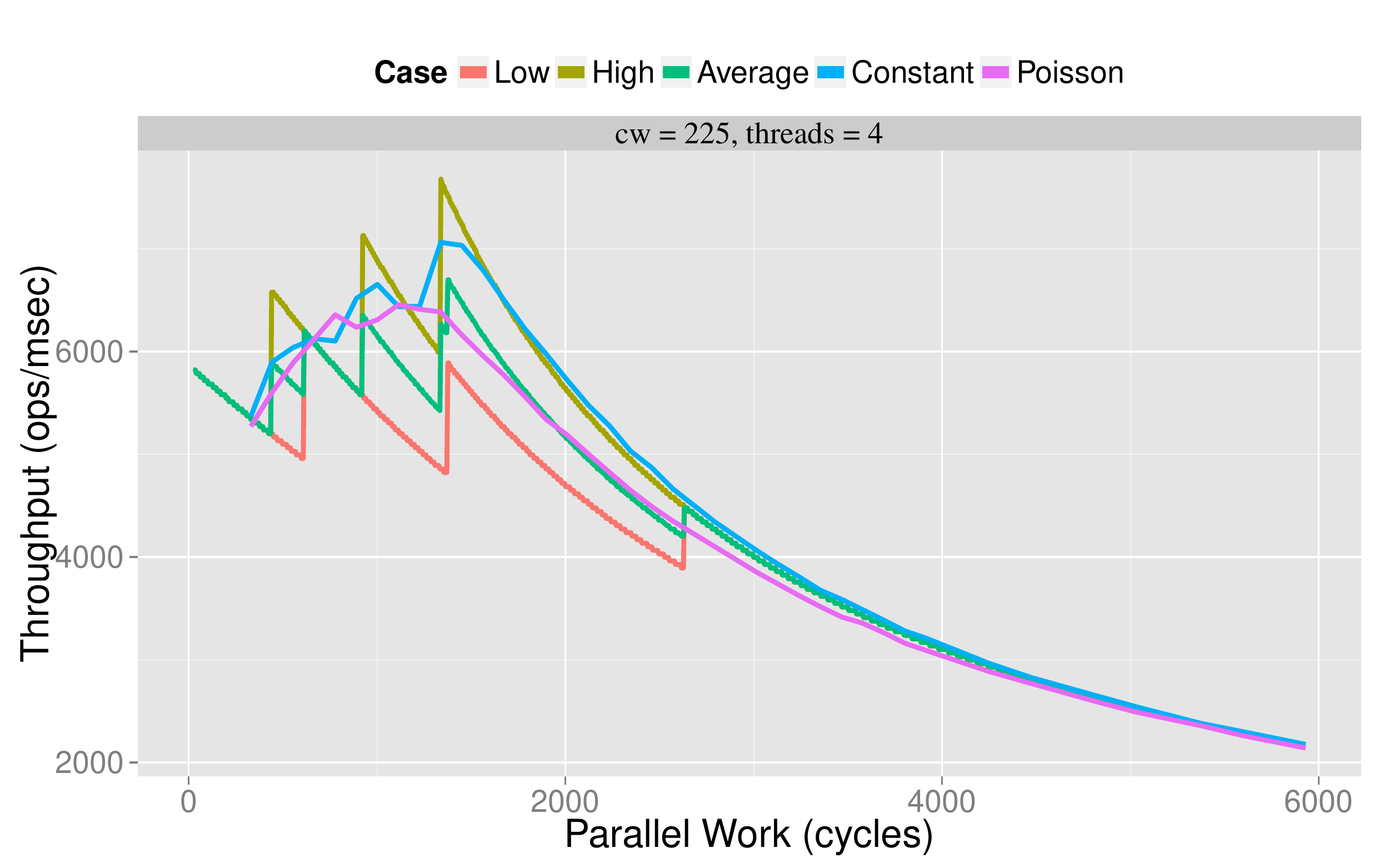}

(a) $4$ threads
\end{center}\end{minipage}\hfill\begin{minipage}{.47\textwidth}\begin{center}
\includegraphics[width=\textwidth]{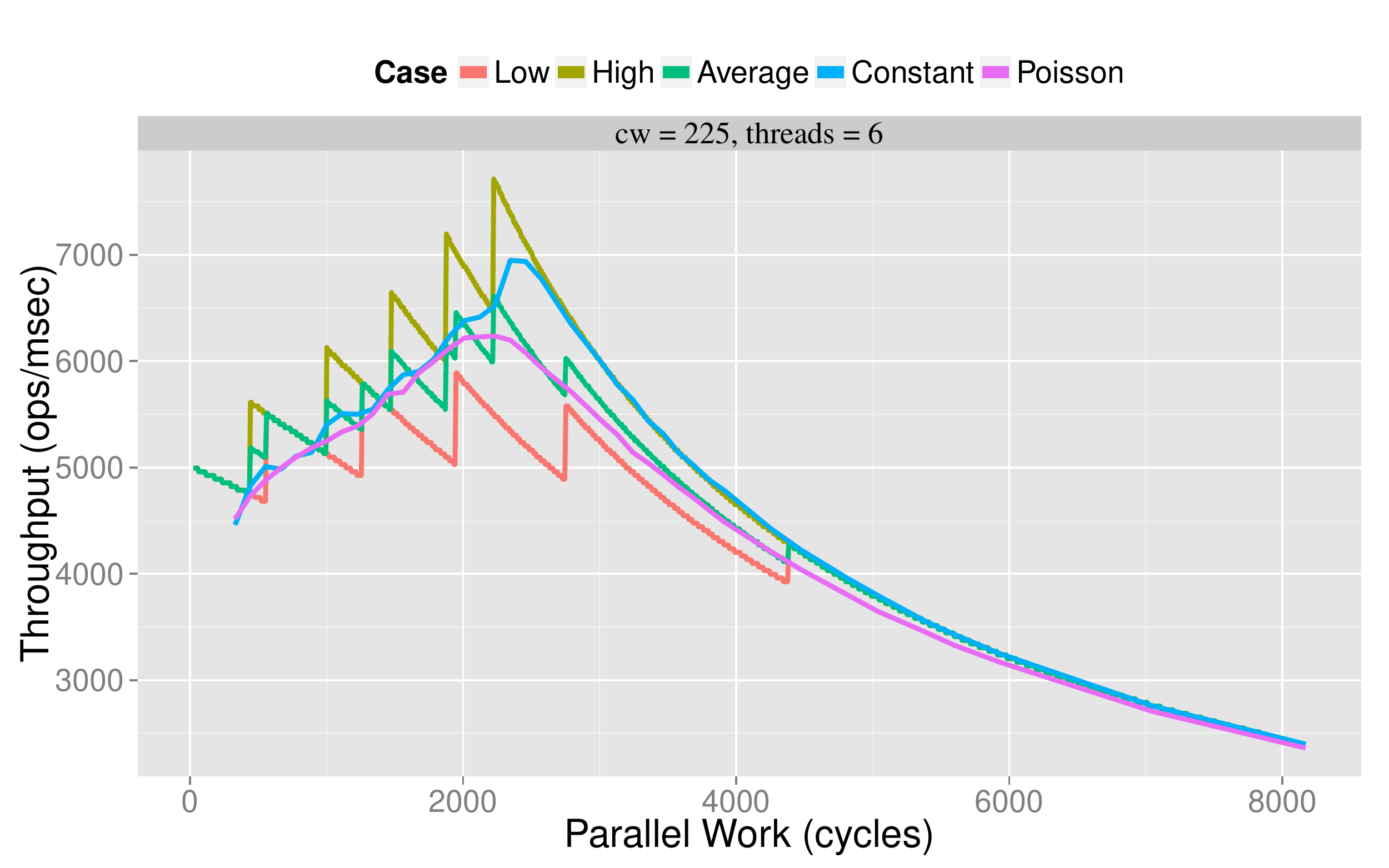}

(b) $6$ threads
\end{center}\end{minipage}

\includegraphics[width=.47\textwidth]{MS_8}

(c) $8$ threads
\end{center}
\caption{\enqop-\deqop on Michael and Scott queues\label{fig:ms}}
\end{figure}}

In order to demonstrate the validity of the model with several
\rls\pp{ (see~\cite{our-long})}, and that the results covers a wider
spectrum of application and designs from the ones we focused in our
model, we studied the following setting: the threads share a queue,
and each thread enqueues an element, executes the parallel section,
dequeues an element, and reiterates.
We consider the queue implementation by Michael and
Scott~\cite{lf-queue-michael}, that is usually viewed as the reference
queue while looking at lock-free queue implementations.

\deqop operations fit immediately into our model but \enqop operations need an adjustment due to
the helping mechanism. Note that without this helping mechanism, a simple queue
implementation would fit directly, but we also want to show that the model is
malleable, \ie the fundamental behavior remains unchanged even if we divert
slightly from the initial assumptions. We consider an equivalent execution that
catches up with the model, and use it to approximate the performance of the
actual execution of \enqop.\pp{ Due to lack of space, the full process is explained
in~\cite{our-long}.}

\rr{
\bigskip
\enqop is composed of two steps. Firstly, the new node is attached to the last
node of the queue via a \cas, that we denote by \casop{A}, leading to a
transient state. Secondly, the tail is redirected to point to the new node via
another \cas, that we denote by \casop{B}, which brings back the queue into a
steady state.

A new \enqop can not proceed before the two steps of previous success are
completed. The first step is the linearization point of operation and the second
step could be conducted by a different thread through the helping mechanism. In
order to start a new \enqop, concurrent \enqop{}s help the completion of the
second step of the last success if they find the queue in the transient
state. Alternatively, they try to attach their node to the queue if the queue is
in the steady state at the instant of check. This process continues until they
manage to attach their node to the queue via a retry loop in which state is checked
and corresponding \casop{} is executed.

The flow of an \enqop is determined by this state checks. Thus, an \enqop could
execute multiple \casop{B} (successful or failing) and multiple \casop{A}
(failing) in an interleaved manner, before succeeding in
\casop{A} at the end of the last \re. If we assume that both states are
equally probable for a check instant which will then end up with a retry, the
number of \casop{}s that ends up with a retry are expected to be
distributed equally among \casop{A} and \casop{B} for each thread. In addition,
each thread has a successful \casop{A} (which linearizes the \enqop)
and a \casop{B} at the end of the operation which could either be successful or
failed by a concurrent helper thread.

We imitate such an execution with an equivalent execution in which threads keep
the same relative ordering of the invocation, return from \enqop
together with same result. In equivalent execution, threads alternate between
\casop{A} and \casop{B} in their \res, and both steps of successful operation is
conducted by the same thread. The equivalent execution can be obtained by
thread-wise reordering of \casop{}s that leads to a \re and exchanging successful
\casop{B}s with the failed counterparts at the end of an \enqop, as the latter ones
indeed fail because of this success of helper threads. The model can be applied to
this equivalent execution by attributing each \casop{A}-\casop{B} couple to a single
iteration and represent it as a larger retry loop since the successful couple
can not overlap with another successful one and all overlapping ones fail. With
a straightforward extension of the expansion formula, we accomodate the \casop{A} in
the critical work which can also expand, and use \casop{B} as the \cas of our model.
\bigskip}

In addition, we take one step further outside the analysis by including a new
case, where the \ps follows a Poisson distribution, instead of being
constant. \rr{\pw is chosen as the mean to generate Poisson distribution instead of
taking it constant. }The results are illustrated in Figure~\ref{fig:ms}. Our
model provides good estimates for the constant \pw and also reasonable results
for the Poisson distribution case, although this case deviates from (/extends) our
model assumptions. The advantage of regularity, which brings synchronization to
threads, can be observed when the constant and Poisson distributions are compared. In the Poisson
distribution, the threads start to fail with larger \pw, which smoothes the curve
around the peak of the throughput curve.

}

\subsection{Discussion}
\label{sec:disc}

In this subsection we discuss the adequacy of our model, specifically
the cyclic argument, to capture the behavior that we observe in practice.
Figure~\ref{fig:failFreq} illustrates the frequency of occurrence of a given number of consecutive 
fails, together with average fails per success values and the throughput 
values, normalized by a constant factor so that they can be seen on the graph. In the background, the
frequency of occurrence of a given number of consecutive fails before success is presented. As a remark, the frequency of 6+ fails is gathered with
6. We expect to see a frequency distribution concentrated around the
average fails per success value, within the bounds computed by our
model.

\rr{\begin{figure}
\begin{center}
\fuckspaa\includegraphics[width=\textwidth]{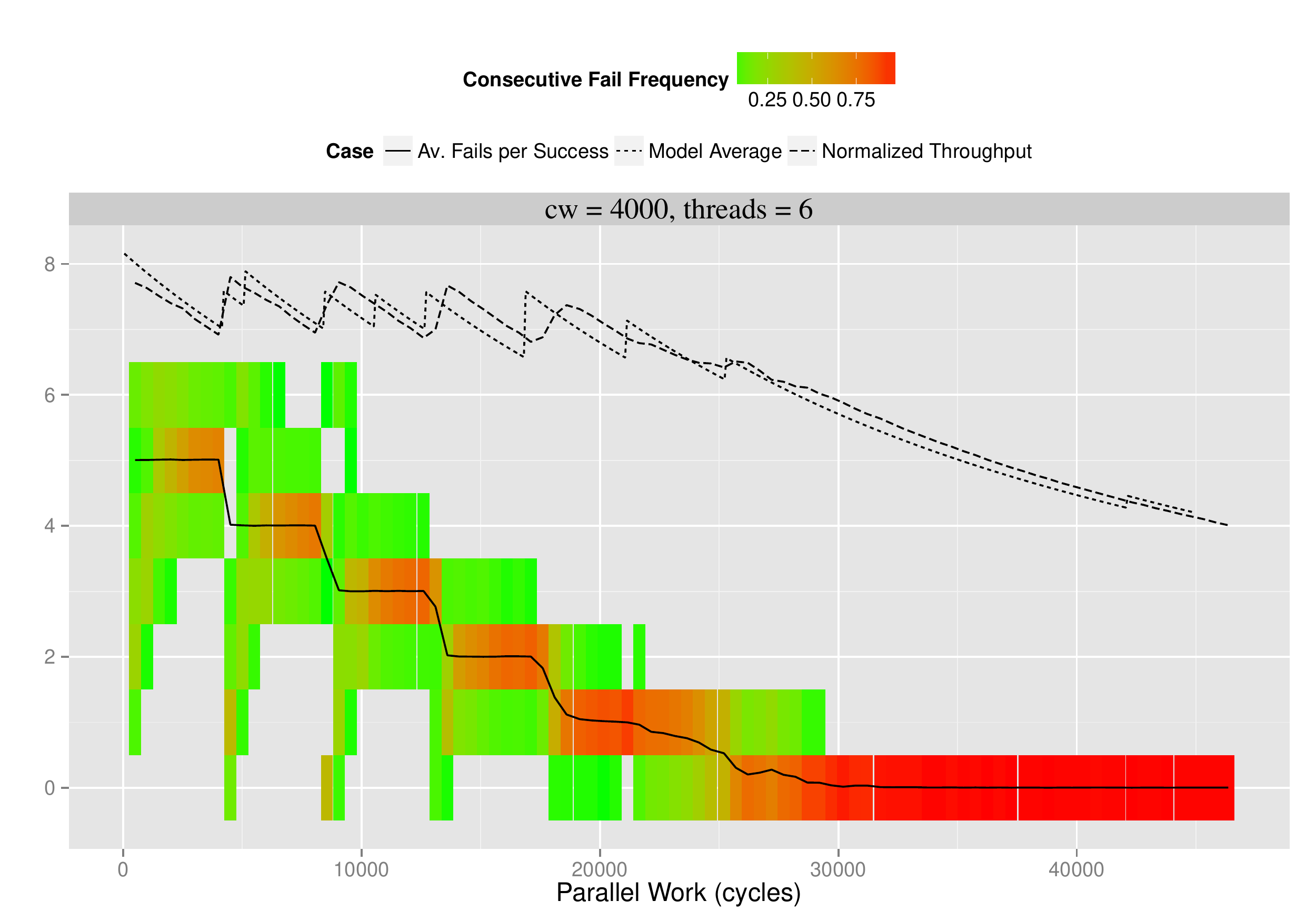}
\end{center}
\caption{Consecutive Fails Frequency\label{fig:failFreq}}
\end{figure}}

\pp{
\setlength{\intextsep}{0.1\baselineskip plus 0.5\baselineskip minus 0.5\baselineskip}
\begin{figure}[b!]
\begin{center}
\begin{minipage}[t]{.48\textwidth}
\includegraphics[width=\textwidth]{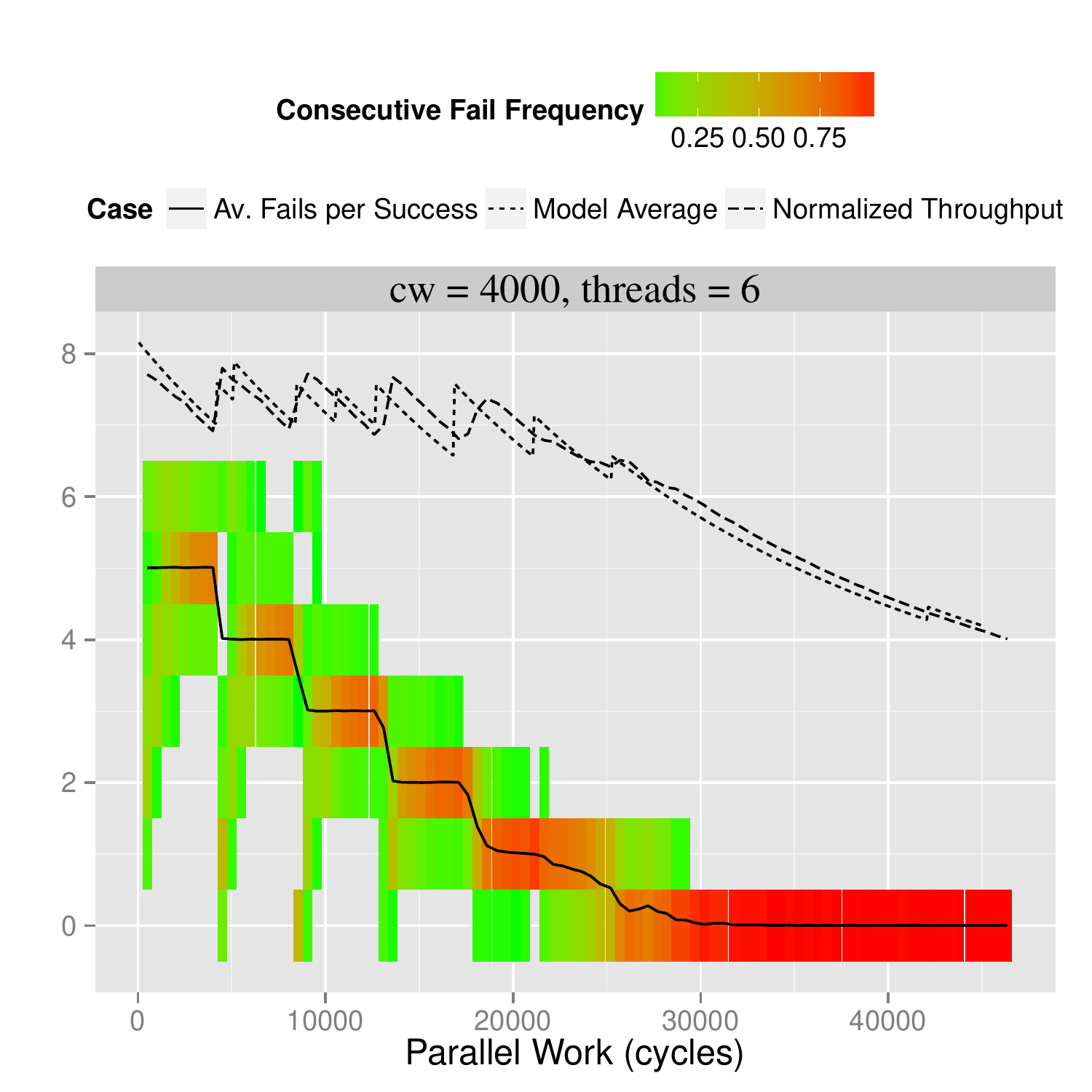}
\caption{Consecutive Fails Frequency\label{fig:failFreq}}
\end{minipage}\hfill\begin{minipage}[t]{.48\textwidth}
\includegraphics[width=\textwidth]{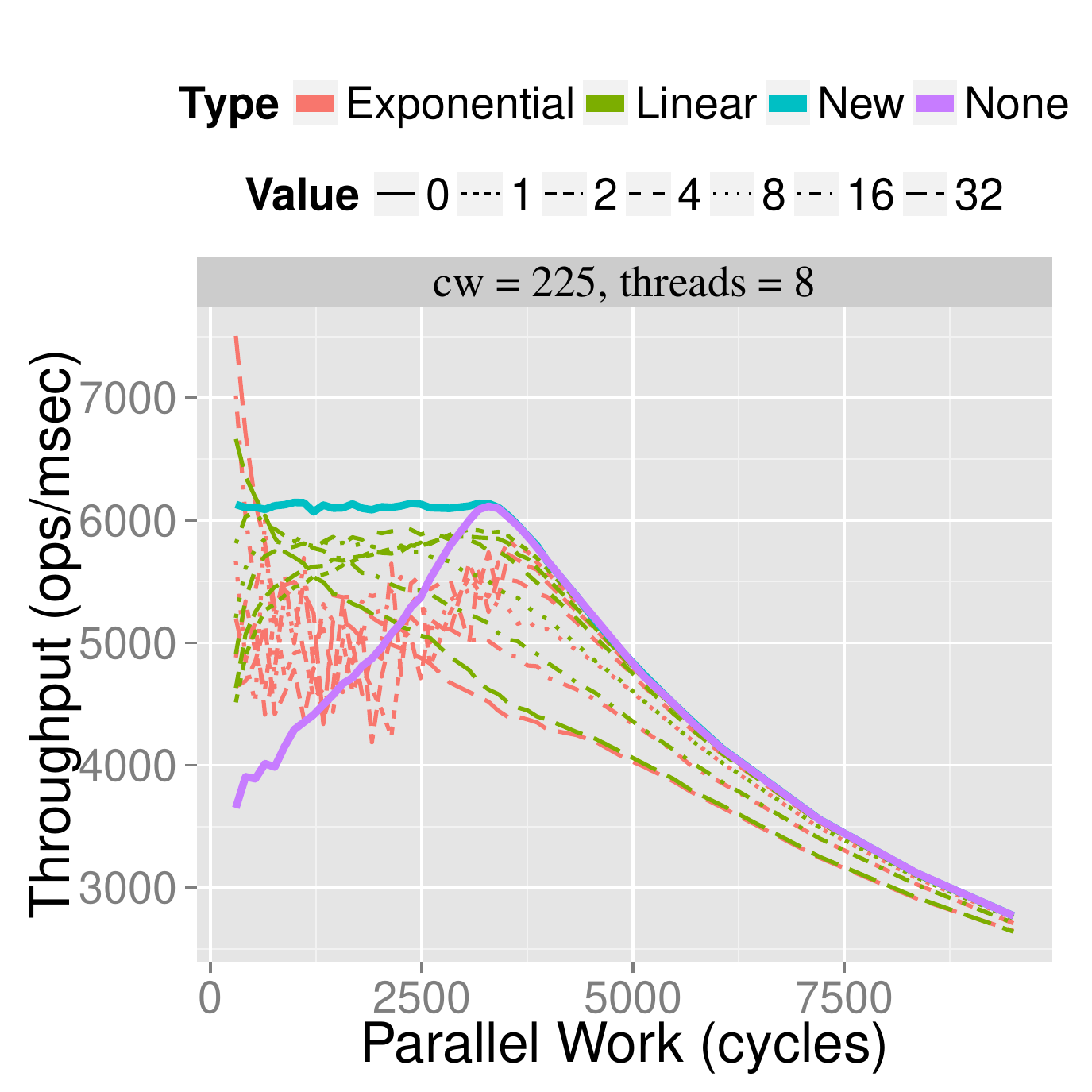}
\caption{Comparison of back-off schemes}
\label{fig:bo}
\end{minipage}\end{center}
\end{figure}}

While comparing the distribution of failures with the throughput, we
could conjecture that the bumps come from the fact that the failures
spread out. However, our model captures correctly the throughput
variations and thus strips down the right impacting factor.
The spread of the distribution of failures indicates the violation of
a stable cyclic execution (that takes place in our model), but in
these regions, $r$ actually gets close to $0$, as well as the minimum
of all gaps. The scattering in failures shows that, during the
execution, a thread is overtaken by another one. Still, as gaps are
close to $0$, the imaginary execution, in which we switch the two
thread IDs, would create almost the same performance effect. This
reasoning is strengthened by the fact that the actual average number
of failures follows the step behavior, predicted by our model.
This shows that even when the real execution is not cyclic and the
distribution of failures is not concentrated, our model that results
in a cyclic execution remains a close approximation of the actual
execution.

\subsection{Back-Off Tuning}
\label{sec:bo}

\rr{
\begin{figure}[h!]
\begin{center}
\begin{minipage}{.47\textwidth}\begin{center}
\includegraphics[width=\textwidth]{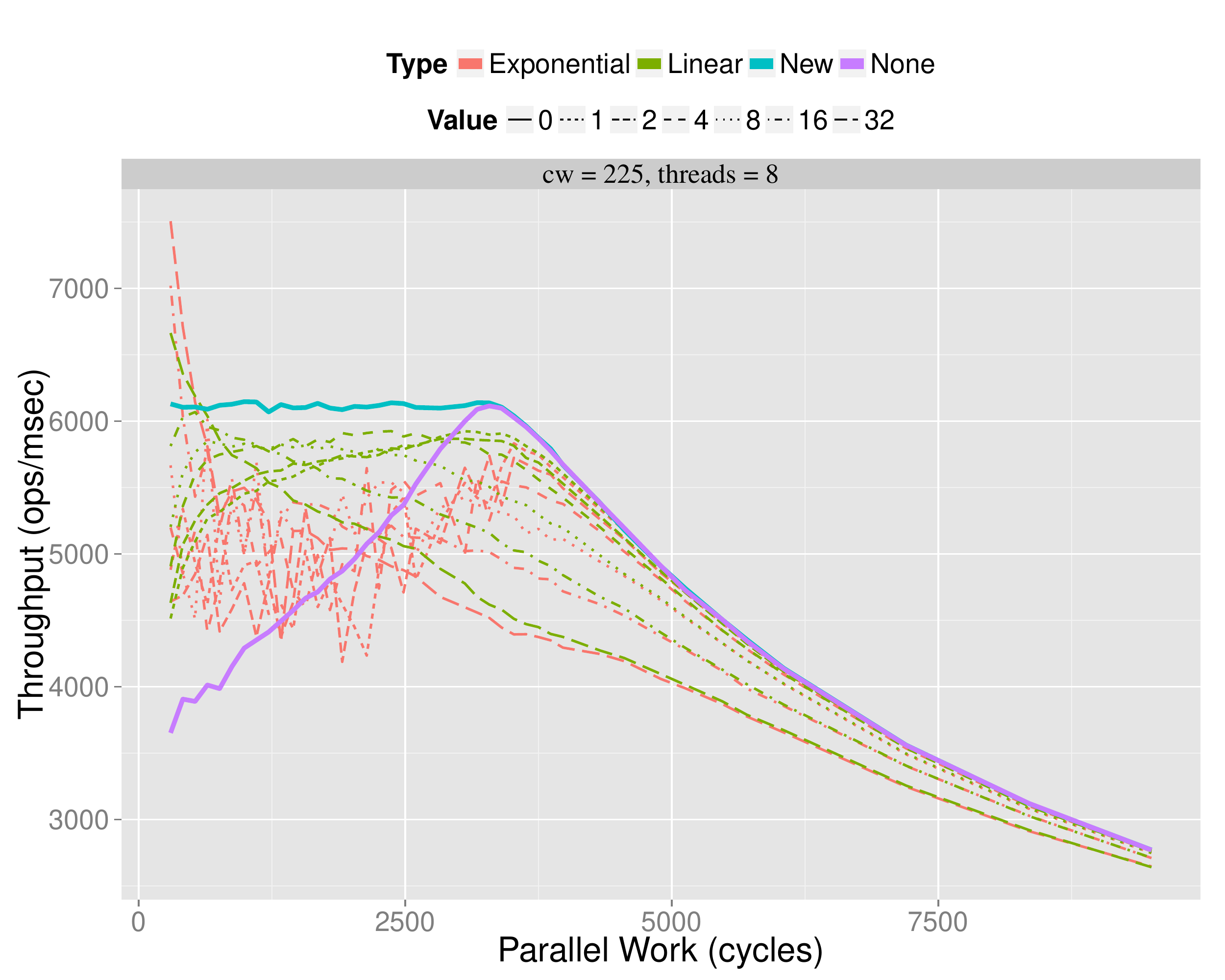}

(a) $8$ threads
\end{center}\end{minipage}\hfill\begin{minipage}{.47\textwidth}\begin{center}
\includegraphics[width=\textwidth]{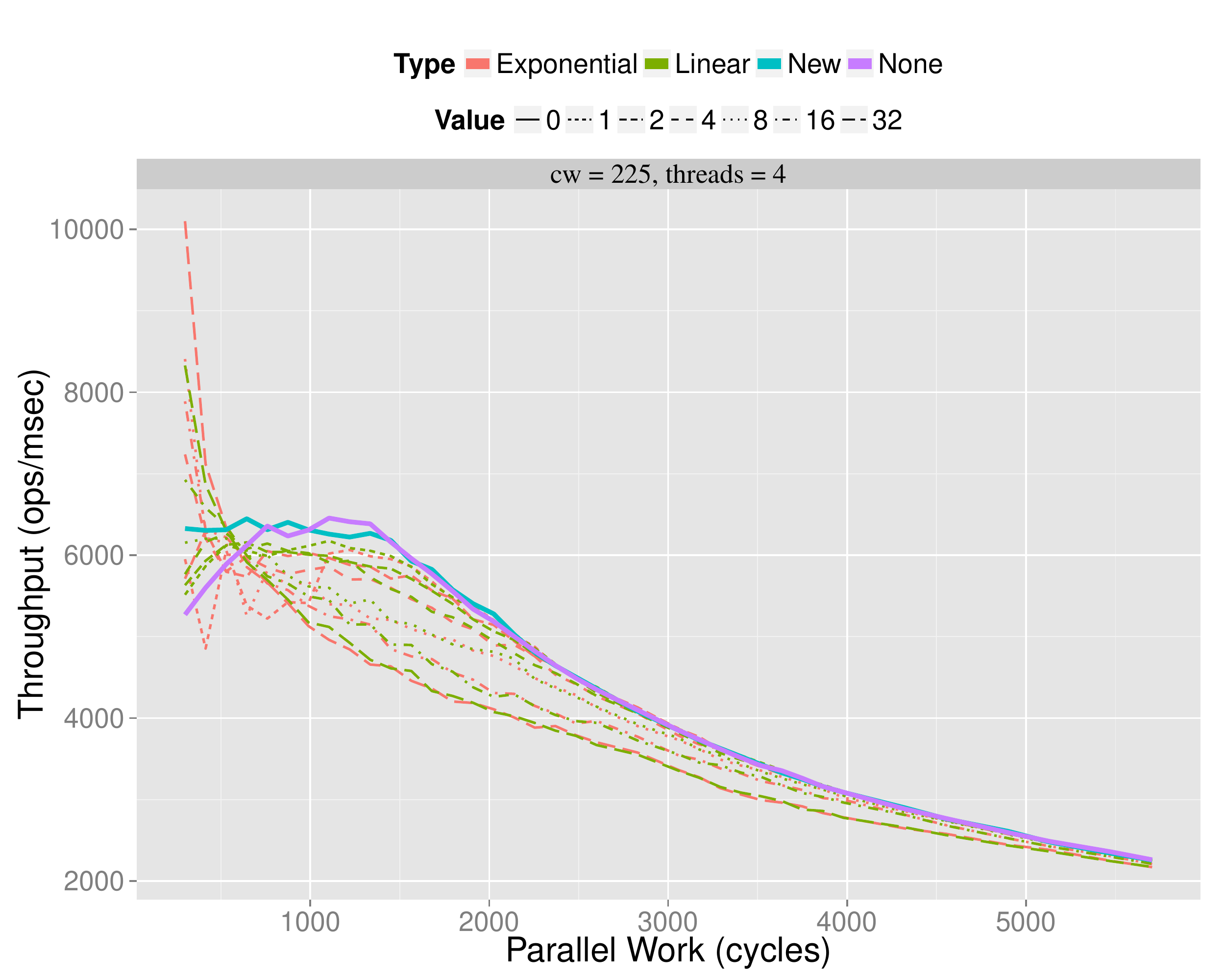}

(b) $4$ threads
\end{center}\end{minipage}
\end{center}
\caption{Comparison of back-off schemes for Poisson Distribution\label{fig:bo-poi}}
\end{figure}
}

Together with the analysis comes a natural back-off strategy: we estimate the
\pw corresponding to the peak point of the average curve, and when the \ps is smaller than the
corresponding \pw, we add a back-off in the \ps, so that the new \ps is at the
peak point.

We have applied exponential, linear and our back-off strategy to the
\enqop/\deqop experiment specified \pp{in~\cite{our-long} (sequence of 
\enqop and \deqop interleaved with \pss)}\rr{above}. Our back-off estimate
provides good results for both types of distribution. In
Figure\pp{~\ref{fig:bo}}\rr{~\ref{fig:bo-poi}} (where the values of
back-off are steps of $115$ cycles), the comparison is plotted for the
Poisson distribution, which is likely to be the worst for our
back-off. Our back-off strategy is better than the other, except for
very small \pss, but other back-off strategies should be tuned for
each value of \pw.

\rr{
We obtained the same shapes while removing the distribution law and considering constant values.
The results are illustrated in Figure~\ref{fig:bo-con}.

\begin{figure}[h!]
\begin{center}
\begin{minipage}{.47\textwidth}\begin{center}
\includegraphics[width=\textwidth]{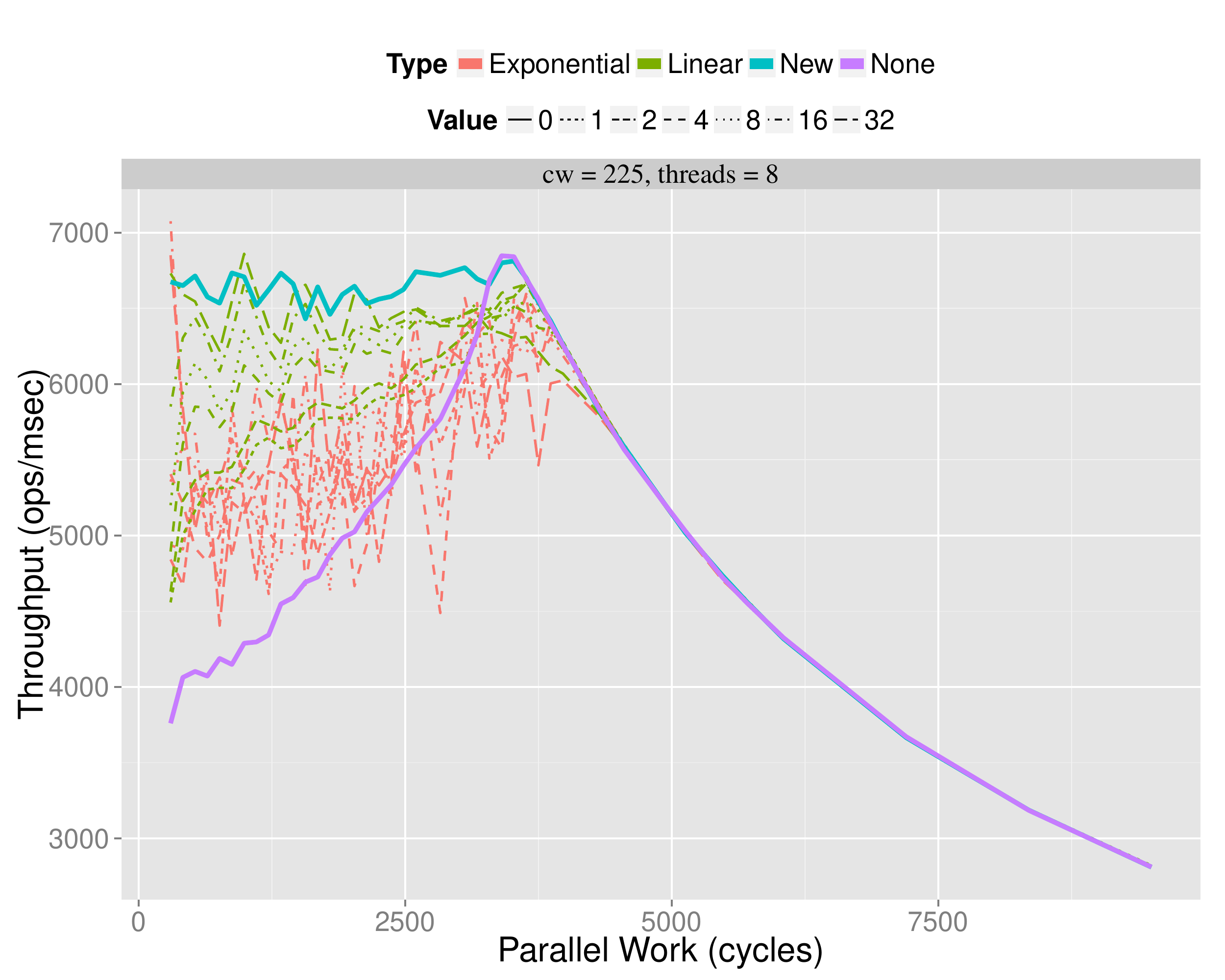}

(a) $8$ threads
\end{center}\end{minipage}\hfill\begin{minipage}{.47\textwidth}\begin{center}
\includegraphics[width=\textwidth]{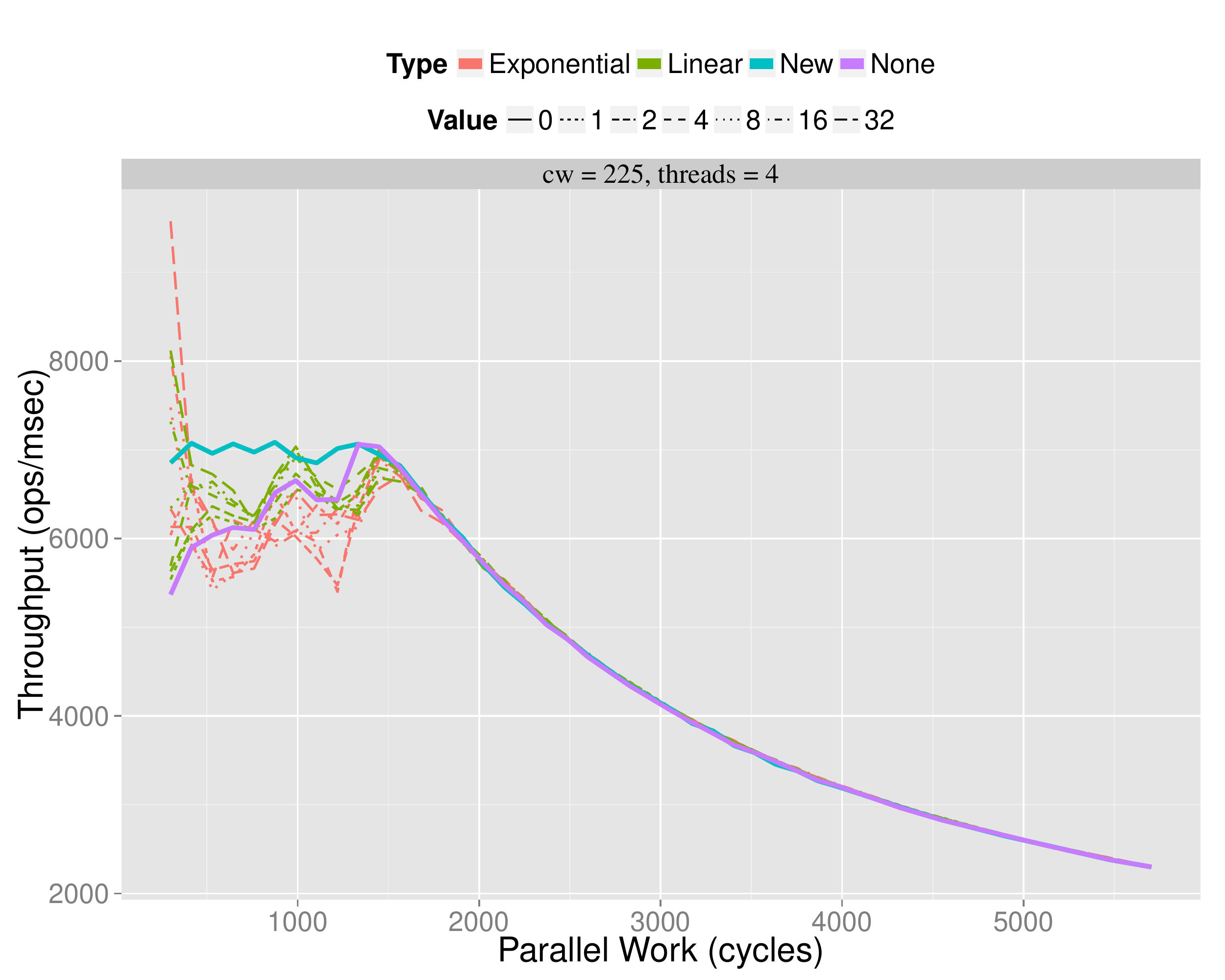}

(b) $4$ threads
\end{center}\end{minipage}
\end{center}
\caption{Comparison of back-off schemes for constant \pw \label{fig:bo-con}}
\end{figure}

}

\section{Conclusion}

In this paper, we have modeled and analyzed the performance of a general class of lock-free
algorithms, and have so been able to predict the throughput
of such algorithms, on actual executions. The analysis rely on the estimation of
two impacting factors that lower the throughput: on the one hand, the expansion,
due to the serialization of the atomic primitives that take place in the \rls; on
the other hand, the \cacas, due to a non-optimal synchronization between the
running threads. We have derived methods to calculate those parameters, along with
the final throughput estimate, that is calculated from a combination of these two
previous parameters. As a side result of our work,
this accurate prediction enables the design of a back-off
technique that performs better than other well-known techniques, namely linear
and exponential back-offs.

As a future work, we envision to enlarge the domain of validity of the
model, in order to cope with \dss whose operations do not have
constant \rl, as well as the framework, so that it includes more various
access patterns.
The fact that our results extend outside the model allows us to be optimistic
on the identification of the right impacting factors. Finally, we also foresee studying back-off techniques
that would combine a back-off in the \ps (for lower contention) and in the \rls (for
higher robustness).

\pp{\bibliographystyle{splncs03}}
\rr{\bibliographystyle{alpha}}
\rr{\bibliography{bibliography/longhead,bibliography/biblio}}
\pp{\bibliography{bibliography/longhead,bibliography/biblio}}

\end{document}

%% file: auto-plot/ex-lemma-1.tex
\node [left] at (0,\yt[1]) {\thr{0}};
\node [left] at (0,\yt[2]) {\thr{1}};
\node [left] at (0,\yt[3]) {\thr{2}};
\pgfmathparse{2.000000*\cle}
\edef\ple{\pgfmathresult}
\pgfmathparse{0.000000*\cle}
\extth{1}{\pgfmathresult}{ini}
\extth{1}{\cle}{csu}
\extth{1}{\ple}{par}
\extth{1}{\cle}{csu}
\extth{1}{\ple}{par}
\extth{1}{\cle}{cfa}
\extth{1}{\cle}{csu}
\extth{1}{\ple}{par}
\extth{1}{\cle}{cfa}
\extth{1}{\cle}{csu}
\extth{1}{\ple}{par}
\extth{1}{\cle}{cfa}
\extth{1}{\cle}{csu}
\extth{1}{\ple}{par}
\extth{1}{\cle}{cfa}
\pgfmathparse{1.500000*\cle}
\extth{2}{\pgfmathresult}{ini}
\extth{2}{\cle}{csu}
\extth{2}{\ple}{par}
\extth{2}{\cle}{csu}
\extth{2}{\ple}{par}
\extth{2}{\cle}{cfa}
\extth{2}{\cle}{csu}
\extth{2}{\ple}{par}
\extth{2}{\cle}{cfa}
\extth{2}{\cle}{csu}
\extth{2}{\ple}{par}
\extth{2}{\cle}{cfa}
\extth{2}{\cle}{csu}
\extth{2}{\ple}{par}
\pgfmathparse{5.750000*\cle}
\extth{3}{\pgfmathresult}{ini}
\extth{3}{\cle}{csu}
\extth{3}{\ple}{par}
\extth{3}{\cle}{cfa}
\extth{3}{\cle}{csu}
\extth{3}{\ple}{par}
\extth{3}{\cle}{cfa}
\extth{3}{\cle}{csu}
\extth{3}{\ple}{par}
\extth{3}{\cle}{cfa}
\extth{3}{\cle}{csu}
\extth{3}{\ple}{par}

%% file: auto-plot/ex-lemma-3.tex
\node [left] at (0,\yt[1]) {\thr{0}};
\node [left] at (0,\yt[2]) {\thr{1}};
\node [left] at (0,\yt[3]) {\thr{2}};
\node [left] at (0,\yt[4]) {\thr{3}};
\pgfmathparse{3.125000*\cle}
\edef\ple{\pgfmathresult}
\pgfmathparse{0.000000*\cle}
\extth{1}{\pgfmathresult}{ini}
\extth{1}{\cle}{csu}
\extth{1}{\ple}{par}
\extth{1}{\cle}{cfa}
\extth{1}{\cle}{csu}
\extth{1}{\ple}{par}
\extth{1}{\cle}{cfa}
\extth{1}{\cle}{cfa}
\extth{1}{\cle}{csu}
\extth{1}{\ple}{par}
\extth{1}{\cle}{cfa}
\extth{1}{\cle}{csu}
\extth{1}{\ple}{par}
\extth{1}{\cle}{cfa}
\extth{1}{\cle}{csu}
\extth{1}{\ple}{par}
\extth{1}{\cle}{cfa}
\pgfmathparse{0.500000*\cle}
\extth{2}{\pgfmathresult}{ini}
\extth{2}{\cle}{cfa}
\extth{2}{\cle}{csu}
\extth{2}{\ple}{par}
\extth{2}{\cle}{cfa}
\extth{2}{\cle}{csu}
\extth{2}{\ple}{par}
\extth{2}{\cle}{cfa}
\extth{2}{\cle}{cfa}
\extth{2}{\cle}{cfa}
\extth{2}{\cle}{csu}
\extth{2}{\ple}{par}
\extth{2}{\cle}{cfa}
\extth{2}{\cle}{csu}
\extth{2}{\ple}{par}
\extth{2}{\cle}{cfa}
\extth{2}{\cle}{csu}
\extth{2}{\ple}{par}
\pgfmathparse{2.250000*\cle}
\extth{3}{\pgfmathresult}{ini}
\extth{3}{\cle}{cfa}
\extth{3}{\cle}{csu}
\extth{3}{\ple}{par}
\extth{3}{\cle}{cfa}
\extth{3}{\cle}{csu}
\extth{3}{\ple}{par}
\extth{3}{\cle}{csu}
\extth{3}{\ple}{par}
\extth{3}{\cle}{cfa}
\extth{3}{\cle}{csu}
\extth{3}{\ple}{par}
\extth{3}{\cle}{cfa}
\extth{3}{\cle}{csu}
\extth{3}{\ple}{par}
\pgfmathparse{9.875000*\cle}
\extth{4}{\pgfmathresult}{ini}
\extth{4}{\cle}{csu}
\extth{4}{\ple}{par}
\extth{4}{\cle}{cfa}
\extth{4}{\cle}{csu}
\extth{4}{\ple}{par}
\extth{4}{\cle}{cfa}
\extth{4}{\cle}{csu}
\extth{4}{\ple}{par}
\extth{4}{\cle}{cfa}
\extth{4}{\cle}{csu}
\extth{4}{\ple}{par}

%% file: auto-plot/ex-lemma-4.tex
\node [left] at (0,\yt[1]) {\thr{0}};
\node [left] at (0,\yt[2]) {\thr{1}};
\node [left] at (0,\yt[3]) {\thr{2}};
\node [left] at (0,\yt[4]) {\thr{3}};
\pgfmathparse{2.500000*\cle}
\edef\ple{\pgfmathresult}
\pgfmathparse{0.000000*\cle}
\extth{1}{\pgfmathresult}{ini}
\extth{1}{\cle}{csu}
\extth{1}{\ple}{par}
\extth{1}{\cle}{cfa}
\extth{1}{\cle}{csu}
\extth{1}{\ple}{par}
\extth{1}{\cle}{cfa}
\extth{1}{\cle}{csu}
\extth{1}{\ple}{par}
\extth{1}{\cle}{csu}
\extth{1}{\ple}{par}
\extth{1}{\cle}{cfa}
\extth{1}{\cle}{csu}
\extth{1}{\ple}{par}
\extth{1}{\cle}{cfa}
\extth{1}{\cle}{csu}
\extth{1}{\ple}{par}
\extth{1}{\cle}{cfa}
\extth{1}{\cle}{csu}
\extth{1}{\ple}{par}
\pgfmathparse{0.125000*\cle}
\extth{2}{\pgfmathresult}{ini}
\extth{2}{\cle}{cfa}
\extth{2}{\cle}{csu}
\extth{2}{\ple}{par}
\extth{2}{\cle}{cfa}
\extth{2}{\cle}{csu}
\extth{2}{\ple}{par}
\extth{2}{\cle}{cfa}
\extth{2}{\cle}{csu}
\extth{2}{\ple}{par}
\extth{2}{\cle}{csu}
\extth{2}{\ple}{par}
\extth{2}{\cle}{cfa}
\extth{2}{\cle}{csu}
\extth{2}{\ple}{par}
\extth{2}{\cle}{cfa}
\extth{2}{\cle}{csu}
\extth{2}{\ple}{par}
\extth{2}{\cle}{cfa}
\pgfmathparse{1.875000*\cle}
\extth{3}{\pgfmathresult}{ini}
\extth{3}{\cle}{cfa}
\extth{3}{\cle}{csu}
\extth{3}{\ple}{par}
\extth{3}{\cle}{cfa}
\extth{3}{\cle}{csu}
\extth{3}{\ple}{par}
\extth{3}{\cle}{cfa}
\extth{3}{\cle}{cfa}
\extth{3}{\cle}{cfa}
\extth{3}{\cle}{cfa}
\extth{3}{\cle}{cfa}
\extth{3}{\cle}{csu}
\extth{3}{\ple}{par}
\extth{3}{\cle}{cfa}
\extth{3}{\cle}{csu}
\extth{3}{\ple}{par}
\extth{3}{\cle}{cfa}
\extth{3}{\cle}{csu}
\extth{3}{\ple}{par}
\pgfmathparse{11.250000*\cle}
\extth{4}{\pgfmathresult}{ini}
\extth{4}{\cle}{csu}
\extth{4}{\ple}{par}
\extth{4}{\cle}{csu}
\extth{4}{\ple}{par}
\extth{4}{\cle}{cfa}
\extth{4}{\cle}{csu}
\extth{4}{\ple}{par}
\extth{4}{\cle}{cfa}
\extth{4}{\cle}{csu}
\extth{4}{\ple}{par}